\documentclass[journal]{IEEEtran}
\ifCLASSOPTIONcompsoc
  \usepackage[nocompress]{cite}
\else
  \usepackage{cite}
\fi

\usepackage{hyperref}
\ifCLASSINFOpdf
\else
\fi
\usepackage{amsmath,amssymb,amsfonts,amsthm}
\usepackage{algorithmic}
\usepackage{graphicx}
\usepackage{textcomp}
\usepackage{xcolor}
\usepackage{xcolor}
\graphicspath{{wienerfigures/}}
\usepackage[margin=1in]{geometry}

\newtheorem{claim}{Claim}
\newtheorem{theorem}{Theorem}
\newtheorem{corollary}{Corollary}

\newtheorem{lemma}{Lemma}
\newtheorem{assu}{Assumption}
\newtheorem{pb}{Problem}

\usepackage{cuted}
\usepackage{cancel}
\usepackage{lipsum}
\allowdisplaybreaks
\begin{document}
%
\title{Sampling of the Wiener Process for Remote Estimation over a Channel with Unknown Delay Statistics}
%
%
%
%

\author{Haoyue~Tang,~\IEEEmembership{Member,~IEEE,}
        Yin~Sun,~\IEEEmembership{Senior~Member,~IEEE}
        and~Leandros~Tassiulas,~\IEEEmembership{Fellow,~IEEE\&ACM}
\IEEEcompsocitemizethanks{\IEEEcompsocthanksitem H. Tang and L. Tassiulas are with the Department
of Electrical Engineering, Yale University, New Haven,
CT, 06511.\protect\\
E-mail: \{haoyue.tang, leandros.tassiulas\}@yale.edu
\IEEEcompsocthanksitem Y. Sun is with the Department of Electrical and Computer Engineering, Auburn University, Auburn, AL, 36849. E-mail: yzs0078@auburn.edu
\IEEEcompsocthanksitem  The work of H. Tang and L. Tassiulas was supported by the NSF CNS-2112562 AI Institute for Edge Computing Leveraging Next Generation Networks (Athena) and the ONR N00014-19-1-2566. The work of Y. Sun
was supported by the ARO grant
W911NF-21-1-024}}

%
%

\markboth{Submitted to IEEE/ACM Transactions on Networking}%
{Shell \MakeLowercase{\textit{et al.}}: Bare Demo of IEEEtran.cls for Computer Society Journals}
%



\IEEEtitleabstractindextext{%
\begin{abstract}
		In this paper, we study an online sampling problem of the Wiener process. The goal is to minimize the mean squared error (MSE) of the remote estimator under a sampling frequency constraint when the transmission delay distribution is unknown. The sampling problem is reformulated into an optional stopping problem, and we propose an online sampling algorithm that can adaptively learn the optimal stopping threshold through stochastic approximation. We prove that the cumulative MSE regret grows with rate $\mathcal{O}(\ln k)$, where $k$ is the number of samples. Through Le Cam's two point method, we show that the worst-case cumulative MSE regret of any online sampling algorithm is lower bounded by $\Omega(\ln k)$. Hence, the proposed online sampling algorithm is minimax order-optimal. Finally, we validate the performance of the proposed algorithm via numerical simulations. 
\end{abstract}

\begin{IEEEkeywords}
Age of Information, Online Learning, Stochastic Approximation
\end{IEEEkeywords}}

\maketitle

\IEEEdisplaynontitleabstractindextext

%
\IEEEpeerreviewmaketitle

\section{Introduction}\label{sec:introduction}

%
%
%
%
\IEEEPARstart{T}he omnipresence of the autonomous driving and the intelligent manufacturing systems involve tasks of sampling and remotely estimating fresh status information. For example, in autonomous driving systems, status information such as the position and the instant speed of cars keep changing, and the controller has to estimate the update-to-date status based on samples collected from the surrounding sensors. To ensure efficient control and system safety, it is important to estimate the fresh status information precisely under limited communication resources and random channel conditions. 

To measure the freshness of the status update information, the Age of Information (AoI) metric has been proposed in \cite{roy_12_aoi}. By definition, AoI captures the difference between the current time and the time-stamp at which the freshest information available at the destination was generated. It is revealed that the AoI minimum sampling and transmission strategies behave differently from utility maximization and delay minimization \cite{roy_15_isit}. Samples with fresher content should be delivered to the destination in a timely manner \cite{sun_17_tit}. 

When the evolution of the dynamic source can be modeled by a random signal process, the mean square estimation error (MSE) based on the available information at the receiver can be used to capture freshness. Sampling to minimize the MSE of the random process in different communication networks are studied in \cite{hajet_03_infocom,ornee_21_ton,nayyar_13_tac,sun_wiener,tsai_2021_ton,gauss_yin}. Considering that the dynamic source is a Wiener process, the optimum sampling policy that minimizes the estimation MSE is shown to have a threshold structure, i.e., a new sample should be taken once the difference between the actual signal value and the estimate based on past samples exceed a certain threshold. Such thresholds also holds for the Ornstein-Uhlenbeck process \cite{ornee_21_ton,ou} and the Gaussian Markov source \cite{gauss_yin}. The optimum sampling thresholds can be obtained by  the bi-section search \cite{sun_wiener} or iterative thresholding \cite{chichun-19-isit} if the delay distribution and the statistics of the channel are known in advance. 

When the statistics of the communication channel is unknown, the problem of sampling and transmissions for data freshness optimization can be formulated into a sequential decision making problem \cite{aoibandit,atay2020aging, banerjee_adversarial_aoi,tripathi2021online,li2021efficient}. By using the AoI as the freshness metric,  \cite{aoibandit,atay2020aging,banerjee_adversarial_aoi} design online link rate selection algorithms based on stochastic bandits. When the channels are time-varying and the transmitter has an average power constraint, \cite{ceran_19_infocomwks,ceran_21_jsac,kam_rl,aba_drl_aoi,aylin_rl} employ reinforcement learning algorithms to minimize the average AoI under unknown channel statistics. Notice that in applications such as the remote estimation, a linear AoI cannot fully capture the data freshness. To solve this problem, Tripathi \emph{et al. } model the information freshness to be a time-varying function of the AoI \cite{tripathi2021online}, and a robust online learning algorithm is proposed. The above research tackles with unknown packet loss rate or utility functions, the problem of designing online algorithms under unknown delay statistics are not well studied. The iterative thresholding algorithm proposed in  \cite{chichun-19-isit} can be applied in the online setting when the delay statistics is unknown, whereas the convergence rate and the optimality of the derived online algorithm are not well understood. 

In this paper, we consider an online sampling problem, where a sensor transmits status updates of the Wiener source to a destination through a channel with random delay. Our goal is to design a sampling policy that minimizes the estimation error when the delay distribution is unknown  a priori. The main contributions of this paper are as follows:
\begin{itemize}
	\item The design of the MSE minimum sampling policy is reformulated as an optimal stopping problem. 
	By analyzing the sufficient conditions of the optimum threshold, we propose an online sampling policy that learns the optimum stopping threshold adaptively through stochastic approximation. Compared with \cite{chichun-19-isit,Tang2205:Sending,tang2022age}, the operation of the proposed algorithm does not require prior knowledge of an upper bound of the optimum threshold. 
	\item We prove that the time averaged MSE of the proposed algorithm converges almost surely to the minimum MSE if the fourth order moment of the transmission delay is finite (Theorem \ref{thm:dep-converge}). In addition, it is shown that the MSE regret, i.e., the sub-optimality gap between the expected cumulative MSE of the proposed algorithm and the optimum policy with distribution knowledge, grows at a speed of $\mathcal{O}(\ln k)$, where $k$ is the number of samples (Corollary \ref{thm:mse-rate}). The perturbed ordinary differential equation (ODE) method is a popular tool for establishing the convergence rate of stochastic approximation algorithms \cite{Kushner2003}. However, this tool requires either the threshold being learned is in a bounded closed set, or the second moment of the updating directions are bounded. Because our algorithm does not require an upper bound on the optimum threshold, and the essential supremum of the transmission delay could be unbounded, we need to develop
	a new method for convergence rate analysis, which is based on the Lyapunov drift method for heavy traffic analysis. 
	\item Further by using the classic Le Cam's two point method, we show that for any causal algorithm that makes sampling decision based on historical information, under the worst case delay distribution, the MSE regret is lower bounded by $\Omega(\ln k)$  (Theorem \ref{thm:converse}). By combining Theorem \ref{thm:mse-rate} and Theorem \ref{thm:converse}, we obtain that the proposed online sampling algorithm achieves the minimax order-optimal regret.
	\item We validate the performance of the proposed algorithm via numerical simulations. In contrast to \cite{chichun-19-isit}, the proposed algorithm could meet an average sampling frequency constraint. 
\end{itemize}



\section{System Model and Problem Formulation}

\subsection{System Model}
As is depicted in Fig.~\ref{fig:model}, we revisit the status update system in \cite{sun_17_tit,arafa_model,sun_wiener}, where a sensor takes samples from a Wiener process and transmits the samples to a receiver through a network interface queue. The network interface serves the update packets on the First-Come-First-Serve (FCFS) basis. An ACK is sent back to the sensor once an update packet is cleared at the interface. We assume that the transmission duration after passing the network interface is negligible. 
\begin{figure}[h]
	\centering
	\includegraphics[width=.4\textwidth]{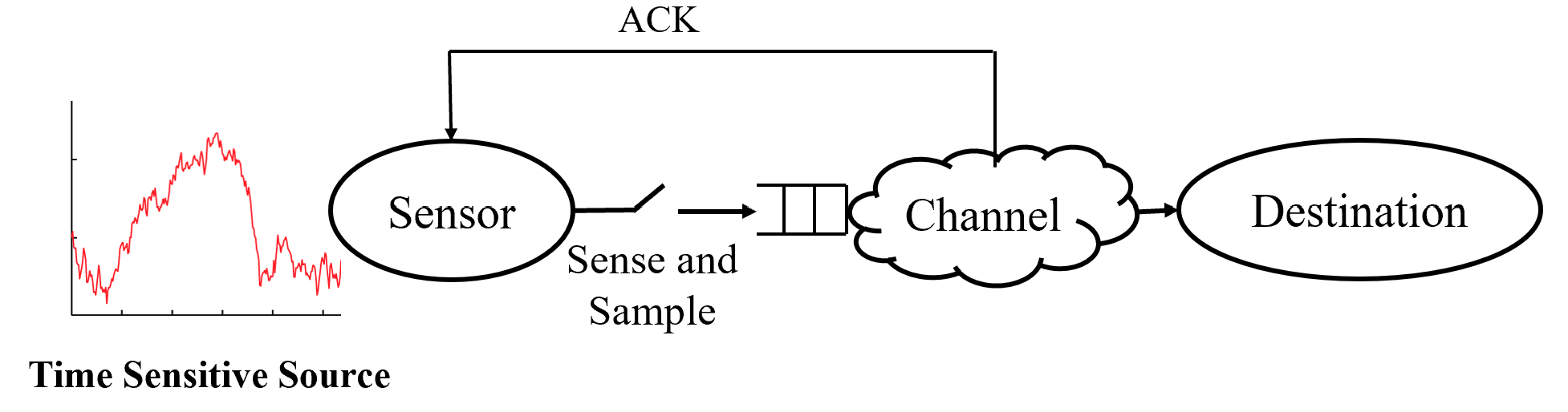}
	\caption{System model. }
	\label{fig:model}
\end{figure} 

Let $X_t\in\mathbb{R}$ denote the value of the Wiener process at time $t\in\mathbb{R}^+$. The sampling time-stamp of the $k$-th sample, denoted by $S_k$, is determined by the sensor at will. Based on the FCFS principle, the network interface will start serving the $k$-th packet after the $(k-1)$-th packet is cleared at the network interface and arrived at the receiver. We assume that the service time $D_k$ are independent and identically distributed (i.i.d) with a probability distribution $\mathbb{P}_D$, and $D_k$ is independent of the wiener process $X_t$. The reception time of the $k$-th packet, denoted by $R_k$ satisfies the following recursive formula: $R_k=\{S_k, R_{k-1}\}+D_k$ and we define $R_0=0$ for simplicity. We assume the average transmission delay $\overline{D}:=\mathbb{E}_{D\sim\mathbb{P}_D}[D]$ is lower bounded by $\overline{D}_{\text{lb}}>0$. 
\subsection{MMSE Estimation}

Let $i(t):=\max_{k\in\mathbb{N}}\{k|R_k\leq t\}$ be the index of the latest sample received by the destination at time $t$. The information available at the receiver at time $t$ can be summarized as follows: (i). The sampling time-stamps, transmission delay and the values of previous samples $\mathcal{M}_t:=\{(S_j, D_j, X_{S_j})\}_{j=1}^{i(t)}$; (ii). The fact that no packet was received during $(R_{i(t)}, t]$. Similar to \cite{sun_17_tit,est_ifac}, we assume that the receiver estimates $X_t$ only based on $\mathcal{M}_t$ and neglects the second part of information. The minimum mean-square error (MMSE) estimator \cite{poor2013introduction} in this case is:
\begin{equation}
	\hat{X}_t=\mathbb{E}[X_t|\mathcal{M}_t]=X_{S_{i(t)}}. \label{eq:MMSEest}
\end{equation}

We use a sequence of sampling time instants $\pi\triangleq\{S_k\}_{k=1}^{\infty}$ to represent a sampling policy. The expected time average mean square error (MSE) under $\pi$ is denoted by $\overline{\mathcal{E}}_\pi$, i.e., 
\begin{equation}
	\overline{\mathcal{E}}_\pi\triangleq\limsup_{T\rightarrow\infty}\mathbb{E}\left[\frac{1}{T}\int_{t=0}^T\left(X_t-X_{S_{i(t)}}\right)^2\mathsf{d}t\right].
\end{equation}
\subsection{Problem Formulation}

Our goal in this work is to design one sampling policy that can minimize the MSE for the estimator when the delay distribution $\mathbb{P}_D$ is unknown. Specifically, we focus on the set of causal policies denoted by $\Pi$, where each policy $\pi\in\Pi$ selects the sampling time $S_k$ of the $k$-th sample based on the transmission delay $\{D_{k'}\}_{k'< k}$ and Wiener process evolution $\{X_t\}_{t\leq S_k}$ from the past. The transmission delay and the evolution of the Wiener process in the future cannot be used to decide the sampling time. Due to the energy constraint, we require that the sampling frequency should below a certain threshold. The optimal sampling problem is organized as follows:
\begin{pb}[MMSE minimization]\label{pb:mse}
	\begin{subequations}
		\begin{align}
			\mathsf{mse}_{\mathsf{opt}}\triangleq&\inf\limits_{\pi\in\Pi}\mathop{\limsup}\limits_{T\rightarrow\infty}\mathbb{E}\left[\frac{1}{T}\int_{t=0}^T\left(\hat{X}_t-X_t\right)^2\mathrm{d}t\right],\label{eq:primalobj}\\
			&\hspace{0.2cm}\text{s.t.}\hspace{0.2cm}\mathop{\limsup}\limits_{T\rightarrow\infty}\mathbb{E}\left[\frac{i(T)}{T}\right]\leq f_{\mathsf{max}}.
		\end{align}
	\end{subequations}
\end{pb}	

\section{Problem Solution}\label{sec:dep}
In this section, 
the MSE minimization problem (i.e., Problem~\ref{pb:mse}) is reformulated into an optimal stopping problem. Let $\pi^\star$ be an optimum policy whose average MSE achieves $\mathsf{mse}_{\mathsf{opt}}$. Sufficient conditions for $\pi^\star$ are provided in Subsection~\ref{sec:dep-off}. The online sampling algorithm $\pi_{\mathsf{online}}$ is provided in Subsection~\ref{sec:dep-online} and Subsection~\ref{sec:dep-analysis} characterizes the behaviors of the online sampling policy. 

\begin{figure}[h]
	\centering
	\includegraphics[width=.33\textwidth]{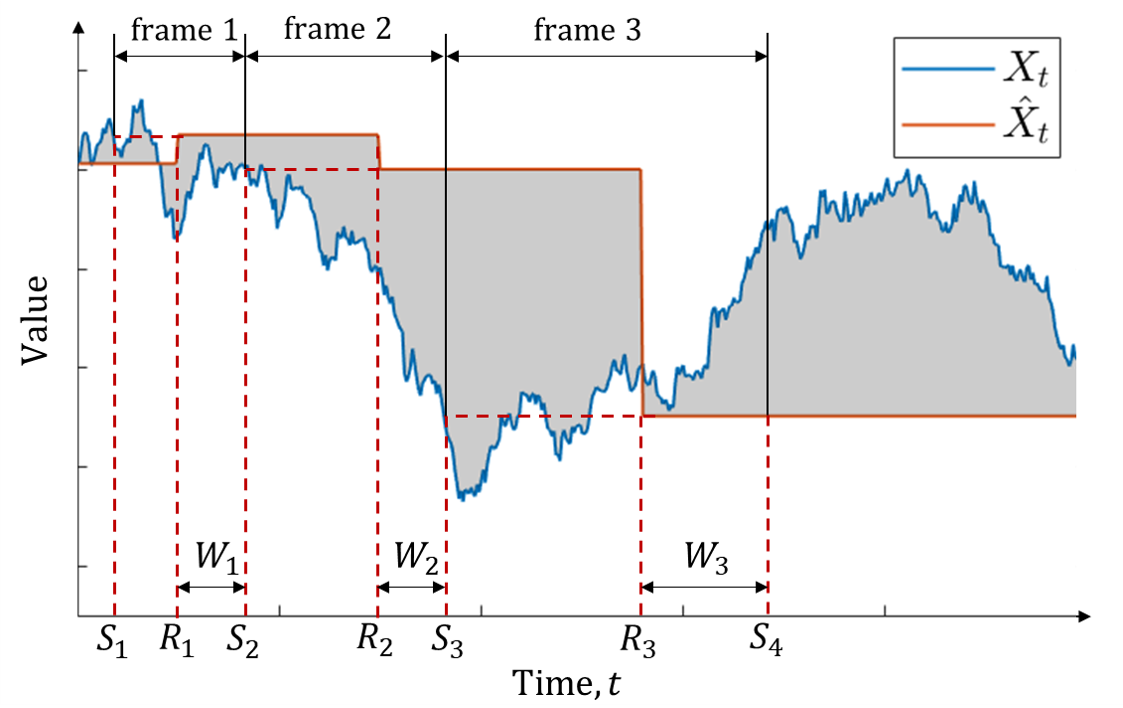}
	\caption{Illustration of the Wiener process and the estimation error. The sampling and reception time-stamp of the $k$-th sample are denoted by $S_k$ and $R_k$, respectively. For MMSE estimator, $\hat{X}_t=X_{S_k}, \forall t\in [R_k, R_{k+1})$. }
	\label{fig:errevolve}
\end{figure}

\subsection{Markov Decision Reformulation~\ref{pb:mse}}\label{sec:dep-rr}
According to \cite[Theorem 1]{sun_wiener}, policy $\pi^\star$ should not take a new sample before the previous sample is delivered to the destination. As is depicted in Fig.~\ref{fig:errevolve}, the waiting time between the delivery time of the $k$-th sample and the sampling time of the $(k+1)$-th sample is denoted by $W_k\geq 0$. Define frame $k$ as the time interval between the sampling time-stamp of the $k$-th and the $(k+1)$-th sample. The following corollary enables us to reformulate Problem~\ref{pb:mse} into a Markov Decision Process. 
\begin{lemma}\label{coro:sig-dep-reformulate}
	Let $\mathcal{I}_k:=(D_k, (X_{S_k+t}-X_{S_k})_{t\geq 0})$ denote the recent information of the sampler in frame $k$. The set of sampling policies that determine the waiting time $W_k$ only based on the recent information $\mathcal{I}_k$ is denoted by $\Pi_{\mathsf{recent}}$. Since for each frame $k$, the difference $X_{S_k+t}-X_{S_k}$ evolves as a Wiener process that is independent of the past $\{X_{S_{k'}+t}-X_{S_{k'}}\}_{k'<k}$, Problem~\ref{pb:mse} can be reformulated into the following Markov decision process:
	\begin{pb}[Markov Decision Process Reformulation]\label{pb:rr}
		\begin{subequations}
			\begin{align}
				\mathsf{mse}_{\mathsf{opt}}\!=\!&\mathop{\inf}_{\pi\in\Pi_{\mathsf{recent}}}\limsup_{K\rightarrow\infty}\left(\frac{\sum_{k=1}^K\mathbb{E}\left[\frac{1}{6}(X_{S_{k+1}}-X_{S_k})^4\right]}{\sum_{k=1}^K\mathbb{E}\left[(S_{k+1}-S_k)\right]}+\overline{D}\right),\label{eq:sig-dep-rr-goal}\\
				&\hspace{0.55cm}\text{s.t. }\hspace{0.15cm}\liminf_{K\rightarrow\infty}\frac{1}{K}\sum_{k=1}^K\mathbb{E}\left[(S_{k+1}-S_k)\right]\geq \frac{1}{f_{\mathsf{max}}}. 
			\end{align}
		\end{subequations}
	\end{pb}
\end{lemma}

The proof is provided in Appendix~\ref{pf:sig-dep-reformulate} of the supplementary material. 

According to \cite[Theorem 1]{sun_wiener}, there exists a stationary policy $\pi^\star$ that selects the waiting time $W_k$ using a conditional probability distribution given the recent $\mathcal{I}_k$ that achieves $\mathsf{mse}_{\mathsf{opt}}$. Next, we will reveal the sufficient conditions of such policy for designing the online algorithm. 

\subsection{Designing $\pi^\star$ with Known $\mathbb{P}_D$}\label{sec:dep-off}

Let $\Pi_{\mathsf{cons}}\triangleq\{\pi\in\Pi_{\mathsf{recent}}|\limsup_{T\rightarrow\infty}\mathbb{E}\left[\frac{i(T)}{T}\right]\leq f_{\mathsf{max}}\}$ denote the set of policies that satisfy the sampling frequency constraint. Since $\pi^\star$ achieves the minimum expected time-average MSE among $\Pi_{\mathsf{cons}}$, we have: 
\begin{equation}
\limsup_{K\rightarrow\infty}\frac{\sum_{k=1}^K\mathbb{E}\left[\frac{1}{6}(X_{S_{k+1}}-X_{S_k})^4\right]}{\sum_{k=1}^K\mathbb{E}[D_k+W_k]}\geq \overline{\mathcal{E}}_{\pi^\star}-\overline{D},\pi\in\Pi_{\mathsf{cons}}. \label{eq:dep-inequal-equiv}
\end{equation}

For simplicity, denote $\gamma^\star:=\overline{\mathcal{E}}_{\pi^\star}-\overline{D}$, which is the average cost of the MDP when the optimum policy $\pi^\star$ is used, i.e., $\gamma^\star=\limsup_{K\rightarrow\infty}\frac{\sum_{k=1}^K\mathbb{E}\left[\frac{1}{6}(X_{S_{k+1}}-X_{S_k})^4\right]}{\sum_{k=1}^K\mathbb{E}[D_k+W_k]}$. Because $\frac{1}{K}\sum_{k=1}^K\mathbb{E}[D_k+W_k]>0$ , for any policy $\pi\in\Pi_{\mathsf{cons}}$, inequality \eqref{eq:dep-inequal-equiv} can be rewritten as:
\begin{align}
\theta_{\pi}(\gamma^\star):=&\liminf_{K\rightarrow\infty}\left(\frac{1}{K}\sum_{k=1}^K\mathbb{E}\left[\frac{1}{6}(X_{S_{k+1}}-X_{S_k})^4\right]\right.\nonumber\\
&\left.-\gamma^\star\cdot\frac{1}{K}\sum_{k=1}^K\mathbb{E}[D_k+W_k]\right)\geq 0.  \label{eq:dep-inequal}
\end{align}

Inequality \eqref{eq:dep-inequal} takes the minimum value 0 if and only if policy $\pi$ is optimum. Therefore, if the ratio $\gamma^\star$ is known, an optimum policy $\pi^\star$ can be obtained by solving the following functional optimization: 
\begin{pb}[Functional Optimization Problem]\label{pb:sig-frac}
\begin{subequations}\begin{align}
		\mathsf{mse}_{\mathsf{opt}}=&\inf_{\pi\in\Pi}\limsup_{K\rightarrow\infty}\left(\frac{1}{K}\sum_{k=1}^K\mathbb{E}\left[\frac{1}{6}\left(X_{S_{k+1}}-X_{S_k}\right)^4\right]\right.\nonumber\\
		&\hspace{1cm}\left.-\gamma^\star\frac{1}{K}\sum_{k=1}^K\mathbb{E}\left[\left(D_k+W_k\right)\right]\right),\label{eq:sig-frac-obj}\\
		&\hspace{0.2cm}\text{s.t.}\hspace{0.2cm}\liminf_{K\rightarrow\infty}\mathbb{E}\left[\frac{1}{K}\sum_{k=1}^K\left(D_k+W_k\right)\right]\geq\frac{1}{f_{\mathsf{max}}}.\label{eq:cons} 
	\end{align}
\end{subequations}
\end{pb}

To solve Problem~\ref{pb:sig-frac}, we can take the Lagrangian duality of the constraint \eqref{eq:cons} with a dual variable $\nu$ and obtain the Lagrange function $\mathcal{L}(\pi, \gamma, \nu)$:
\begin{align}
&\mathcal{L}(\pi, \gamma, \nu)\triangleq\limsup_{K\rightarrow\infty}\left(\frac{1}{K}\sum_{k=1}^K\mathbb{E}\left[\frac{1}{6}(X_{S_{k+1}}-X_{S_k})^4\right]\right.\nonumber\\
&\hspace{2cm}\left.-(\gamma+\nu)\frac{1}{K}\sum_{k=1}^K\mathbb{E}\left[\left(S_{k+1}-S_k\right)\right]\right)+\nu\frac{1}{f_{\mathsf{max}}}.\label{eq:lagrange-dep}
\end{align}

We say that a stationary policy $\pi$ has a threshold structure, if the waiting time $W_k$ is determined by:
\begin{equation}W_k=\inf\{w\geq 0\big||X_{S_k+D_k+w}-X_{S_k}|\geq \tau\}.\label{eq:opt-dep}
\end{equation}

Let $Z_t$ be a Wiener process staring from $t=0$.  Let $D$ be the random transmission delay following distribution $\mathbb{P}_D$ and the value of the Wiener process at the random time $D$ is denoted by $Z_D$. Using the threshold policy \eqref{eq:opt-dep}, the expected frame-length $L_k:=D_k+W_k$ and $\frac{1}{6}(X_{S_{k+1}}-X_{S_k})^4$ has the following properties:
\begin{lemma}\cite[Corollary 1 Restated]{sun_wiener}\label{lemma:cond-l}
\begin{subequations}
	\begin{align}
		&\mathbb{E}[L_k]=\mathbb{E}\left[\max\{\tau^2, Z_D^2\}\right],\\
		&\mathbb{E}\left[\frac{1}{6}(X_{S_{k+1}}-X_{S_k})^4\right]=\frac{1}{6}\mathbb{E}\left[\max\{\tau^2, Z_D^2\}^2\right]. 
	\end{align}
\end{subequations}
\end{lemma}

As is revealed by \cite{sun_wiener}, the optimum policy $\pi^\star$ has a threshold structure as in equation \eqref{eq:opt-dep}. To design an off-line algorithm that can learn the updating threshold $\tau^\star$ of $\pi^\star$, 
we then reveal the necessary conditions that $\tau^\star$ should satisfy. 
With slightly abuse of notations, let $\mathcal{L}(\tau, \gamma, \nu)$ denote the expected value of the Lagrange function $\mathcal{L}(\pi, \gamma, \nu)$ when a stationary policy $\pi$ with threshold $\tau$ is used. According to Lemma~\ref{lemma:cond-l}, $\mathcal{L}(\tau, \gamma, \nu)$ can be computed as follows:
\begin{align}
\mathcal{L}(\tau, \gamma, \nu)=&\mathbb{E}\left[\frac{1}{6}\max\{\tau^2, Z_D^2 \}^2\right]-(\gamma+\nu)\mathbb{E}[\max\{\tau^2, Z_D^2\}]\nonumber\\
&+\nu\frac{1}{f_{\mathsf{max}}}.
\end{align}

\hspace{-10pt}\emph{Condition 1: }\cite[Theorem 6 Restated]{sun_wiener} Let $\tau(\gamma, \nu)$ be the optimum sampling threshold that minimizes function $\mathcal{L}(\tau, \gamma, \nu)$, which can be computed as follows:
\begin{equation}
\tau(\gamma, \nu):=\arg\inf_{\tau\geq 0}\mathcal{L}(\tau, \gamma ,\nu)=\sqrt{3(\gamma+\nu)}.\label{eq:tauopt}
\end{equation}

Recall that for any policy $\pi\in\Pi_{\text{cons}}$ with threshold $\tau$,  inequality \eqref{eq:dep-inequal-equiv} implies
\begin{equation}
\theta_\pi(\gamma^\star)=\frac{1}{6}\mathbb{E}\left[\max\{\tau^2, Z_D^2\}^2\right]-\gamma^\star\mathbb{E}\left[\max\{\tau^2, Z_D^2\}\right]\geq 0. \label{eq:theta}
\end{equation}
According to \eqref{eq:tauopt}, inequality \eqref{eq:theta} holds with equality if and only if $\pi^\star$ with threshold $\tau^\star=\sqrt{3(\gamma^\star+\nu^\star)}$ is used. 

\hspace{-10pt}\emph{Condition 2: }\cite[Eq.~(123, 125)]{sun_wiener}\begin{equation}
	\nu^\star\left(\mathbb{E}\left[\max\{3(\gamma^\star+\nu^\star), Z_D^2\}\right]-\frac{1}{f_{\mathsf{max}}}\right)=0, \nu^\star\geq 0. \label{eq:cs}
\end{equation}
Adding the Complete Slackness (CS) condition \eqref{eq:cs} on both sides of \eqref{eq:theta}, the necessary condition for $\gamma^\star$ then becomes:
\begin{equation}
\overline{g}_\nu(\gamma^\star)=\theta_{\pi^\star}(\gamma^\star)=0, \label{eq:equation-offline}
\end{equation}
where function $\overline{g}_\nu(\gamma):=\mathbb{E}[g_\nu(\gamma;Z_D)]$ is the expectation of function $g_{\nu}(\gamma;Z_D)$ defined as follows: 
\begin{equation}
g_\nu(\gamma;Z_D):=\frac{1}{6}\max\{3(\gamma+\nu), Z_D^2\}^2-\gamma\max\{3(\gamma+\nu), Z_D^2\}.\label{eq:gdef}
\end{equation}

As is shown by \cite[Theorem 7]{sun_wiener}, the duality gap between $\overline{\mathcal{E}}_{\pi^\star}$ and $\sup_{\nu\geq 0}\inf_{\pi}\mathcal{L}(\pi, \gamma^\star, \nu)$ is zero, and \eqref{eq:equation-offline} becomes a necessary and sufficient condition. 
%
%
%
%

\subsection{An Online Algorithm $\pi_{\mathsf{online}}$}\label{sec:dep-online}

When $\mathbb{P}_D$ is unknown but $\nu^\star$ is known, we can approximate $\gamma^\star$ by solving equation \eqref{eq:equation-offline} through stochastic approximation \cite{neely2021fast,Kushner2003,robbins_monro}. Notice that the role of $\nu^\star$ is to satisfy the sampling frequency constraint. To achieve this goal, we approximate $\nu^\star$ by maintaining a sequence $\{U_k\}$ that records the sampling constraint violations up to frame $k$. 

The algorithm is initialized by selecting $\gamma_1=0$ and $U_1=0$
. In each frame $k$, the sampling and updating rules are as follows:

\hspace{-10pt}\underline{1. Sampling: }We treat $\nu_k:=\frac{1}{V}U_k^+$ as the dual optimizer $\nu$, where $V>0$ is fixed as a constant. The waiting time $W_{k+1}$ is selected to minimize the Lagrange function \eqref{eq:lagrange-dep}, and according to the statement after equation \eqref{eq:theta}, $W_k$ is selected by:
\begin{equation}
	W_k=\inf\{w\geq 0|\left|X_{S_k+D_k+w}-X_{S_k}\right|\geq\sqrt{3\left(\gamma_k+\nu_k\right)}\}.\label{eq:o-dep-wait}
\end{equation}

\hspace{-10pt}\underline{2. Update $\gamma_k$: }To search for the root $\gamma>0$ of equation $\overline{g}_{\nu_k}(\gamma)=0$, we update $\gamma_{k}$ through the Robbins-Monro algorithm \cite{robbins_monro}. In each frame $k$, we are given an i.i.d sample $\delta X_k=X_{S_k+D_k}-X_{S_k}\sim Z_D$, and the Robbins-Monro algorithm operates by:
\begin{align}
	&\gamma_{k+1}=\left(\gamma_{k}+\eta_kY_k\right)^+,\label{eq:robbins-monro-gamma}
\end{align}
where $Y_k=g_{\nu_k}(\gamma_k;\delta X_k)$ and function $g_{\nu}(\cdot)$ is defined in \eqref{eq:gdef}. Recall that $\overline{D}_{\text{lb}}$ is a non-zero lower bound of the average delay, the step-size $\{\eta_k\}$ is selected by:
\begin{equation}
	\eta_k=\frac{1}{\overline{D}_{\text{lb}}(2+k^\alpha)},\alpha\in(0.5, 1]. \label{eq:stepsize}
\end{equation}

\hspace{-10pt}\underline{3. Update $U_k$: }To guarantee that the sampling frequency constraint is not violated, we update the violation $U_k$ up to the end of frame $k$ by:
\begin{equation}
U_{k+1}=U_k+\left(\frac{1}{f_{\mathsf{max}}}-(D_k+W_k)\right). \label{eq:debt-evolve}
\end{equation}

\subsection{Theoretical Analysis}\label{sec:dep-analysis}
We analyze the convergence and optimality of algorithm $\pi_{\mathsf{online}}$. We assume there is no sampling frequency constraint, i.e., $f_{\mathsf{max}}=\infty$ and make the following assumption on distribution $\mathbb{P}_D$:
\begin{assu}
The fourth order moment of the transmission delay is upper bounded by $B$, i.e., 
\[\mathbb{E}[D^4]\leq B<\infty.\]
\end{assu}

The convergence behavior of the optimum threshold $3\gamma^\star$ and the MSE performance are manifested in the following theorems:
\begin{theorem}\label{thm:dep-converge}
The proposed algorithm learns the optimum parameter $\gamma^\star$ almost surely, i.e., 
\begin{equation}
\lim_{k\rightarrow\infty}\gamma_k=\gamma^\star, \hspace{0.3cm}\text{w.p.1}.\label{eq:gamma-as}
\end{equation}
\end{theorem}
The proof of Theorem \ref{thm:dep-converge} is obtained by the ODE method in \cite[Chapter 5]{Kushner2003} and is provided in Appendix~\ref{pf:dep-converge}. 

\begin{theorem}\label{thm:rate-converge}
The second moment of $(\gamma_k-\gamma^\star)$ satisfies:
\begin{equation} 
\sup_k\frac{1}{\eta_k}\mathbb{E}\left[|\gamma_k-\gamma^\star|^2\right]<\infty. 
\end{equation}
Specifically, if $\alpha=1$ and $\eta_k=\frac{1}{\overline{D}_{{\rm lb}}(2+k^\alpha)}$, then the mean square error decays with rate $\mathbb{E}[(\gamma_k-\gamma^\star)^2]=\mathcal{O}(1/k)$. 
\end{theorem}

One challenge in the proof of Theorem \ref{thm:rate-converge} is that $\gamma_k$ is unbounded and the second moment of $Y_k$ is unbounded. We notice that $Y_k$ could become very large when $\gamma_k$ is much larger than the true value $\gamma^\star$, but the truncation of $(\gamma_k+\eta_k Y_k)^+$ to non-negative part actually prevents the actual update $|(\gamma_k+\eta_k Y_k)^+-\gamma_k|$ from becoming too large. Based on this observation, we adopt a method from the heavy-traffic analysis by introducing the unused rate $\chi_k:=(-(\gamma_k+\eta_kY_k))^+$, then prove that the variance of the amount of the actual updating $(\eta_kY_k+\chi_k)$ is finite. Detailed proofs are provided in Appendix~\ref{pf:rate-converge}.

\begin{theorem}\label{thm:mse-as}
The average MSE under policy $\pi_{\mathsf{online}}$ converges to $\overline{\mathcal{E}}_{\pi^\star}$ almost surely, i.e.,
\begin{equation}
\limsup_{k\rightarrow\infty}\frac{\int_{t=0}^{S_{k+1}}(X_t-\hat{X}_t)^2{\rm{d}}t}{S_{k+1}}=\overline{\mathcal{E}}_{\pi^\star}, \hspace{0.3cm}\text{w.p.1}.\label{eq:thm-1-}
\end{equation} 
\end{theorem}

With the mean-square convergence of $\gamma_k$, the proof of Theorem~\ref{thm:mse-as} is a direct application of the perturbed ODE method \cite{Kushner2003} and is provided in Appendix~\ref{pf:mse-rate} of the supplementary material.

By using Theorem~\ref{thm:rate-converge} and Theorem~\ref{thm:mse-as}, we can upper bound the growth rate of the cumulative MSE optimality gap in the following corollary:
\begin{corollary}\label{thm:mse-rate}
If $\alpha=1$, then the growth rate of the cumulative MSE optimality gap up to the $k$-th sample can be bounded as follows:
\begin{equation}
\left(\mathbb{E}\left[\int_0^{S_{k+1}}(X_t-\hat{X}_t)^2{\rm d}t\right]-\overline{\mathcal{E}}_{\pi^\star}\mathbb{E}[S_k]\right)=\mathcal{O}\left(\ln k\right). 
\end{equation}
\end{corollary}
The proof of Corollary~\ref{thm:mse-rate} is provided in Appendix~\ref{pf:sig-dep-reformulate} of the supplementary material.

\begin{theorem}\label{thm:converse}
For any distribution $\mathbb{P}$, let $\pi^\star(\mathbb{P})$ denote the MSE minimum sampling policy when the delay $D\sim\mathbb{P}$. The threshold obtained by solving equation \eqref{eq:equation-offline} is denoted by $\gamma^\star(\mathbb{P})$. After $k$-samples are taken, the minimax estimation error $\gamma^\star(\mathbb{P})$ is lower bounded by:
\begin{equation}
\inf_{\hat{\gamma}}\sup_{\mathbb{P}}\mathbb{E}\left[(\hat{\gamma}-\gamma^\star(\mathbb{P}))^2\right]=\Omega(1/k). \label{eq:converse-est}
\end{equation}

	Let $p_w(\mathbb{P}):={\rm Pr}(Z_D^2\leq 3\gamma^\star(\mathbb{P})|D\sim\mathbb{P})$ denote the probability of waiting by using policy $\pi^\star(\mathbb{P})$.
Specifically, let $p_{\rm w, \rm{uni}}^\star:={\rm{Pr}}(Z_D^2\leq 3\gamma^\star_{{\rm uni}}|D\sim{\rm Uni}([0, 1]))$. Let $\Pi_h$ denote the set of policies which the sampling decision $S_k$ is made based on historical information $\mathcal{H}_{k-1}$. We have the following result: \begin{align}
	&\inf_{\pi\in\Pi_h}\sup_{\mathbb{P}}\left(\mathbb{E}\left[\int_0^{S_{k+1}}(X_t-\hat{X}_t)^2{\rm d}t\right]-\overline{\mathcal{E}}_{\pi^\star(\mathbb{P})}\mathbb{E}[S_{k+1}]\right)\nonumber\\
	\geq&\frac{1}{4}\left(\frac{1}{24}(1-\delta)\delta p_{\text{w,uni}}^\star\right)^2\times\left(\sum_{k'=1}^k\frac{1}{k'}\right)=\Omega\left(\ln k\right). \label{eq:regconverse}
\end{align}
\end{theorem}

As the transmission delay $\mathbb{P}_D$ considered in the paper does not belong to a specific family and could be quite general, obtaining a point-wise converse bound on $\mathbb{E}[(\hat{\gamma}-\gamma^\star(\mathbb{P}))^2]$ for each distribution $\mathbb{P}$ is impossible. As an alternative, a minimax risk bound $\mathbb{E}[(\hat{\gamma}-\gamma^\star(\mathbb{P}))^2]$ over a general distribution set $\mathcal{P}$ can be obtained using Le Cam's two point method for non-parametric estimation \cite{nonpara}. The core idea is to construct two distributions $\mathbb{P}_1, \mathbb{P}_2$, whose $\ell_1$ distance $|\mathbb{P}_1^{\otimes k}-\mathbb{P}_2^{\otimes k}|_1$ can be upper bounded by a constant, but $(\gamma^\star(\mathbb{P}_1)-\gamma^\star(\mathbb{P}_2))^2\geq\Omega(1/k)$ is difficult to distinguish. Such a construction is still challenging because $\gamma^\star(\mathbb{P})$ cannot be obtained in closed form even for the simpliest distribution families such as the delta distribution or exponential distribution. Notice that the estimation error of $\gamma^\star$ is closely related to the estimation error $\overline{g}_\nu(\cdot)$ at a given point. Therefore, the construction of $\mathbb{P}_1$ and $\mathbb{P}_2$ for obtaining the converse bound of Hölder smooth functions \cite[Chapter 2]{nonpara} are adopted. The proof of inequality \eqref{eq:regconverse} is a direct application of the minimax estimation error \eqref{eq:converse-est}. Detailed proof of Theorem~\ref{thm:converse} is provided in Appendix \ref{pf:converse}.

	\section{Simulation Results}
	In this section, we provide simulation results to verify the theoretic findings and illustrate the performance of our proposed algorithms. We notice that the MSE minimization problem is closely related to the AoI minimization problem, where the AoI at time $t$, denoted by $A(t)=t-S_{i(t)}$. For signal-ignorant sampling policies (i.e., the sensor cannot always observe the time-varying process), according to the analysis in \cite[Section IV-B]{sun_17_tit}, policies that minimize the average AoI achieves the minimum MSE. Therefore, we choose both offline and online AoI minimization policies ($\pi_{\text{AoI}}^\star$ from \cite{sun_17_tit}, $\pi_{\text{itr}}$ from \cite{chichun-19-isit}) for comparison. To show the convergence of online learning algorithm, we plotted the average MSE performance of the optimum off-line algorithm $\pi^\star$ from \cite{sun_wiener}. 
	
	The transmission delay follows the log-normal distribution parameterized by $\mu$ and $\sigma$ such that the density function of the probability measure $\mathbb{P}_D$ is:
	\[p(x):=\frac{\mathbb{P}_D(\text{d}x)}{\text{d}x}=\frac{1}{\sigma\sqrt{2\pi}}\exp\left(-\frac{(\ln x-\mu)^2}{2\sigma^2}\right). \]
	In simulations, we set $\mu=0.8$ and $\sigma=1.2$, the expected time-averaged MSE is computed by taking the average of 20 runs. Fig.~\ref{fig:mse-frame} depicts the time-average MSE performance up to the $k$-th frame of different sampling policies. The evolution of $\{\gamma_k\}$ and the MSE regret $\mathbb{E}\left[\int_0^{S_{k+1}}(X_t-\hat{X}_t)^2\text{d}t\right]-\overline{E}_{\pi^\star}\mathbb{E}[S_{k+1}]$ are depicted in Fig.~\ref{fig:gap}. From Fig.~\ref{fig:mse-frame}, with $5\times 10^4$ samples, the time averaged MSE is almost the same as using the optimum policy. From Fig.~\ref{fig:reg}, the MSE regret is almost a logarithm function of frame $k$. The asymptotic MSE behaviour is consistent with the convergence results in Theorem~\ref{thm:mse-as} and Corollary~\ref{thm:mse-rate}. 
	\begin{figure}[h]
		\centering
		\includegraphics[width=.5\textwidth]{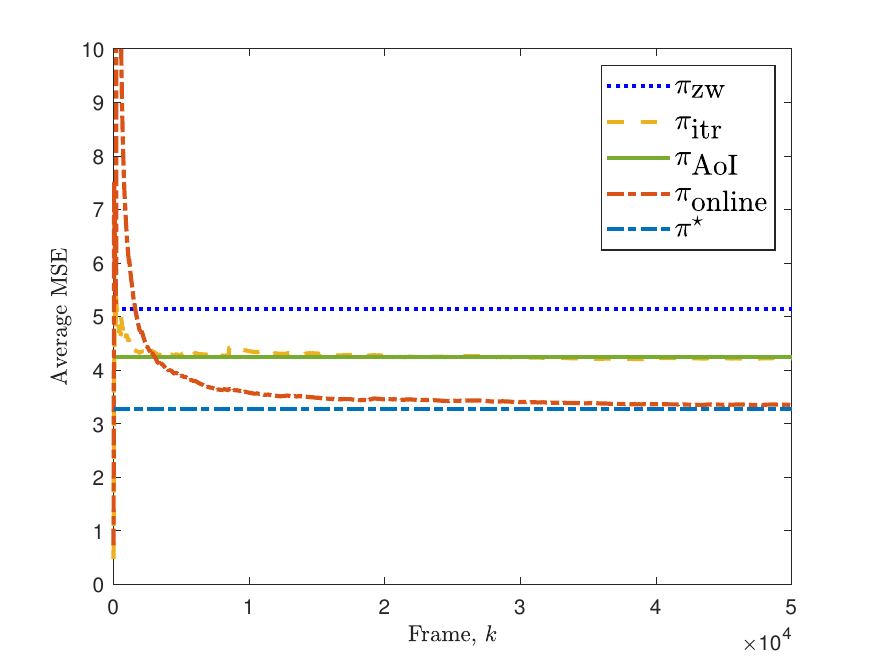}
		\caption{The time average MSE evolution as a function of frame $k$. }
		\label{fig:mse-frame}
	\end{figure}
	
	\begin{figure}[h]
		\centering
		\includegraphics[width=.5\textwidth]{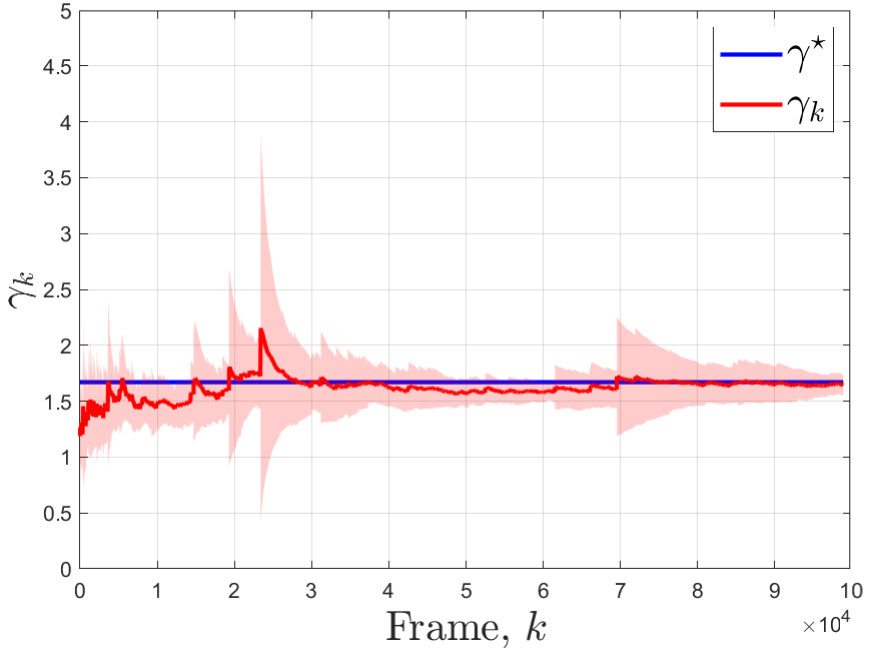}
		\caption{The evolution of the threshold estimate $\gamma_k$. }
		\label{fig:gap}
	\end{figure}
	\begin{figure}[h]
		\centering
		\includegraphics[width=.5\textwidth]{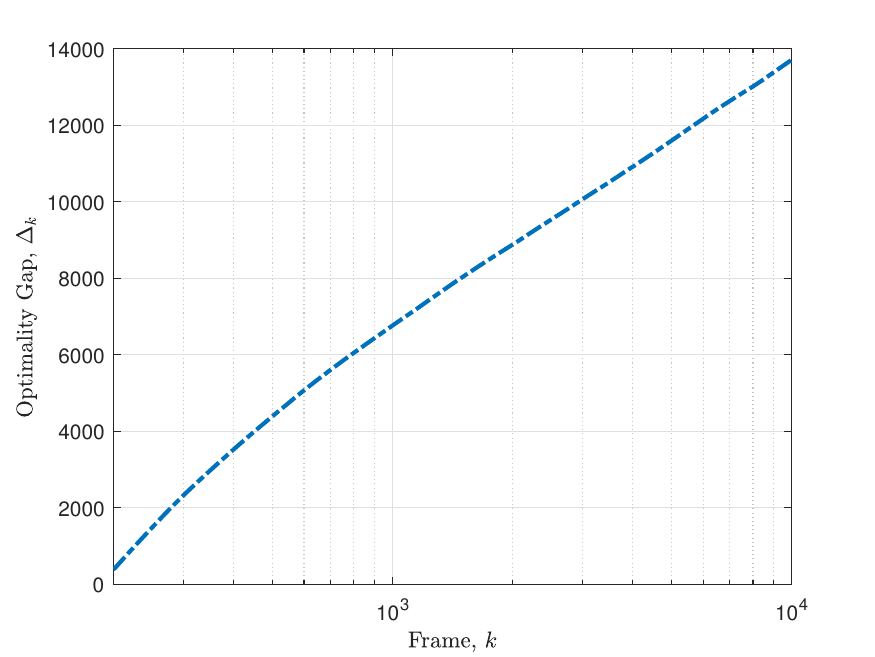}
		\caption{the MSE regret $\Delta_k:=\mathbb{E}\left[\int_0^{S_{k+1}}(X_t-\hat{X}_t)^2\text{d}t\right]-\overline{\mathcal{E}}_{\pi^\star}\mathbb{E}[S_{k+1}]$ (right)}
		\label{fig:reg}
	\end{figure}
	
	When there is a sampling frequency constraint, the average MSE and the average sampling interval achieved by policy $\pi_{\text{online}}$ are depicted in Fig.~\ref{fig:cons} and Fig.~\ref{fig:inte}, respectively. We set $f_{\text{max}}=\frac{1}{10\overline{D}}$. From these figures, one can observe that the average MSE of $\pi_{\text{online}}$ is close to the optimum MSE $\overline{\mathcal{E}}_{\pi^\star}$ and the sampling frequency can be satisfied. In addition, by choosing  a larger $V$, a smaller MSE performance can be achieved, whereas a larger number of iterations are needed to meet the sampling frequency constraint. 
	
	\begin{figure}[h]
		\centering
		\includegraphics[width=.5\textwidth]{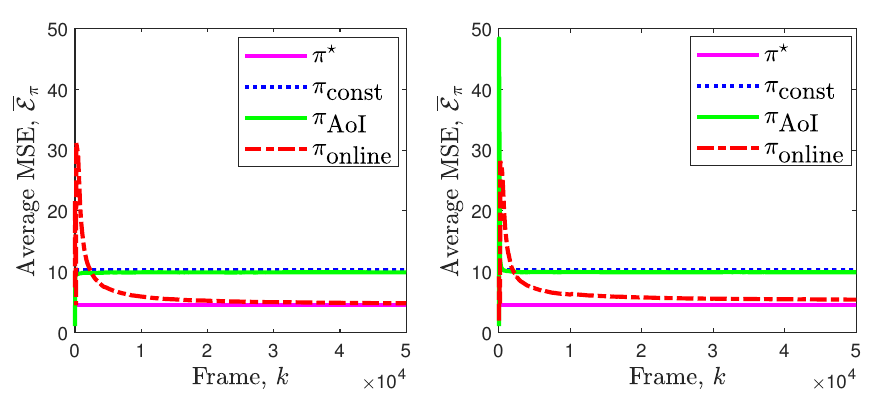}
		\caption{The time average MSE evolution as a function of frame $k$. (Left: $V=10$, Right: $V=1$. ) }
		\label{fig:cons}
	\end{figure}
	
	\begin{figure}[h]
		\centering
		\includegraphics[width=.5\textwidth]{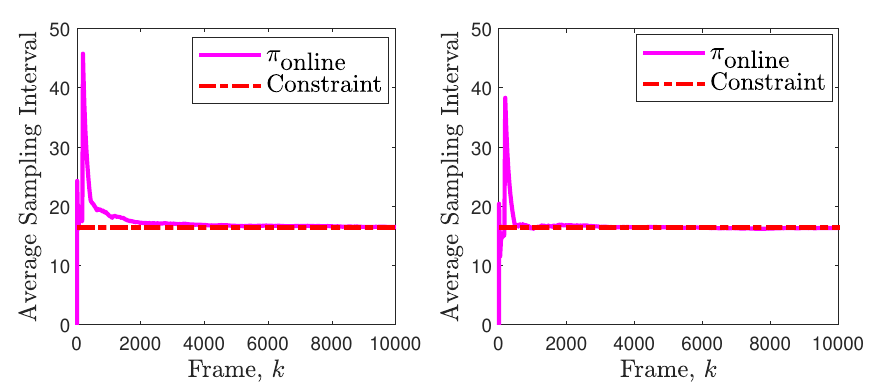}
		\caption{The average sampling interval under different constant $V$. (Left: $V=10$, Right: $V=1$. )}
		\label{fig:inte}
	\end{figure}

		\section{Conclusions}

In this work, we studied the problem of sampling a Wiener process for remote estimation over a channel with unknown delay statistics. By reformulating the MSE minimization problem as a renewal-reward process, we proposed an online sampling algorithm that can adaptively learn the optimum algorithm as the number of samples grows. We showed that the average MSE obtained by the proposed algorithm converges to the minimum MSE almost surely, and the cumulative MSE has an order of $\mathcal{O}(\ln k)$, where $k$ is the number of samples. We then prove that the cumulative MSE regret of any algorithm is at best $\Omega(\ln k)$. Numerical simulation results validate the convergence behaviors of the proposed algorithm. 
\appendices\label{sec:pf}
\section{Notations and Preliminary Lemmas}
In Table~\ref{tab:notations}, we summarize the notations used in the following proofs. Throughout the proofs, we use $N_1, N_2, \cdots$ to denote absolute constants and $C_1(\cdot), C_2(\cdot)$ to denote polynomials with finite order. For ease
of exposition, the specific values and expressions of the constants and functions may vary across different context.
\begin{table}[h]
	\caption{Notations}
	\label{tab:notations}
	\begin{tabular}{|c|p{7cm}|}
		\hline
		Notation&Meaning\\
		\hline
		$Z_t$& a Wiener process staring from time 0\\
		$l_\gamma$ & length of running time using stopping rule $\tau_\gamma:=\inf\{{t\geq D}||Z_t|\geq\sqrt{3\gamma}\}$\\ 
		$\delta X_k$ & $\delta X_k:=X_{S_k+D_k}-X_{S_k}$\\
		$Q_k$ & $Q_k:=\frac{1}{6}\left(X_{S_k+D_k}-X_{S_k}\right)^4$\\
		$L_k$ & $L_k:=S_{k+1}-S_k=D_k+W_k$, frame length $k$\\
		$E_k$ & $E_k:=\int_{S_k}^{S_{k+1}}(X_t-\hat{X}_t)^2\text{d}t$, cumulative estimation error in frame $k$\\
		$q(\gamma)$ & $q(\gamma):=\frac{1}{6}\mathbb{E}\left[\max\{3\gamma_k, Z_D^2\}^2\right]$, the expectation of $Q_k$ when $\gamma_k=\gamma$\\
		$l(\gamma)$ & $l(\gamma):=\mathbb{E}[\max\{3\gamma, Z_D^2\}]$, expected frame length $L_k$ when $\gamma_k=\gamma$\\
		$\mathcal{I}_k$ & $(D_k, (X_{t}-X_{S_k})_{S_k\leq t<S_{k+1}}$, information in frame $k$\\
		$\mathcal{H}_k$ & $\mathcal{H}_k:=\{\mathcal{I}_{\kappa}\}_{\kappa\leq k}$ historical information up to the end of frame $k$\\
		$\mathbb{E}_k[\cdot]$ & Conditional expectation  $\mathbb{E}[\cdot|\mathcal{H}_{k-1}]$\\
		$t_k$ & $t_k:=\sum_{i=1}^k\eta_k$ or $t_k:=\sum_{i=1}^k\epsilon_k$, the cumulative step-sizes depending on the context\\
		$m(t)$ & $m(t)$ is the unique $k$ so that $t_k\leq t\leq t_{k+1}$\\
		\hline
	\end{tabular}
\end{table}

\begin{lemma}\label{coro:gammadep-bound}
	Let $M:=\mathbb{E}[D^2]$, the optimum ratio $\gamma^\star$ is upper and lower bounded by:
	\begin{equation}
		\frac{1}{6}\overline{D}\leq \gamma^\star\leq\frac{1}{2}\frac{M+2\overline{D}\frac{1}{f_{\mathsf{max}}}+\frac{1}{f_{\mathsf{max}}^2}}{\overline{D}+\frac{1}{f_{\mathsf{max}}}}. 
	\end{equation}
\end{lemma}
The proof is provided in Appendix~\ref{pf:gammadep-bound}.

\begin{lemma}\label{lemma:4order}
	For threshold $\gamma<\infty$, the first, second and fourth order moments of the stopping time $\tau_\gamma$ are bounded, i.e., 
	\begin{subequations}
		\begin{align}
			&\mathbb{E}[l_\gamma]\leq 3\gamma+\overline{D},\\
			&\mathbb{E}[l_{\gamma}^2]\leq\frac{10}{3}\left((3\gamma)^2+3\sqrt{B}\right), \\ &\mathbb{E}\left[l_{\gamma}^4\right]<4^3\left((3\gamma)^4+105B\right)<\infty. 
		\end{align}
	\end{subequations}
\end{lemma}
The proof of Lemma~\ref{lemma:4order} is provided in Appendix~\ref{pf:4order}.


\begin{lemma}\label{lemma:g}
	Function $\overline{g}_0(\gamma)=q(\gamma)-\gamma l(\gamma)$ and has the following properties:
	\begin{itemize}
		\item[(i)] $\overline{g}_0(\gamma)$ is concave and monotonically decreasing. The second order derivative $-3\leq\overline{g}_0''(\gamma)\leq 0$. 
		
		\item[(ii)] $\overline{g}_0(\gamma^\star)=0$
		
		\item[(iii)] For $\gamma\neq\gamma^\star$, 
		$(\gamma-\gamma^\star)\overline{g}_0(\gamma)\leq-l(\gamma^\star)(\gamma-\gamma^\star)^2\leq 0$. 
	\end{itemize}
\end{lemma}

The proof of Lemma \ref{lemma:g} is provided in Appendix~\ref{pf:lemma-g}.


\begin{corollary}\label{coro:squareub}
	For each $\gamma_k<\infty$, if the fourth order moment of the delay satisfies $\mathbb{E}[D^4]<B<\infty$, given historical transmission $\mathcal{H}_{k-1}$, the conditional second order moment of the cumulative error in frame $E_k=\int_{S_k}^{S_{k+1}}(X_t-\hat{X}_t)^2\text{d}t$ can be bounded as follows:
	\begin{align}
		&\mathbb{E}_k[E_k^2]=3(X_{S_k}-X_{S_{k-1}})^2\sqrt{B}\nonumber\\
		&\hspace{1cm}+12C_1(\gamma_k, B)(X_{S_k}-X_{S_{k-1}})^2+3C_2(\gamma_k, B)<\infty,
	\end{align}
	where $C_1$ and $C_2$ are fourth order polynomials of $\gamma$. 
\end{corollary}

The proof of Corollary~\ref{coro:squareub} is provided in Appendix~\ref{pf:corosquareub}.

	\section{Proof of Theorem~\ref{thm:dep-converge}}\label{pf:dep-converge}

To show that $\gamma_k$ converges to $\gamma^\star$ almost surely, we use the sufficient condition from \cite[p.190, Theorem 7.1]{Kushner2003}. Recall that the step-size $\eta_k$=$\frac{1}{k}$. Define $t_0=0$ and denote the sum of step-sizes up to frame $k$ by $t_k:=\sum_{i=1}^k\eta_k$. 
For $t\geq 0$, let $m(t)$ be the unique $k\in\mathbb{N}^+$ so that $t_k\leq t<t_{k+1}$. 
Without a sampling constraint, $\nu_k\equiv 0, \forall k$. Then the update rule for $\gamma_k$ from equation \eqref{eq:robbins-monro-gamma} can be rewritten in the following recursive form:
\begin{equation}
	\gamma_{k+1}=\gamma_k+\eta_k\underbrace{\left(\overline{g}_0(\gamma_k)+\delta M_k\right)}_{Y_k},
\end{equation}
where we recall that $Y_k=\frac{1}{6}\max\{3\gamma_k, \delta X_k^2\}^2-\gamma_k\max\{3\gamma_k, \delta X_k^2\}$ and $\delta M_k$ is the difference between realization and the conditional expectation  $\mathbb{E}_k[Y_k]=\overline{g}_{0}(\gamma_k)$. Notice that the difference $\delta M_k:=Y_k-\mathbb{E}_k[Y_k]$ depends only on the transmission delay and the Wiener process evolution $(X_t-X_{S_k})$ in frame $k$ and $\gamma_k$, which can be predictable given $\mathcal{H}_{k-1}$ and is therefore a martingale sequence. We then show that $\{Y_k\}, \{\delta M_k\}$ have the following properties, 	

\hspace{-10pt}\textbf{(1.1)} For each constant $N<\infty$, $\sup_k\mathbb{E}[|Y_k|\mathbb{I}_{(|\gamma_k|\leq N)}]$ is bounded, i.e.,
\begin{align}
	&\sup_k\mathbb{E}\left[|Y_k|\mathbb{I}_{(|\gamma_k|\leq N)}\right]\nonumber\\
	\leq &\sup_k\mathbb{E}\left[\frac{1}{6}\max\{3\gamma_k, \delta X_k^2\}^2\cdot\mathbb{I}_{(|\gamma_k|\leq N)}\right]\nonumber\\
	&+\sup_k\mathbb{E}\left[\gamma_k\max\{3\gamma_k,\delta X_k^2\}\mathbb{I}_{(|\gamma_k|\leq N)}\right]\nonumber\\
	<&\frac{1}{6}\left(9N^2+\mathbb{E}[ Z_D^4]\right)+N\cdot\left(3N+ \mathbb{E}[Z_D^2]\right)\overset{(a)}{\leq}\infty. 
\end{align}
where inequality $(a)$ is because $\mathbb{E}[Z_D^4]=3\mathbb{E}[D^2]\leq3\sqrt{\mathbb{E}[D^4]}<\infty$ and $\mathbb{E}[Z_D^2]=\mathbb{E}[D]<\infty$. 

\hspace{-10pt}\textbf{(1.2)} Function $Y_k=g(\gamma_k;\delta x)$ is continuous in $\gamma_k$ for each $\delta x$. 

\hspace{-10pt}\textbf{(1.3)} The martingale sequence $\delta M_k\mathbb{I}_{(|\gamma_k|\leq N)}$ can be bounded as follows:
\begin{align}
	&\text{Var}[\delta M_k\mathbb{I}_{(|\gamma_k|\leq N)}]\leq\mathbb{E}[Y_k^2\mathbb{I}_{(|\gamma_k|\leq N)}]\nonumber\\
	\leq&\mathbb{E}\left[\left(\frac{1}{6}\max\{3\gamma_k, \delta X_k^2\}^2-\gamma_k\max\{3\gamma_k, \delta X_k^2\}\right)^2\mathbb{I}_{(|\gamma_k|\leq N)}\right]\nonumber\\
	\overset{(b)}{\leq}&\mathbb{E}\Bigg[2\left(\left(\frac{1}{6}\max\{3\gamma_k,\delta X_k^2\}^2\right)^2+\gamma_k^2\left(\max\{3\gamma_k, \delta X_k^2\}\right)^2\right)\nonumber\\
	&\times\mathbb{I}_{(|\gamma_k|\leq N)}\Bigg]\nonumber\\
	\overset{(c)}{\leq}&2\left(\frac{1}{36}(3 N)^4+105B+N^2(9N^2+3\sqrt{B})\right)\leq N_1.\label{eq:deltaMub}
\end{align}
where inequality $(b)$ is because $\mathbb{E}[(a-b)^2]\leq\mathbb{E}[2(a^2+b^2)]$; inequality $(c)$ is because $\delta X_t\sim Z_D$ is a Wiener process starting from $t=0$ and therefore, $\mathbb{E}[Z_D^8]=105\mathbb{E}[D^4]\leq 105 B$. 

Since sequence $\delta M_k\mathbb{I}_{(|\gamma_k|\leq N)}$ has mean zero. Its value only depends on $\gamma_k$ and the Wiener process evolution in frame $k$. The correlation $\mathbb{E}\left[\delta M_i\mathbb{I}_{(|\gamma_i|\leq N)}\cdot \delta M_j\mathbb{I}_{(|\gamma_j|\leq N)}\right]=0, \forall i\neq j$. As the variance of $\delta M_k\mathbb{I}_{(\gamma_k\leq N)}$ is bounded in inequality \eqref{eq:deltaMub}, the stepsizes $\eta_k$ satisfies $\sum_{k=1}^\infty\frac{1}{2D_{\mathsf{lb}}}k^{-2\alpha}=\frac{1}{2_{\mathsf{lb}}}\left(1+\frac{1}{2\alpha-1}\right)$, according to
\cite[Chapter 5, Eq. (5.3.18)]{Kushner2003}, for each $\mu>0$ we have
\begin{equation}
	\lim_{k\rightarrow\infty}\text{Pr}\left(\sup_{j\geq k}\max_{0\leq t\leq T}\left|\sum_{i=m(jT)}^{m(jT+t)-1}\epsilon\delta M_i\mathbb{I}_{(|\gamma_i|\leq N)}\right|\geq\mu\right)=0. 
\end{equation}

Let $\gamma_k(\omega)$ be the value of ratio $\gamma$ on sample path $\omega$. Recall that the stepsizes $\{\eta_k\}$ selected in \eqref{eq:stepsize} satisfies $\sum_{k=1}^\infty\eta_k=\infty, \sum_{k=1}^\infty\eta_k^2\leq\infty$. According to \cite[p.170, Theorem 1.2]{Kushner2003}, with probability 1, the limit $\lim_{k\rightarrow\infty}\theta_k(\omega)$ are trajectories of the following ordinary differential equation (ODE), i.e., 
\begin{equation}
	\dot{\gamma}=\overline{g}_0(\gamma). \label{ode:dep}
\end{equation}

The next step is to show the solution of the ODE in equation~\eqref{ode:dep} converges to $\gamma^\star$ as time diverges. 
Equation \eqref{eq:equation-offline} implies $\overline{g}_0(\gamma^\star)=0$ and therefore, $\gamma^\star$ is an equilibrium point of ODE~\eqref{ode:dep}. To show that the ODE is stationary at $\gamma=\gamma^\star$, we use the Lyapunov approach by defining function $V(\gamma):=\frac{1}{2}\left(\gamma-\gamma^\star\right)^2$,whose time derivative $\dot V=\frac{\text{d}}{\text{d}t}V(\gamma(t))$ can be computed by;
\begin{align}
\dot V=\left(\gamma-\gamma^\star\right)\dot\gamma
=\left(\gamma-\gamma^\star\right)\overline{g}_0(\gamma).
\end{align}

According to Lemma~\ref{lemma:g}-(iii), $\dot V=(\gamma-\gamma^\star)\overline{g}_0(\gamma)<0$, the stability of $\gamma^\star$ is verified through Lyapunov theorem. 

\section{Proof of Theorem~\ref{thm:rate-converge}}\label{pf:rate-converge}
The analysis of the convergence rate is obtained through Lyapunov analysis, where the Lyapunov function is denoted by $V(\gamma):=\frac{1}{2}(\gamma-\gamma^\star)^2$. The proof is divided into two steps: first we will upper bound the Lyapunov drift for each $\gamma_k$ by showing the following equation holds:
\begin{equation}
	\mathbb{E}_k[V(\gamma_{k+1})]-V(\gamma_k)\leq -\eta_k\overline{D}_{\mathsf{lb}}V(\gamma_k)+\mathcal{O}(\eta_k^2N_1). \label{eq:step-sizedrift}
\end{equation}
Then, based on \eqref{eq:step-sizedrift}, we then compute $\mathbb{E}[V(\gamma_k)]$ directly. 

\hspace{-10pt}\textbf{Step 1: Bounding the Lyapunov Drift:} The analysis is divided into two cases: For $\gamma_k\leq 3\gamma^\star$, inequality \eqref{eq:step-sizedrift} can be verified easily (Case 1); For $\gamma_k\geq 3\gamma^\star$ we will first establish the relationship between $\mathbb{E}_k[V(\gamma_{k+1})]-V(\gamma_k)$ and $\text{Var}[Y_k]$, then upper bound $\text{Var}[Y_k]$ using the fact that $Z_D^2$ is sub-Gaussian when $D$ is fourth order bounded (Case 2). Detailed proofs are as follows: 

\hspace{-10pt}\textbf{Case 1:} If $\gamma_k\leq 3\gamma^\star$, we have:\begin{align}
	&\mathbb{E}_k[V(\gamma_{k+1})]-V(\gamma_k)\nonumber\\
	=&\mathbb{E}_k\left[\frac{1}{2}\left((\gamma_k+\eta_kY_k)^+-\gamma^\star\right)^2\right]-\frac{1}{2}(\gamma_k-\gamma^\star)^2\nonumber\\
	\leq&\mathbb{E}_k\left[\frac{1}{2}(\gamma_k-\gamma^\star+\eta_kY_k)^2-\frac{1}{2}(\gamma_k-\gamma^\star)^2\right]\nonumber\\
	\overset{(a)}{=}&(\gamma_k-\gamma^\star)\eta_k\overline{g}_0(\gamma_k)\nonumber\\
	&+\frac{1}{2}\eta_k^2\mathbb{E}_k\left[\left(\frac{1}{6}\max\{3\gamma_k, \delta X_k^2\}^2-\gamma_k\max\{3\gamma_k, \delta X_k^2\}\right)^2\right]\nonumber\\
	\overset{(b)}{\leq}&-2\eta_kl(\gamma^\star)V(\gamma_k)\nonumber\\
	&+\frac{1}{2}\eta_k^2\left(\frac{1}{36}((9\gamma^\star)^4+B)+(3\gamma^\star)^2((9\gamma^\star)^2+3\sqrt{B})\right),\label{eq:stepsizez-ub1}
\end{align}
where equality $(a)$ is because $\mathbb{E}_k[Y_k]=\mathbb{E}_k[g_0(\gamma_k;\delta X_k)]=\overline{g}_0(\gamma_k)$; inequality $(b)$ is obtained because according to Lemma~\ref{lemma:g}-(iii), $(\gamma_k-\gamma^\star)\overline{g}_0(\gamma_k)\leq -l(\gamma^\star)(\gamma_k-\gamma^\star)^2=-2l(\gamma^\star)V(\gamma_k)$ and the assumption that $\gamma_k\leq 3\gamma^\star$. 

\hspace{-10pt}\textbf{Case 2:} If $\gamma_k\geq 3\gamma^\star$,  $\gamma_{k+1}=\left(\gamma_k+\eta_kY_k\right)^+$ is truncated into the non-negative real part. We can view the evolution of $\gamma_k$ as a queueing system, where the queue $\gamma_k$ is non-negative, $\eta_k Y_k$ is the arrival rate minus the service rate. Therefore, it is natural to introduce the ``unused rate'' from \cite{eryilmaz2012asymptotically}, which is denoted by $\chi_k:=\left(-\left(\gamma_k+\eta_kY_k\right)\right)^+$. If $\chi_k=0$, $(\gamma_k+\eta_kY_k)\chi_k=0=-\chi_k^2$ and if $\chi_k\geq 0$, $\gamma_k+\eta_kY_k=-\chi_k$, therefore
\begin{equation} (\gamma_k+\eta_kY_k)\chi_k=-\chi_k^2.\label{eq:unused}
\end{equation}
Since $\gamma_k+\eta_kY_k+\chi_k\geq 0$, we have:
\begin{equation}
	-\mathbb{E}_k[\gamma_k+\eta_kY_k]\leq\mathbb{E}_k[\chi_k]. \label{eq:chiub}
\end{equation}

We can then upper bound $\mathbb{E}_k[V(\gamma_{k+1})-V(\gamma_k)]$ by:
\begin{align}
	&\mathbb{E}_k\left[V(\gamma_{k+1})-V(\gamma_k)\right]\nonumber\\
	=&\mathbb{E}_k\left[\frac{1}{2}\left(\gamma_k-\gamma^\star+\eta_kY_k+\chi_k\right)^2-\frac{1}{2}\left(\gamma_k-\gamma^\star\right)^2\right]\nonumber\\
	=&\mathbb{E}_k\left[\frac{1}{2}\left(\gamma_k-\gamma^\star+\eta_kY_k\right)^2-\frac{1}{2}(\gamma_k-\gamma^\star)^2\right.\nonumber\\
	&\left.+\frac{1}{2}\chi_k^2+(\gamma_k+\eta_kY_k)\chi_k-\gamma^\star\chi_k\right]\nonumber\\
	\overset{(c)}{=}&\mathbb{E}_k\left[\frac{1}{2}(\gamma_k-\gamma^\star+\eta_kY_k)^2-\frac{1}{2}(\gamma_k-\gamma^\star)^2-\frac{1}{2}\chi_k^2-\gamma^\star\chi_k\right]\nonumber\\
	\overset{(d)}\leq&\frac{1}{2}\left(\gamma_k-\gamma^\star+\eta_k\mathbb{E}_k[Y_k]\right)^2-\frac{1}{2}(\gamma_k-\gamma^\star)^2+\frac{1}{2}\eta_k^2\text{Var}[Y_k]\nonumber\\
	&-\frac{1}{2}\mathbb{E}_k[\chi_k]^2-\gamma^\star\mathbb{E}_k[\chi_k]\nonumber\\
	=&\frac{1}{2}\left(\gamma_k-\gamma^\star+\eta_k\mathbb{E}_k[Y_k]\right)^2-\frac{1}{2}(\gamma_k-\gamma^\star)^2+\frac{1}{2}\eta_k^2\text{Var}[Y_k]\nonumber\\
	&-\frac{1}{2}\left(-\mathbb{E}_k[\chi_k]-\gamma^\star\right)^2+\frac{1}{2}(\gamma^\star)^2,\label{eq:stepsize-drift}
\end{align}
where equality $(c)$ is because equation \eqref{eq:unused}; inequality $(d)$ is obtained because $\mathbb{E}_k[\chi_k^2]\geq\mathbb{E}_k[\chi_k]^2\geq 0$;

To upper bound \eqref{eq:step-sizedrift}, we then further divide the analysis into two cases:

\hspace{-10pt}\underline{Case 2(a)}: If $\mathbb{E}_k[\gamma_k+\eta_k Y_k]\leq \gamma^\star$, we then have $\mathbb{E}_k[\gamma_k-\gamma^\star+\eta_k Y_k]\leq 0$. According to \eqref{eq:chiub},  $|-\mathbb{E}_k[\chi_k]-\gamma^\star|\geq |\gamma_k-\gamma^\star+\eta_k\mathbb{E}_k[Y_k]|$. Therefore, inequality \eqref{eq:stepsize-drift} can be upper bounded by:
\begin{align}
	&\mathbb{E}_k\left[V(\gamma_{k+1})-V(\gamma_k)\right]\nonumber\\
	\leq&-\frac{1}{2}(\gamma_k-\gamma^\star)^2+\frac{1}{2}(\gamma^\star)^2+\frac{1}{2}\eta_k^2\text{Var}[Y_k]\nonumber\\
	\overset{(e)}{\leq}&-\frac{1}{4}(\gamma_k-\gamma^\star)^2+\frac{1}{2}\eta_k^2\text{Var}[Y_k]\nonumber\\
	\overset{(f)}{\leq}&-2\eta_k\overline{D}_{\mathsf{lb}}V(\gamma_k)+\frac{1}{2}\eta_k^2\text{Var}[Y_k],\label{eq:stepsizez-ub2}
\end{align} 
where inequality $(e)$ is obtained because $\frac{1}{4}(\gamma_k-\gamma^\star)^2\geq(\gamma^\star)^2\geq\frac{1}{2}(\gamma^\star)^2$ because in Case 2 we have $\gamma_k\geq 3\gamma^\star$; inequality $(f)$ is obtained because $\eta_k\overline{D}_{\mathsf{lb}}\leq\frac{1}{2}$ by the step-size selection rule in equation \eqref{eq:stepsize}.  

\hspace{-10pt}\underline{Case 2(b):} If $\mathbb{E}_k[\gamma_k+\eta_kY_k]\geq \gamma^\star$, considering that $\mathbb{E}_k[Y_k]=\overline{g}_0(\gamma_k)<0$ for $\gamma_k\geq \gamma^\star$, we have $0>\mathbb{E}_k[\eta_kY_k]\geq-(\gamma_k-\gamma^\star)$. Inequality \eqref{eq:stepsize-drift} can be bounded by:
\begin{align}
	&\mathbb{E}_k[V(\gamma_{k+1})-V(\gamma_k)]\nonumber\\
	\overset{(g)}{\leq}&\frac{1}{2}(\gamma_k-\gamma^\star)(\gamma_k-\gamma^\star+\eta_k\mathbb{E}_k[Y_k])-\frac{1}{2}(\gamma_k-\gamma^\star)^2\nonumber\\
	&+\frac{1}{2}\eta_k^2\text{Var}[Y_k]\nonumber\\
	\leq&\frac{1}{2}\eta_k(\gamma_k-\gamma^\star)\overline{g}_0(\gamma_k)+\frac{1}{2}\eta_k^2\text{Var}[Y_k]\nonumber\\
	\overset{(h)}{\leq} &-\frac{1}{2}\eta_kl(\gamma^\star)(\gamma_k-\gamma^\star)^2+\frac{1}{2}\eta_k^2\text{Var}[Y_k]\nonumber\\
	=&-\eta_kl(\gamma^\star) V(\gamma_k)+\frac{1}{2}\eta_k^2\text{Var}[Y_k],\label{eq:stepsizez-ub3}
\end{align}
where equality $(g)$ is because $\left(-\mathbb{E}_k[\chi_k]-\gamma^\star\right)^2\geq(\gamma^\star)^2$ and $(\gamma_k-\gamma^\star+\eta_k\mathbb{E}_k[Y_k])^2\leq(\gamma_k-\gamma^\star+\eta_k\mathbb{E}_k[Y_k])(\gamma_k-\gamma^\star)$; inequality $(h)$ is due to Lemma~\ref{lemma:g}-(iii). 

For proceed to show inequality \eqref{eq:step-sizedrift} for $\gamma_k\geq 3\gamma^\star$, we need to upper bound $\text{Var}[Y_k]$ in inequalities \eqref{eq:stepsizez-ub2} and \eqref{eq:stepsizez-ub3}. First, we compute the expectation $\mathbb{E}[Y_k]$ as follows:
\begin{align}
	\mathbb{E}_k[Y_k]=&\mathbb{E}\left[\frac{1}{6}\max\{3\gamma_k, Z_D^2\}^2-\gamma_k\max\{3\gamma_k, Z_D^2\}\right]\nonumber\\
	=&-\frac{3}{2}\gamma_k^2+\mathbb{E}\left[(\frac{1}{6}Z_D^4-\gamma_kZ_D^2+\frac{3}{2}\gamma_k^2)\mathbb{I}_{(Z_D^2\geq 3\gamma_k)}\right]\nonumber\\
	=&-\frac{3}{2}\gamma_k^2+\mathbb{E}\left[\frac{1}{6}(Z_D^2-3\gamma_k)^2\mathbb{I}_{(Z_D^2\geq 3\gamma_k)}\right]\nonumber\\
	\leq&-\frac{3}{2}\gamma_k^2+\mathbb{E}\left[\frac{1}{6}(Z_D^2)^2\right]\nonumber\\
	\leq&-\frac{3}{2}\gamma_k^2+\frac{1}{2}\mathbb{E}[D^2]\leq-\frac{3}{2}\gamma_k^2+\frac{1}{2}\sqrt{B}. 
\end{align}

Given historical information $\mathcal{H}_{k-1}$, the variance of $Y_k$ can be computed by:
\begin{align}
	&\text{Var}[Y_k|\mathcal{H}_{k-1}]\nonumber\\
	=&\mathbb{E}_k\left[(Y_k-\mathbb{E}_k[Y_k])^2\right]\nonumber\\
	=&\mathbb{E}_k\left[\left(\frac{1}{6}Z_D^4-\gamma_kZ_D^2+\frac{3}{2}\gamma_k^2-\frac{3}{2}\gamma_k^2-\mathbb{E}_k[Y_k]\right)^2\mathbb{I}_{(Z_D^2\geq 3\gamma_k)}\right]\nonumber\\
	&+\mathbb{E}_k\left[\left(-\frac{3}{2}\gamma_k^2-\mathbb{E}_k[Y_k]\right)^2\mathbb{I}_{(Z_D^2\leq 3\gamma_k)}\right]\nonumber\\
	\overset{(h)}{\leq} &\frac{1}{4}B+2\mathbb{E}_k\left[\left(\frac{1}{6}Z_D^4-\gamma_kZ_D^2+\frac{3}{2}\gamma_k^2\right)^2\mathbb{I}_{(Z_D^2> 3\gamma_k)}\right]\nonumber\\
	&+2\mathbb{E}_k\left[\left(-\frac{3}{2}\gamma_k^2-\mathbb{E}_k[Y_k]\right)^2\mathbb{I}_{(Z_D^2> 3\gamma_k)}\right]\nonumber\\
	\leq&\frac{3}{4}B+\frac{1}{3}\mathbb{E}_k\left[(Z_D^2-3\gamma_k)^4\mathbb{I}_{(Z_D^2\geq 3\gamma_k)}\right]\nonumber\\
	\leq&\frac{3}{4}B+\frac{1}{3}\mathbb{E}[Z_D^8]\leq(35+\frac{3}{4})B,
\end{align}
where $(i)$ is because $\mathbb{E}_k[Y_k]\leq-\frac{3}{2}\gamma_k^2+\frac{1}{2}\sqrt{B}$ implies $(-\frac{3}{2}\gamma_k^2-\mathbb{E}_k[Y_k])^2\leq\frac{1}{4}B$ and $(a+b)^2\leq 2(a^2+b^2)$.

Denote $N_1:=\max\{(35+\frac{3}{4})B,\frac{1}{36}((9\gamma^\star)^4+B)+(3\gamma^\star)^2((9\gamma^\star)^2+3\sqrt{B})\}$, inequalities \eqref{eq:stepsizez-ub1}, \eqref{eq:stepsizez-ub2} and \eqref{eq:stepsizez-ub3} then lead to:
\begin{equation}
	\mathbb{E}_k[V(\gamma_{k+1})]-V(\gamma_k)\leq-\eta_k\overline{D}_{\mathsf{lb}}V(\gamma_k)+\eta_k^2 N_1. \label{eq:stepdrift-final}
\end{equation}
\textbf{Step 2: Computing $\mathbb{E}[V(\gamma_k)]$ through iteration:}
Taking the expectation with respect to $\mathcal{H}_{k-1}$ on both sides of \eqref{eq:stepdrift-final}, we have:
\begin{equation}
	\mathbb{E}[V(\gamma_{k+1})]\leq(1-\eta_k\overline{D}_{\mathsf{lb}})\mathbb{E}[V(\gamma_k)]+\eta_k^2N_1. \label{eq:perstep}
\end{equation}
Multiplying inequality \eqref{eq:perstep} from $i=1$ to $k$ yields:
\begin{align}
	\mathbb{E}[V(\gamma_{k+1})]\leq&\prod_{i=1}^k(1-\eta_i\overline{D}_{\mathsf{lb}})V(\gamma_0)\nonumber\\
	&+\sum_{i=1}^k\eta_i^2N_1\cdot\prod_{j=i+1}^k(1-\eta_j\overline{D}_{\mathsf{lb}}). \label{eq:telescope-stepsize}
\end{align}

Since the stepsize selected by \eqref{eq:stepsize} satisfies \[\eta_k\rightarrow 0, \liminf_{k}\min_{n\geq i\geq  m(t_k-T)}\frac{\eta_n}{\eta_i}=1\]
according to \cite[p. 343, Eq. (4.8)]{Kushner2003}, term $\prod_{i=1}^k(1-\eta_i\overline{D}_{\mathsf{lb}})=\mathcal{O}(\eta_k)$. Therefore, 
\begin{equation}
	\sup_k\mathbb{E}\left[\frac{(\gamma_k-\gamma^\star)^2}{\eta_k}\right]=\sup_k\mathbb{E}\left[2V(\theta_k)/\eta_k\right]=\mathcal{O}(1). 
\end{equation}

This finishes the proof of Theorem~\ref{thm:rate-converge}. 

	\section{Proof of Theorem~\ref{thm:converse}}\label{pf:converse}
	
	\subsubsection{Proof of Inequality \eqref{eq:converse-est}}
	
	Let $\mathbb{P}_1, \mathbb{P}_2$ be two delay distributions and let $\gamma_1^\star, \gamma_2^\star$ be the solution to \eqref{eq:equation-offline} when $D\sim\mathbb{P}_1$ and $D\sim\mathbb{P}_2$, respectively. 
	Through Le Cam's inequality \cite{yu1997assouad}, we have:
	\begin{equation}
		\inf_{\hat{\gamma}}\sup_{\mathbb{P}}\mathbb{E}\left[\left(\hat{\gamma}-\gamma^\star(\mathbb{P})\right)^2\right]\geq (\gamma_1^\star-\gamma_2^\star)^2\cdot\left(\mathbb{P}_1^{\otimes  k}\wedge \mathbb{P}_2^{\otimes k}\right),\label{eq:lecam-gamma}
	\end{equation}
	where $\mathbb{P}\wedge \mathbb{Q}:=\int_{\Omega}\min\{p(x), q(x)\}\text{d}x$ and $\mathbb{P}^{\otimes k}$ is the product of distribution of $k$ i.i.d random variables drawn from $\mathbb{P}$. 
	
	To use Le Cam's inequality \eqref{eq:lecam-gamma}, we need to find two distributions $\mathbb{P}_1$ and $\mathbb{P}_2$, whose $\ell_1$ distance $|\mathbb{P}_1^{\otimes k}-\mathbb{P}_2^{\otimes k}|_1$ is bounded, and the difference $(\gamma_1^\star-\gamma_2^\star)^2$ is of order $1/k$. We consider $\mathbb{P}_1$ to be a uniform distribution on $[0, 1]$ and let $\gamma_1^\star$ be the optimum ratio of distribution $\mathbb{P}_1$.  Through Corollary~\ref{coro:gammadep-bound}, we can obtain a loose upper bound on $\gamma_1^\star$ as follows:
	\begin{equation}
		\gamma_1^\star<\frac{1}{2}\frac{\mathbb{E}[D^2]}{\mathbb{E}[D]}=\frac{1}{3}. \label{eq:gamma1lb}
	\end{equation} 
	
	Let $c\leq\frac{1}{2}$ be a constant and we denote
	\begin{equation}
		\delta=\min\{1-3\gamma_1^\star, 1/3, p_{\text{w, uni}}^\star/2\}.\label{eq:deltadef}
	\end{equation} 
	Let $\mathbb{P}_2$ be a probability distribution with probability density function $p_2(x)$ defined as follows:
	\begin{equation}
		p_2(x)=\begin{cases}
			1-c\sqrt{1/k}, &x\leq \frac{1}{2}\delta;\\
			1, &\frac{1}{2}\delta<x\leq 1-\frac{1}{2}\delta;\\
			1+c\sqrt{1/k}, &x>1-\frac{1}{2}\delta;\\
			0, &\text{otherwise}.
		\end{cases}\label{eq:p2def}
	\end{equation}
	
	We will first bound $(\gamma_1^\star-\gamma_2^\star)^2$ (in Step 1) and $\mathbb{P}_1^{\otimes k}\wedge\mathbb{P}_2^{\otimes k}$ (in Step 2) as follows:
	
	\hspace{-10pt}\textbf{Step 1: Lower bounding $\gamma_2^\star-\gamma_1^\star$: }
	For notational simplicity, denote function $h_1(\gamma):=	\mathbb{E}_{D\sim\mathbb{P}_1}[\frac{1}{6}\max\{3\gamma, Z_D^2\}^2-\gamma\max\{3\gamma, Z_D^2\}]$ and $h_2(\gamma):=	\mathbb{E}_{D\sim\mathbb{P}_2}[\frac{1}{6}\max\{3\gamma, Z_D^2\}^2-\gamma\max\{3\gamma, Z_D^2\}]$. According to the definition of $\mathbb{P}_2$ in \eqref{eq:p2def}, for each $\gamma$, the difference between $h_1(\gamma)$ and $h_2(\gamma)$ can be computed by:
	\begin{align}
		&h_2(\gamma)-h_1(\gamma)\nonumber\\
		=&\int_{1-\delta/2}^1\frac{c}{\sqrt{k}}\mathbb{E}\left[\frac{1}{6}\max\{3\gamma, Z_D^2\}^2-\gamma\max\{3\gamma, Z_D^2\}\big|D=x\right]\text{d}x\nonumber\\
		&-\int_{0}^{\delta/2}\frac{c}{\sqrt{k}}\mathbb{E}\left[\frac{1}{6}\max\{3\gamma, Z_D^2\}^2-\gamma\max\{3\gamma, Z_D^2\}\big|D=x\right]\text{d}x\nonumber\\
		\overset{(a)}{=}&\int_{1-\delta/2}^1\frac{c}{\sqrt{k}}\mathbb{E}\left[\frac{1}{6}(Z_D^2-3\gamma)^2\mathbb{I}_{(Z_D^2\geq 3\gamma)}|D=x\right]\text{d}x\nonumber\\
		&-\int_{0}^{\delta/2}\frac{c}{\sqrt{k}}\mathbb{E}\left[\frac{1}{6}(Z_D^2-3\gamma)^2\mathbb{I}_{(Z_D^2\geq 3\gamma)}|D=x\right]\text{d}x, \label{eq:hdiff}
	\end{align}
	where inequality $(a)$ is obtained because 
	\begin{align}&\frac{1}{6}\max\{3\gamma, Z_D^2\}^2-\gamma\max\{3\gamma, Z_D^2\}\nonumber\\
		=&-\frac{3}{2}\gamma^2+\frac{1}{6}(Z_D^2-3\gamma)^2\mathbb{I}_{(Z_D^2\geq 3\gamma)}.
	\end{align}
	Since $\gamma_1^\star$ is the optimum ratio for delay distribution $\mathbb{P}_1$, we have $h_1(\gamma_1^\star)=0$. 
	%
	%
	According to equation \eqref{eq:hdiff}, function $h_2(\gamma_1^\star)$ can be lower bounded by:
	\begin{align}
		&h_2(\gamma_1^\star)\nonumber\\
		\overset{(b)}{\geq}&\frac{c}{\sqrt{k}}\cdot\int_{1-\delta/2}^1\mathbb{E}\left[\frac{1}{6}(Z_D^2-3\gamma_1^\star)^2\mathbb{I}_{(Z_D^2\geq 3\gamma_1^\star)}\big|D=x\right]\text{d}x\nonumber\\
		&-\int_{0}^{\delta/2}\frac{c}{\sqrt{k}}\frac{1}{2}x^2\text{d}x\nonumber\\
		\geq&\frac{c}{\sqrt{k}}\cdot\int_{1-\delta/2}^1\mathbb{E}\left[\frac{1}{6}(Z_D^2-3\gamma_1^\star)^2\mathbb{I}_{(Z_D^2\geq 3\gamma_1^\star)}\big|D=x\right]\text{d}x\nonumber\\
		&-\frac{c}{\sqrt{k}}\frac{1}{6}\left(\frac{\delta}{2}\right)^3. \label{eq:h2mid}
	\end{align}
	where inequality $(b)$ is because $\mathbb{E}[\frac{1}{6}(Z_D^2-3\gamma)^2\mathbb{I}_{(Z_D^2\geq 3\gamma)}|D=x]\leq\mathbb{E}[\frac{1}{6}Z_D^4|D=x]=\frac{1}{2}x^2$. 
	
	We then proceed to lower bound $\mathbb{E}\left[\frac{1}{6}(Z_D^2-3\gamma_1^\star)^2\mathbb{I}_{(Z_D^2\geq 3\gamma_1^\star)}|D=x\right]$ for each delay realization $x\in[1-\delta/2, 1]$ as follows:
	\begin{align}
		&\mathbb{E}\left[\frac{1}{6}(Z_D^2-3\gamma_1^\star)^2\mathbb{I}_{(Z_D^2\geq 3\gamma_1^\star)}|D=x\right]\nonumber\\
		\overset{(c)}{\geq}&\mathbb{E}\left[\frac{1}{6}(Z_D^2-3\gamma_1^\star)^2\mathbb{I}_{(3\gamma_1^\star\leq Z_D^2\leq x)}+\frac{1}{6}(Z_D^2-x)^2\mathbb{I}_{(Z_D^2\geq x)}|D\!=\!x\right]\nonumber\\
		\geq&\mathbb{E}\left[\frac{1}{6}(Z_D^2-3\gamma_1^\star)^2\mathbb{I}_{(3\gamma_1^\star\leq Z_D^2\leq x)}\right]\nonumber\\
		&+\frac{1}{6}\left(\text{Var}[Z_D^2|D=x]-x^2\text{Pr}\left(Z_D^2\leq x|D=x\right)\right)\nonumber\\
		\overset{(d)}{\geq}&\frac{1}{6}x^2\geq\frac{1}{6}(1-\delta/2)^2,\label{eq:h2lb}
	\end{align}
	where inequality $(c)$ is because $\delta\geq1-3\gamma_1^\star$ by equation \eqref{eq:deltadef}, and for the conditional mean $\mathbb{E}[Z_D^2|D=x]=x\geq 1-\delta/2\geq 1-\delta \geq3\gamma_1^\star$; inequality $(d)$ is because $\text{Var}[Z_D^2|D=x]=2x^2$ and $x^2\text{Pr}(Z_D^2\leq x)\leq x^2$ and $\mathbb{E}\left[\frac{1}{6}(Z_D^2-3\gamma_1^\star)^2\mathbb{I}_{(3\gamma_1^\star\leq Z_D^2\leq x)}\right]\geq 0$. 
	Plugging inequality \eqref{eq:h2lb} into \eqref{eq:h2mid} and recall that $\delta<1$ by definition, we have the lower bound of $h_2(\gamma_1^\star)$:
	\begin{align}
		h_2(\gamma_1^\star)\geq&\frac{c}{\sqrt{k}}\frac{\delta}{2}\frac{1}{6}\left(\left(1-\frac{\delta}{2}\right)^2-\left(\frac{\delta}{2}\right)^2\right)\nonumber\\
		\geq&\frac{c}{\sqrt{k}}\frac{\delta}{12}(1-\delta)>0. \label{eq:h2lbf}
	\end{align}
	By Lemma~\ref{lemma:g}-(i), function $h_2(\cdot)$ is monotonically decreasing. Since $h_2(\gamma_1^\star)> 0$ and $h_2(\gamma_2^\star)=0$, we can conclude that $\gamma_2^\star\geq\gamma_1^\star$. We then proceed to bound $\gamma_2^\star-\gamma_1^\star$ through Taylor expansion at $\gamma=\gamma_1^\star$.
	\begin{equation}
		h_2(\gamma_2^\star)=h_2(\gamma_1^\star)+h_2'(\gamma)(\gamma_2^\star-\gamma_1^\star),
	\end{equation}
	where $\gamma\in[\gamma_1^\star, \gamma_2^\star]$. Therefore, $\gamma_2^\star$ can be computed by:
	\begin{equation}
		\gamma_2^\star-\gamma_1^\star=-\frac{h_2(\gamma_1^\star)}{h_2'(\gamma)}. 
	\end{equation}
	%
	
	To lower bound $\gamma_2^\star$, we will first find a loose upper bound of $\gamma_2^\star$ using Lemma~\ref{coro:gammadep-bound}: \begin{equation}\gamma_2^\star\leq\frac{1}{2}\frac{\mathbb{E}_{D\sim\mathbb{P}_2}[D^2]}{\mathbb{E}_{D\sim\mathbb{P}_2}[D]}\leq\frac{1}{2}\left(\frac{1}{3}+\delta \cdot c\sqrt{1/k}\right),\label{eq:gamma2lbloose}
	\end{equation}
	Therefore, since $\delta<1/3$, we have $|h_2'(\gamma)|\leq|h_2'(\gamma_2^\star)|=\mathbb{E}[\max\{3\gamma_2^\star, Z_D^2\}]\leq\overline{D}+3\gamma_{\text{2, ub}}\leq 1+\frac{1}{2}+\frac{3}{2}c\sqrt{\frac{1}{k}}\delta\leq 2$. Then by inequality \eqref{eq:h2lbf}, we have 
	\begin{equation}
		\gamma_2^\star-\gamma_1^\star\geq\frac{-h_2(\gamma_1^\star)}{h_2'(\gamma_2^\star)}\geq\frac{1}{24}(1-\delta)\delta c\sqrt{\frac{1}{k}}. \label{eq:gamma2lb}
	\end{equation}
	
	\hspace{-10pt}\textbf{Step 2: Lower bounding $\mathbb{P}_1^{\otimes k}\wedge \mathbb{P}_2^{\otimes k}$: }Let $|\mathbb{P}-\mathbb{Q}|=\int_{\Omega}|\text{d}\mathbb{P}-\text{d}\mathbb{Q}|$ be the $\ell_1$ distance between probability distribution $\mathbb{P}$ and $\mathbb{Q}$. Then
	\begin{align}
		\mathbb{P}_1^{\otimes k}\wedge \mathbb{P}_2^{\otimes k}=&\int\min\{\mathbb{P}_1^{\otimes k}(\text{d}x), \mathbb{P}_2^{\otimes k}(\text{d}x)\}\nonumber\\
		=&\int\mathbb{P}_1^{\otimes k}(\text{d}x)\cdot\left(1-\frac{\left(\mathbb{P}_2^{\otimes k}(\text{d}x)-\mathbb{P}_1^{\otimes k}(\text{d}x)\right)^+}{\mathbb{P}_1^{\otimes k}(\text{d}x)}\right)\nonumber\\
		=&1-\int\left(\mathbb{P}_2^{\otimes k}(\text{d}x)-\mathbb{P}_1^{\otimes k}(\text{d}x)\right)^+\nonumber\\
		=&1-\frac{1}{2}|\mathbb{P}_1^{\otimes k}-\mathbb{P}_2^{\otimes k}|_1.\label{eq:wedgepq}
	\end{align}
	Equality \eqref{eq:wedgepq} enables us to lower bound $\mathbb{P}_1^{\otimes k}\wedge \mathbb{P}_2^{\otimes k}$ by upper bounding the $\ell_1$ distance $|\mathbb{P}_1^{\otimes k}-\mathbb{P}_2^{\otimes k}|_1$, which can be obtained by the Pinsker's inequality:
	\begin{align}
		&\frac{1}{2}\left|\mathbb{P}_1^{\otimes k}-\mathbb{P}_2^{\otimes k}\right|_1\nonumber\\
		\leq&\sqrt{\frac{1}{2}D_{\mathsf{KL}}(\mathbb{P}_2^{\otimes k}||\mathbb{P}_1^{\otimes k})}\nonumber\\
		=&\sqrt{\frac{1}{2}kD_{\mathsf{KL}}(\mathbb{P}_2||\mathbb{P}_1)}\nonumber\\
		\overset{(e)}{\leq}&\sqrt{\frac{1}{2}k\int_0^1p_2(x)\ln p_2(x)\text{d}x}\nonumber\\
		\overset{(f)}{\leq}&\sqrt{\frac{1}{2}k\int_0^1\left(p_2(x)-1+\frac{1}{\min\{p_2(x), 1\}}(p_2(x)-1)^2\right)\text{d}x}\nonumber\\
		\overset{(g)}{\leq}&\sqrt{\frac{1}{2}k\frac{1}{\inf_{0\leq d\leq 1} p_2(d)}\int_0^1(p_2(x)-1)^2\text{d}x}\nonumber\\
		\leq&\sqrt{\frac{1}{2}k\frac{1}{1-c\sqrt{1/k}}\delta\frac{c^2}{k}}\leq\sqrt{\delta c^2},\label{eq:klub-last}
	\end{align}
	where inequality $(e)$ is because the density function $p_1(x)=1$ for uniform distribution, therefore $D_{\mathsf{KL}}(\mathbb{P}_2||\mathbb{P}_1)=\int_0^1p_2(x)\ln p_2(x)\text{d}x$; inequality $(f)$ is because function $g(t):=(t\ln t)$ is convex, its derivative $g(t)''=1/t$, therefore, through Taylor expansion we have $g(t)\leq g(1)+(t-1)+\frac{1}{2}\frac{1}{\min\{t, 1\}}(t-1)^2=(t-1)+\frac{1}{2}\frac{1}{\min\{t, 1\}}(t-1)^2$; inequality $(g)$ is because $\int_0^1p_2(x)\text{d}x=1$. 
	
	By choosing $c=1/2$ and recall that $\delta<1$, inequality \eqref{eq:klub-last} can be upper bounded by:
	\begin{equation}
		\frac{1}{2}|\mathbb{P}_1^{\otimes k}-\mathbb{P}_2^{\otimes k}|_1\leq\frac{1}{2}. \label{eq:klub}
	\end{equation}
	
	Plugging \eqref{eq:klub} into \eqref{eq:wedgepq} yields:
	\begin{equation}
		\mathbb{P}_1^{\otimes k}\wedge \mathbb{P}_2^{\otimes k}\geq\frac{1}{2}. \label{eq:infub}
	\end{equation}
	
	Finally, plugging \eqref{eq:infub} and \eqref{eq:gamma2lb} into the Le Cam's inequality \eqref{eq:lecam-gamma} finishes the proof of inequality \eqref{eq:converse-est}:
	\begin{equation}
		\inf_{\hat{\gamma}}\sup_{\mathbb{P}}(\hat{\gamma}- \gamma^{\star}(\mathbb{P}))^2\geq\frac{1}{2}\left(\frac{1}{24}(1-\delta)\delta p_{\text{w,uni}}^\star\right)^2\cdot\frac{1}{k}=:N.\label{eq:minimaxf}
	\end{equation} 
	\subsubsection{Proof of Inequality \eqref{eq:regconverse}}
	The proof is divided in to three step: First we decomposite the cumulative MSE gap up to $S_{k+1}$ into the cumulative MSE gap within each frame, and then lower bound the MSE regret in each frame using the difference between frame-length $L_k$ and the optimum frame-length $l^\star(\mathbb{P})$; then we obtain the minimax lower bound of $(\mathbb{E}[L_k]-l^\star(\mathbb{P}))^2$ and finish the proof. 

\noindent\textbf{Step 1}: Cumulative MSE decomposition:
\begin{align}
	&\mathbb{E}\left[\int_0^{S_{k+1}}(X_t-\hat{X}_t)^2\text{d}t-(\gamma^\star+\overline{D})S_{k+1}\right]\nonumber\\
	=&\sum_{k'=1}^k\mathbb{E}\left[\int_{S_{k'}}^{S_{k'+1}}(X_t-\hat{X}_t)^2\text{d}t-(\gamma^\star+\overline{D})L_{k'}\right]\nonumber\\
	=&\sum_{k'=1}^k\bigg(\mathbb{E}\bigg[\int_{S_{k'}}^{S_{k'}+D_{k'}}(X_t-X_{S_{k'-1}})^2\text{d}t\nonumber\\
	&+\int_{S_{k'}+D_{k'}}^{S_{k'+1}}(X_t-X_{S_{k'}})^2\text{d}t-(\gamma^\star+\overline{D})L_{k'}\bigg]\bigg)\nonumber\\
	=&\sum_{k'=1}^k\bigg(\mathbb{E}\bigg[\int_{S_{k'}}^{S_{k'}+D_{k'}}(X_t-X_{S_{k'}}+X_{S_{k'}}-X_{S_{k'-1}})^2\text{d}t\nonumber\\
	&+\int_{S_{k'}+D_{k'}}^{S_{k'+1}}(X_t-X_{S_{k'}})^2\text{d}t-(\gamma^\star+\overline{D})L_{k'}\bigg]\bigg)\nonumber\\
	\overset{(a)}{=}&\sum_{k'=1}^k\bigg(\mathbb{E}\bigg[\int_{S_{k'}}^{S_{k'}+D_{k'}}(X_t-X_{S_{k'}})^2\text{d}t+(X_{S_{k'}}-X_{S_{k'-1}})^2D_{k'}\nonumber\\
	&+\int_{S_{k'}+D_{k'}}^{S_{k'+1}}(X_t-X_{S_{k'}})^2\text{d}t-(\gamma^\star+\overline{D})L_{k'}\bigg]\bigg)\nonumber\\
	=&\sum_{k'=1}^{k}\left(\mathbb{E}\left[\int_{S_{k'}}^{S_{k'+1}}(X_t-X_{S_{k'}})^2\text{d}t-\gamma^\star L_{k'}\right]\right)\nonumber\\
	&+\sum_{k'=1}^k\mathbb{E}\left[(X_{S_{k'}}-X_{S_{k'-1}})^2D_{k'}-\overline{D} L_{k'}\right]\nonumber\\
	\overset{(b)}{=}&\sum_{k'=1}^{k}\left(\mathbb{E}\left[\int_{S_{k'}}^{S_{k'+1}}(X_t-X_{S_{k'}})^2\text{d}t-\gamma^\star L_{k'}\right]\right)\nonumber\\
	&+\sum_{k'=1}^k\left(\mathbb{E}[\overline{D}L_{k'-1}-\overline{D} L_{k'}]\right)\nonumber\\
	=&\sum_{k'=1}^{k}\underbrace{\mathbb{E}\left[\int_{S_{k'}}^{S_{k'+1}}(X_t-X_{S_{k'}})^2\text{d}t-\gamma^\star L_{k'}\right]}_{=:\Upsilon_k}-\overline{D}\mathbb{E}[L_k],\label{eq:decompose}
\end{align}
where equation $(a)$ is because for $\mathbb{E}[(X_t-X_{S_{k'}})(X_{S_{k'}}-X_{S_{k'-1}})]=0$, equation $(b)$ is because $D_{k'}$ is independent of $X_{S_{k'}}-X_{S_{k'-1}}$ and $\mathbb{E}[D_{k'}]=\overline{D}$, $\mathbb{E}[(X_{S_{k'}}-X_{S_{k'-1}})^2]=\mathbb{E}[S_{k'}-S_{k'-1}]=L_{k'-1}$. 

We then proceed to lower bound each item $\Upsilon_k$ in equation \eqref{eq:decompose} using the following Lemma:
\begin{lemma}\label{lemma:regrettoerror}
	For each sample policy $\pi$ with a random sampling interval $\tau$, let $l_\pi:=\mathbb{E}[\tau]=\mathbb{E}[Z_\tau^2]$ denote the expected running length. Recall that $\gamma^\star(\mathbb{P}), l^\star(\mathbb{P})$ are the optimum ratio and optimum frame length when delay distribution $D\sim\mathbb{P}$ and  $p_w(\mathbb{P}):=\text{Pr}\left(Z_D^2\leq3\gamma^\star(\mathbb{P})\right)$ be the probability of waiting
	, the following inequality holds:
	\begin{equation}
		\mathbb{E}\left[\int_{t=0}^{\tau}Z_t^2\text{d}t\right]-\gamma^\star(\mathbb{P})\mathbb{E}[\tau]\geq\frac{1}{6}p_w(\mathbb{P})\left(l_\pi-l^\star(\mathbb{P})\right)^2,\label{eq:regrettoerror}
	\end{equation}
	where $l^\star(\mathbb{P}):=\mathbb{E}_{D\sim\mathbb{P}}[\max\{3\gamma^\star(\mathbb{P}), Z_D^2\}]$ is the average frame length when the optimum policy $\pi^\star(\mathbb{P})$ is used. 
\end{lemma}

Proof for Lemma~\ref{lemma:regrettoerror} is provided in Appendix~\ref{pf:regrettoerror}. Notice that $X_t-X_{S_k}$ is a Wiener Process starting from time $t=S_k$, $\mathcal{H}_{k-1}$ records the previous delay and Wiener process evolution at the beginning of frame $k$. Since the sampling policy in frame $k$ depends on $\mathcal{H}_{k-1}$ and $\delta X_{k}=X_{S_k+R_k}-X_{S_k}$, we can lower bound the worst case regret of $\Upsilon_k$ as follows:
\begin{align}
	&\inf_{\pi}\sup_{\mathbb{P}}\Upsilon_k\nonumber\\
	=&\inf_\pi\sup_{\mathbb{P}}\mathbb{E}\left[\int_{t=S_k}^{S_{k+1}}(X_t-X_{S_{k}})^2\text{d}t-\gamma^\star(\mathbb{P})L_k\right]\nonumber\\
	\geq&\inf_{\pi}\sup_{\mathbb{P}}\frac{1}{6}p_w(\mathbb{P})\mathbb{E}_{\mathcal{H}_{k-1}}\left[\left(\mathbb{E}[L_k|\mathcal{H}_{k-1}]-l^\star(\mathbb{P})\right)^2\right]\nonumber\\
	\geq&\inf_{\pi}\max_{\mathbb{P}\in\{\mathbb{P}_1, \mathbb{P}_2\}}\frac{1}{6}p_w(\mathbb{P})\mathbb{E}_{\mathcal{H}_{k-1}}\left[\left(\mathbb{E}[L_k|\mathcal{H}_{k-1}]-l^\star(\mathbb{P})\right)^2\right]\nonumber\\
	\geq&\frac{1}{6}\underbrace{\min\{p_w(\mathbb{P}_1), p_w(\mathbb{P}_2)\}}_{=:H_1}\nonumber\\
	&\times\underbrace{\inf_{\pi}\max_{\mathbb{P}\in\{\mathbb{P}_1, \mathbb{P}_2\}}\mathbb{E}_{\mathcal{H}_{k-1}}\left[\left(\mathbb{E}[L_k|\mathcal{H}_{k-1}]-l^\star(\mathbb{P})\right)^2\right]}_{=:H_2}. \label{eq:regrettoerrormax}
\end{align}

Inequality \eqref{eq:regrettoerrormax} works for any distribution $\mathbb{P}_1$ and $\mathbb{P}_2$. We select $\mathbb{P}_1$ to be the uniform distribution over interval $[0, 1]$ and $\mathbb{P}_2$ using equation \eqref{eq:p2def}. Then the first term $H_1$ in \eqref{eq:regrettoerrormax} can be lower bounded by:
\begin{align}
	H_1&=\min\{p_w(\mathbb{P}_1), p_w(\mathbb{P}_2)\}\nonumber\\
	&=\min\{\text{Pr}\left(Z_D^2\geq 3\gamma_1^\star|D\sim\mathbb{P}_1\right), \text{Pr}\left(Z_D^2\geq 3\gamma_2^\star|D\sim\mathbb{P}_2\right)\}\nonumber\\
	&\overset{(c)}{\geq}\min\{\frac{\mathbb{E}_{\mathbb{P}_1}[Z_D^2]}{3\gamma_1^\star}, \frac{\mathbb{E}_{\mathbb{P}_2}[Z_D^2]}{3\gamma_2^\star}\}\nonumber\\
	&\overset{(d)}{\geq}\min\{\frac{\mathbb{E}_{\mathbb{P}_1}[D]}{3\times\frac{1}{3}}, \frac{\mathbb{E}_{\mathbb{P}_2}[D]}{3\times\frac{7}{24}}\}\nonumber\\
	&\overset{(e)}{\geq}\min\{1/2, 4/7\}=1/2,\label{eq:h1lb}
\end{align} 
where inequality $(c)$ is by Markov inequality; inequality $(d)$ is because $\mathbb{E}[Z_D^2]=\mathbb{E}[D]$ by the optimal stopping theorem, $\gamma_1^\star\leq\frac{1}{3}$ from \eqref{eq:gamma1lb} and $\gamma_2^\star\leq\frac{1}{2}\left(\frac{1}{3}+\delta\cdot c\sqrt{1/k}\right)\leq\frac{7}{24}$; inequality $(e)$ is because $\mathbb{E}_{\mathbb{P}_1}[D]=1/2$ for uniform distribution $\mathbb{P}_1$ and $\mathbb{E}_{\mathbb{P}_2}[D]\geq\mathbb{E}_{\mathbb{P}_1}[D]=1/2$ due to the distribution of $\mathbb{P}_2$ in equation \eqref{eq:p2def}. It then remains to prove that the second term $H_2$ in \eqref{eq:regrettoerrormax}.

\noindent\textbf{Step 2}: Since $L_k$ is made using $k$ i.i.d samples $\delta X^{\otimes k}=\{\delta X_k=(X_{S_k+D_k}-X_{S_k})\}$, where $\delta X^{\otimes (k-1)}$ are from $\mathcal{H}_{k-1}$ and $\delta X_k=X_{S_k+R_k}-X_{S_k}$, $\mathbb{E}[L_k|\mathcal{H}_{k-1}]$ can be viewed as an deterministic estimator for the corresponding $l^\star(P)$. Let $\hat{l}:\mathbb{R}^k\mapsto\mathbb{R}^+$ an arbitrary deterministic estimation function, term $H_2$ in equation \eqref{eq:regrettoerrormax} is equivalent to:
\begin{align}
	H_2=\inf_{\hat{l}}\max\Big(&\mathbb{E}_{\mathbb{P}_1}[(\hat{l}(\delta X^{\otimes k})-l^\star(\mathbb{P}))^2],\nonumber\\
	& \mathbb{E}_{\mathbb{P}_2}[(\hat{l}(\delta X^{\otimes k})-l^\star(\mathbb{P}))^2]\Big). \label{eq:h2equiv}
\end{align}

To obtain the lower bound of \eqref{eq:h2equiv}, we come up with the following optimization problem:
\begin{pb}{}
	\label{pb:minimax}
	\begin{subequations}
	\begin{align}
		\epsilon^\star:=&\min_{\epsilon, \hat{l}} \epsilon, \\
		\text{s.t., }&\mathbb{E}_{\mathbb{P}_1}[(\hat{l}(\delta X^{\otimes k})-l_1^\star)^2]\leq\epsilon,\\
		&\mathbb{E}_{\mathbb{P}_2}[(\hat{l}(\delta X^{\otimes k})-l_2^\star)^2]\leq\epsilon.
	\end{align}
\end{subequations}
\end{pb}

The minimum $\epsilon^\star$ satisfies:
\begin{equation}
	\epsilon^\star\geq\frac{1}{6}\left(\frac{1}{24}(1-\delta)\delta p_{\text{w, uni}}^\star\right)^2\frac{1}{k}.\label{eq:h2lbfinal}
\end{equation}

Detailed proof is provided in Appendix~\ref{pf:pbminimax}.

\noindent\textbf{Step 3:} Plugging \eqref{eq:h1lb} and \eqref{eq:h2lbfinal} into \eqref{eq:regrettoerrormax}, we have:
\begin{align}
	&\inf_{\pi}\sup_{\mathbb{P}}\Upsilon_k\geq\frac{1}{24}\left(\frac{1}{24}(1-\delta)\delta p_{\text{w, uni}}^\star\right)^2\frac{1}{k}.\label{eq:minimaxl}
\end{align}

Summing up $\Upsilon_k$ from $k'=\{1, 2, \cdots, k\}$ and plugging \eqref{eq:minimaxl} into \eqref{eq:decompose}, we have:
\begin{align}
	&\inf_\pi\mathbb{E}\left[\int_0^{S_{k+1}}(X_t-\hat{X}_t)^2\text{d}t-(\gamma^\star+\overline{D})S_{k+1}\right]\nonumber\\
	\geq&\sum_{k'=1}^k\inf_\pi\sup_\mathbb{P}\Upsilon_k-\overline{D}\mathbb{E}[L_k]\nonumber\\
	=&\sum_{k'=1}^k\inf_\pi\sup_\mathbb{P}\Upsilon_k-\overline{D}(\mathbb{E}[L_k]-l^\star(\mathbb{P}))-\overline{D}l^\star(\mathbb{P})\nonumber\\
	\geq&\frac{1}{24}\left(\frac{1}{24}(1-\delta)\delta p_{\text{w, uni}}^\star\right)^2\times\left(\sum_{k'=1}^k\frac{1}{k'}\right)=\Omega(\ln k). 
\end{align}

\section{Solution to Problem~\ref{pb:minimax}}\label{pf:pbminimax}
	We use the Lagrange method for solving the optimization problem. Let $\rho(\cdot):\mathbb{R}^k\mapsto\mathbb{R}$ and $\lambda_1, \lambda_2\geq 0$ be Lagrange multipliers, the Lagrange function for solving Problem~\ref{pb:minimax} is as follows:
\begin{align}
	\mathcal{L}(\varepsilon, \hat{l}, \lambda_1, \lambda_2)=&\varepsilon+\lambda_1(\mathbb{E}_{\mathbb{P}_1}[(\hat{l}(\delta X^{\otimes k})-l_1^\star)^2]-\varepsilon)\nonumber\\
	&+\lambda_2(\mathbb{E}_{\mathbb{P}_2}[(\hat{l}(\delta X^{\otimes k})-l_2^\star)^2]-\varepsilon). 
\end{align}
The Gâteaux derivative of the Lagrange $\mathcal{L}$ in the direction of $\rho(\cdot):\mathbb{R}^k\mapsto \mathbb{R}$ is defined as
\begin{align}
	&\delta \mathcal{L}(\hat{l}; \varepsilon, \lambda_1, \lambda_2, \rho)\nonumber\\
	:=&\lim_{\epsilon\rightarrow0}\frac{\mathcal{L}(\varepsilon, \hat{l}+\epsilon\rho, \lambda_1, \lambda_2)-\mathcal{L}(\varepsilon, \hat{l}, \lambda_1, \lambda_2)}{\epsilon}\nonumber\\
	=&2\rho(\delta X^{\otimes k})\Big(\lambda_1p_1(\delta X^{\otimes k})(\hat{l}(\delta X^{\otimes k})-l_1^\star)\nonumber\\
	&+\lambda_2 p_2(\delta X^{\otimes k})(\hat{l}(\delta X^{\otimes k})-l_2^\star)\Big). 
\end{align}

Let $(\hat{l}^\star, \varepsilon^\star, \lambda_1^\star, \lambda_2^\star)$ be the dual optimizer. To satisfy the KKT condition, we require:
\begin{subequations}
	\begin{align}
		&\delta\mathcal{L}(\hat{l};\varepsilon^\star, \lambda_1^\star, \lambda_2^\star, \rho)\Big|_{\hat{l}=\hat{l}^\star}=0,\forall \rho, \label{eq:grad0}\\
		&\frac{\partial \mathcal{L}(\hat{l};\varepsilon^\star, \lambda_1^\star, \lambda_2^\star, \rho) }{\partial \varepsilon}\Big|_{\varepsilon=\varepsilon^\star}=1-(\lambda_1^\star+\lambda_2^\star)=0,\label{eq:kkt-1} \\
	\end{align}
and the Complete Slackness (CS) condition require:
\begin{align}
		&\lambda_1^\star\left(\mathbb{E}_{\mathbb{P}_1}\left[\left(\hat{l}^{\star}(\delta X^{\otimes k})-l_1^\star\right)^2\right]-\varepsilon^\star\right)=0, \label{eq:cs-1}\\
		&\lambda_2^\star \left(\mathbb{E}_{\mathbb{P}_2}\left[\left(\hat{l}^\star(\delta X^{\otimes k})-l_2^\star\right)^2\right]-\varepsilon^\star\right)=0,  \label{eq:cs-2}
	\end{align}

The KKT condition in equation \eqref{eq:grad0} implies the optimum estimator $\hat{l}^\star$ is:
\begin{equation}
	\hat{l}^\star(\delta X^{\otimes k})=\frac{\lambda_1^\star p_1(\delta X^{\otimes k})l_1^\star+\lambda_2^\star p_2(\delta X^{\otimes k})l_2^\star}{\lambda_1^\star p_1(\delta X^{\otimes k})+\lambda_2^\star p_2(\delta X^{\otimes k})}, \label{eq:muopt}
\end{equation}
and equation \eqref{eq:kkt-1} requires:
\begin{equation}
	\lambda_1^\star+\lambda_2^\star=1. 
\end{equation}
\end{subequations}
It can be verified that $\lambda_1^\star\neq 0$ and $\lambda_2^\star\neq 0$ because if $\lambda_1^\star=0$, to satisfy equation \eqref{eq:muopt}, we have $\hat{l}^\star(\delta X^{\otimes k})\equiv l_2^\star$. Then $\epsilon^\star=(l_2^\star-l_1^\star)^2$ is clearly not the optimum value. Then for fixed $\lambda_1, \lambda_2$, by plugging function \eqref{eq:muopt} into \eqref{eq:cs-1} and \eqref{eq:cs-2}, we have:
\begin{align}
	\epsilon^\star=&\mathbb{E}_{\mathbb{P}_1}\left[(\hat{l}^\star(\delta X^{\otimes k})-l_1^\star)^2\right]\nonumber\\
	=&(l_2^\star-l_1^\star)^2\int\frac{(\lambda_2^\star p_2(\delta X^{\otimes k}))^2p_1(\delta X^{\otimes k})}{(\lambda_1^\star p_1(\delta X^{\otimes k})+\lambda_2^\star p_2(\delta X^{\otimes k}))^2}\text{d}\delta X^{\otimes k}, \label{eq:eps-1} \\
	\epsilon^\star=&\mathbb{E}_{\mathbb{P}_2}\left[(\hat{l}^\star(\delta X^{\otimes k})-l_2^\star)^2\right]\nonumber\\
	=&(l_2^\star-l_1^\star)^2\int\frac{(\lambda_1^\star p_1(\delta X^{\otimes k}))^2 p_2(\delta X^{\otimes k})}{(\lambda_1^\star p_1(\delta X^{\otimes k})+\lambda_2^\star p_2(\delta X^{\otimes k}))^2}\text{d}\delta X^{\otimes k}. \label{eq:eps-2}
\end{align}

Since $\lambda_1^\star+\lambda_2^\star=1$, \eqref{eq:eps-1} and \eqref{eq:eps-2} imply:
\begin{align}
	\varepsilon^\star&=\lambda_1^\star\varepsilon^\star+\lambda_2^\star\varepsilon^\star\nonumber\\
	&=\lambda_1^\star\mathbb{E}_{\mathbb{P}_1}\left[(\hat{l}^\star(\delta X^{\otimes k})-l_1^\star)^2\right]+\lambda_2^\star\mathbb{E}_{\mathbb{P}_2}\left[(\hat{l}^\star(\delta X^{\otimes k})-l_2^\star)^2\right]\nonumber\\
	&=(l_2^\star-l_1^\star)^2\int\frac{\lambda_1^\star p_1(\delta X^{\otimes k})\times\lambda_2^\star p_2(\delta X^{\otimes k})}{\lambda_1^\star p_1(\delta X^{\otimes k})+\lambda_2^\star p_2(\delta X^{\otimes k})}\text{d}\delta X^{\otimes k}\nonumber\\
	&\overset{(f)}{\geq}(l_2^\star-l_1^\star)^2\int\frac{1}{2}\min\{\lambda_1^\star p_1(\delta X^{\otimes k}), \lambda_2^\star p_2(\delta X^{\otimes k})\}\text{d}\delta X^{\otimes k}\nonumber\\
	&\geq\frac{1}{2}(l_2^\star-l_1^\star)^2\min\{\lambda_1^\star, \lambda_2^\star\}\left(\mathbb{P}_1^{\otimes k}\wedge\mathbb{P}_2^{\otimes k}\right)\label{eq:minimaxbound}. 
\end{align}
where inequality $(f)$ is because $\frac{a\times b}{a+b}\geq\frac{1}{2}\min\{a, b\}$. 

Next, we bound each term in \eqref{eq:minimaxbound} respectively. 

\noindent\textbf{Term 1} The lower bound of $l_2^\star-l_1^\star$ is as follows: \begin{align}
	l_2^\star-l_1^\star
	=&\int_{0}^1\mathbb{E}\left[\max\{3\gamma_2^\star, Z_D^2\}|D=x\right]\text{d}x\nonumber\\
	&+\int_{1-\delta/2}^1\frac{c}{\sqrt{k}}\mathbb{E}\left[\max\{3\gamma_2^\star, Z_D^2\}|D=x\right]\text{d}x\nonumber\\
	&-\int_{0}^{\delta/2}\frac{c}{\sqrt{k}}\mathbb{E}\left[\max\{3\gamma_2^\star, Z_D^2\}|D=x\right]\text{d}x\nonumber\\
	&-\int_{0}^1\mathbb{E}\left[\max\{3\gamma_1^\star, Z_D^2\}|D=x\right]\text{d}x.\label{eq:ldiff}
\end{align}

Notice that if $x_1\geq x_2$, 
\begin{align}
	\mathbb{E}[\max\{3\gamma, Z_D^2\}|D=x_1]-\mathbb{E}[\max\{3\gamma, Z_D^2\}|D=x_2]\geq 0.
\end{align} Therefore, inequality \eqref{eq:ldiff} can be bounded by:
\begin{align}
	l_2^\star-l_1^\star\geq&\int_{0}^1\mathbb{E}\left[\max\{3\gamma_2^\star, Z_D^2\}|D=x\right]\text{d}x\nonumber\\
	&-\int_{0}^1\mathbb{E}\left[\max\{3\gamma_1^\star, Z_D^2\}|D=x\right]\text{d}x\nonumber\\
	\geq&3(\gamma_2^\star-\gamma_1^\star)\mathbb{E}_{D\sim\mathbb{P}_1}[\text{Pr}(Z_D^2\leq 3\gamma_1^\star)]\nonumber\\
	\overset{(g)}{\geq}&\frac{1}{24}(1-\delta)\delta c p_{\text{w, uni}}^\star\sqrt{\frac{1}{k}},\label{eq:llb}
\end{align}
where inequality $(g)$ is obtained by equation \eqref{eq:gamma2lb}. 

\noindent\textbf{Term 2} To lower bound $\min\{\lambda_1, \lambda_2\}$, recall that equation \eqref{eq:eps-1} equals \eqref{eq:eps-2}, we have:
\begin{align}
	&(\lambda_2^\star)^2\int \frac{p_1(\delta X^{\otimes k})\times p_2(\delta X^{\otimes k})}{\lambda_1^\star p_1(\delta X^{\otimes k})+\lambda_2^\star p_2(\delta X^{\otimes k})}p_2(\delta X^{\otimes k})\text{d}\delta X^{\otimes k}\nonumber\\
	=&(\lambda_1^\star)^2\int \frac{p_1(\delta X^{\otimes k})\times p_2(\delta X^{\otimes k})}{\lambda_1^\star p_1(\delta X^{\otimes k})+\lambda_2^\star p_2(\delta X^{\otimes k})}p_1(\delta X^{\otimes k})\text{d}\delta X^{\otimes k}. \label{eq:lambdabd}
\end{align}
Equation \eqref{eq:lambdabd} implies we can upper and lower bound $\lambda_1/\lambda_2$ as follows:
\begin{equation}
	\inf\sqrt{\frac{p_2(\delta X^{\otimes k})}{p_1(\delta X^{\otimes k})}}\leq\frac{\lambda_1}{\lambda_2}\leq\sup\sqrt{\frac{p_2(\delta X^{\otimes k})}{p_1(\delta X^{\otimes k})}}. 
\end{equation}

According to the density function defined in \eqref{eq:p2def}, we have:
\begin{equation}
	\sqrt{1-c\sqrt{\frac{1}{k}}}\leq\frac{\lambda_1}{\lambda_2}\leq\sqrt{1+c\sqrt{\frac{1}{k}}}
\end{equation}

Since $c\leq1/2$, we have $\sqrt{1/2}\leq\frac{\lambda_1}{\lambda_2}\leq\sqrt{3/2}$ and therefore \begin{equation}
	\min\{\lambda_1, \lambda_2\}\geq1/3.\label{eq:lambdamin}
\end{equation}.

Finally, $\left(\mathbb{P}_1^{\otimes k}\wedge\mathbb{P}_2^{\otimes k}\right)\geq1/2$ according to \eqref{eq:infub}. Plugging \eqref{eq:llb} and \eqref{eq:lambdamin} into \eqref{eq:minimaxbound}, we have:
\begin{equation}
	H_2\geq\frac{1}{6}\left(\frac{1}{24}(1-\delta)\delta cp_{\text{w, uni}}^\star\right)^2\frac{1}{k}.
\end{equation}

\section{Proof of Theorem~\ref{thm:mse-as}}\label{pf:mse-as}

Notice that the waiting time $W_k\geq 0, \forall k$, we have:
\begin{equation}
	\liminf_{k\rightarrow\infty}\frac{1}{k}\sum_{k'=1}^k(D_{k'}+W_{k'})\geq\liminf_{k\rightarrow\infty}\frac{1}{k}\sum_{k'=1}^kD_{k'}=\overline{D}>0, \text{w.p.1}.
\end{equation}

Therefore, to show sequence $\{\frac{\int_{0}^{S_{k+1}}(X_t-\hat{X}_t)^2\text{d}t}{S_{k+1}}\}$ converges to $\overline{\mathcal{E}}_{\pi^\star}$ with probability 1, it is sufficient to show that the following sequence
\begin{align}
	\theta_k:=&\frac{1}{k}\int_0^{S_{k+1}}(X_t-\hat{X}_t)^2\text{d}t-(\gamma^\star+\overline{D})S_{k+1}\nonumber\\
	=&\frac{1}{k}\sum_{k'=1}^k\left(\int_{S_{k'}}^{S_{k'+1}}(X_t-\hat{X}_t)^2\text{d}t-(\gamma^\star+\overline{D})L_{k'}\right)
\end{align}
converges to 0 with probability 1.

Recall that $E_{k'}=\int_{S_{k'}}^{S_{k'+1}}(X_t-\hat{X}_t)^2\text{d}t$ is the cumulative error in frame $k'$, we can rewrite $\theta_k$ in the following recursive form:
\begin{align}
	\theta_k=&\frac{1}{k}\left((k-1)\theta_{k-1}+E_k-(\gamma^\star+\overline{D})L_k\right)\nonumber\\
	=&\theta_{k-1}+\frac{1}{k}\left(-\theta_{k-1}+E_k-(\gamma^\star+\overline{D})L_k\right).\label{eq:G-def}
\end{align}

For notational simplicity, denote $G_k:=\left(-\theta_{k-1}+E_k-(\gamma^\star+\overline{D})L_k\right)$, which can be viewed as the descent direction and can be further decomposed into:
\begin{align}
	G_k=&-\theta_{k-1}+\int_{S_{k}}^{S_{k}+D_k}(X_t-X_{S_{k-1}})^2\text{d}t\nonumber\\
	&+\int_{S_k+D_k}^{S_k+D_k+W_k}(X_t-X_{S_k})^2\text{d}t-(\gamma^\star+\overline{D})L_k\nonumber\\
	=&-\theta_{k-1}+\int_{S_k}^{S_k+D_k}(X_t-X_{S_k}+X_{S_k}-X_{S_{k-1}})^2\text{d}t\nonumber\\
	&+\int_{S_k}^{S_{k+1}}(X_t-X_{S_k})^2\text{d}t-(\gamma^\star+\overline{D})L_k\nonumber\\
	=&-\theta_{k-1}+\underbrace{(X_{S_k}-X_{S_{k-1}})^2D_k}_{=:G_{k, 1}}\nonumber\\
	&+2\underbrace{(X_{S_k}-X_{S_{k-1}})\cdot\int_{S_k}^{S_k+D_k}(X_t-X_{S_k})\text{d}t}_{=:G_{k, 2}}\nonumber\\
	&+\underbrace{\int_{S_k}^{S_{k+1}}(X_t-X_{S_k})^2\text{d}t}_{=:G_{k, 3}}-\underbrace{(\gamma^\star+\overline{D})L_k}_{=:G_{k, 4}}.\label{eq:theta-recurse}
\end{align}

Give historical transmissions $\mathcal{H}_{k-1}$, $\gamma_k$ can be predicted and $X_{S_k}-X_{S_{k-1}}$ is fixed, $X_t-X_{S_k}$ evolves like a Wiener process and is independent of $X_{S_k}-X_{S_{k-1}}$. Therefore, the conditional mean of $G_{k, 1}, \cdots, G_{k, 4}$ can be computed as follows:
\begin{subequations}
	\begin{align}
		\mathbb{E}_k\left[G_{k, 1}\right]=&\overline{D}(X_{S_k}-X_{S_{k-1}})^2,\label{eq:g1}\\
		\mathbb{E}_k\left[G_{k, 2}\right]=&0,\label{eq:g2}\\
		\mathbb{E}_k\left[G_{k, 3}\right]=&\frac{1}{6}\mathbb{E}_k\left[\max\{3\gamma_k, Z_D^2\}^2\right]=q(\gamma_k),\label{eq:g3}\\
		\mathbb{E}_k\left[G_{k, 4}\right]=&(\gamma^\star+\overline{D})\mathbb{E}_k\left[\max\{3\gamma_k, Z_D^2\}\right]=(\gamma^\star+\overline{D})l(\gamma_k).\label{eq:g4}
	\end{align}
\end{subequations}
where equation \eqref{eq:g1} is because $D_k$ is independent of $X_{S_k}-X_{S_{k-1}}$; equation \eqref{eq:g2} is because $X_t-X_{S_k}$ is independent of $X_{S_k}-X_{S_{k-1}}$ and has mean 0 for all $t\geq S_k$; equation \eqref{eq:g3} and \eqref{eq:g4} is because of Lemma~\ref{lemma:cond-l}. With equation \eqref{eq:g1}-\eqref{eq:g4}, given historical transmissions $\mathcal{H}_{k-1}$, we can compute the conditional expectation of $G_k$ as follows:
\begin{align}
	&\mathbb{E}_k[G_k]\nonumber\\
	=&\mathbb{E}_k\left[-\theta_{k-1}+G_{k, 1}+2G_{k, 2}+G_{k, 3}-G_{k, 4}\right]\nonumber\\
	=&-\theta_{k-1}+(X_{S_k}-X_{S_{k-1}})^2\overline{D}+q(\gamma_k)-(\gamma^\star+\overline{D})l(\gamma_k)\nonumber\\
	=&-\theta_{k-1}+q(\gamma_k)-\gamma_k l(\gamma_k)+\underbrace{\overline{D}\left(l(\gamma_{k-1})-l(\gamma_{k})\right)}_{=:\beta_{k, 1}}\nonumber\\
	&+\underbrace{\overline{D}\left(\left(X_{S_k}-X_{S_{k-1}}\right)^2-l(\gamma_{k-1})\right)}_{=:\beta_{k, 2}}+\underbrace{(\gamma_k-\gamma^\star)l(\gamma_k)}_{=:\beta_{k, 3}}.\label{eq:beta-def}
\end{align}

Denote function
\begin{equation}
	f(\theta, \gamma):=-\theta+\mathbb{E}\left[\frac{1}{6}\max\{3\gamma, Z_D^2\}^2-\gamma\max\{3\gamma, Z_D^2\}\right],\label{eq:f}
\end{equation}
and let function $\overline{f}(\cdot)$ be:
\begin{equation}
	\overline{f}(\theta):=f(\theta, \gamma^\star). \label{eq:fbar}
\end{equation}

In the following analysis, we will prove that sequence $\{\theta_k\}$ converges to the stationary point of an ODE induced by function $\overline{f}(\theta)$. Let $\delta M_{k}:=G_k-\mathbb{E}_k[G_{k}]$ and let $\delta M_{k, i}:=G_{k, i}-\mathbb{E}_k\left[G_{k, i}\right]$ be the difference between each term and their conditional mean. We view $\frac{1}{k}=:\epsilon_k$ as the updating step-sizes, which satisfies:
\begin{equation}
	\sum_k \epsilon_k=\infty, \sum_k\epsilon_k^2<\infty. 
\end{equation}
With $\epsilon_k, \beta_{k, 1}, \beta_{k, 2}$ and $\delta M_k$, the recursive equation \eqref{eq:theta-recurse} can be rewritten as follows:
\begin{equation}
	\theta_k=\theta_{k-1}+\epsilon_k\left(f(\theta_{k-1}, \gamma_k)+\beta_{k, 1}+\beta_{k,2}+\beta_{k, 3}+\delta M_k\right). 
\end{equation}

Similarly, denote $t_0=0$ and $t_k:=\sum_{i=0}^{k-1}\epsilon_i$ to be the cumulative step-size sequences. Let $m(t)$ be the unique $k\in\mathbb{N}^+$ such that $t_{m(t)}\leq t<t_{m(t)}+1$. We then state the following characteristics of $G_k$ and $\delta M_k$, detailed proofs are in  Appendix~\ref{pf:claim}:

\begin{claim}\label{claim:thm2}Sequences $\{G_k\}$ and $\{\delta M_k\}$ have the following properties:
	
\end{claim}
\hspace{-10pt}\textbf{(2.1)} For each constant $N$,  $\sup_k\mathbb{E}\left[|G_k|\mathbb{I}_{(|\theta_k|\leq N)}\right]<\infty$.

\hspace{-10pt}\textbf{(2.2)} Function $f(e, \gamma)$ is continuous in $e$ for each $\gamma$. 

\hspace{-10pt}\textbf{(2.3)} For any $T>0$, the following limit hold for all $\theta$:
\begin{align}
	&\lim_{k\rightarrow\infty}{\rm Pr}\left(\sup_{j\geq k}\max_{0\leq t\leq T}\left|\sum_{i=m(jT)}^{m(jT+t)-1}\epsilon_i\left(f(\theta, \gamma_i)-\overline{f}(\theta)\right)\right|\geq \mu\right)\nonumber\\
	&\hspace{6.5cm}=0.\label{eq:mse-claim2} 
\end{align}

\hspace{-10pt}\textbf{(2.4)} For any $T>0$, the difference sequence satisfies:
\begin{equation}
	\lim_{k\rightarrow\infty}{\rm Pr}\left(\sup_{j\geq k}\max_{0\leq t\leq T}\left|\sum_{i=k}^j\epsilon_i\delta M_i\right|\geq \mu\right)=0. 
\end{equation}

\hspace{-10pt}\textbf{(2.5)} The bias sequence satisfies:
\begin{align}
	&\lim_{k\rightarrow\infty}{\rm Pr}\left(\sup_{j\geq k}\max_{0\leq t\leq T}\left|\sum_{i=m(jT)}^{m(jT+t)-1}\epsilon_i(\beta_{i, 1}+\beta_{i, 2}+\beta_{i, 3})\right|\geq \mu\right)\nonumber\\
	&\hspace{6.5cm}=0. \label{eq:martingale-claim2}
\end{align}

\hspace{-10pt}\textbf{(2.6)} For each $\theta$, function $f$ can be bounded as follows:
\begin{equation}
	f(\theta, \gamma)=\overline{f}(\theta)+\rho(\gamma),
\end{equation}
where $\rho(\gamma)=-(q(\gamma)-\gamma l(\gamma))$ and for any $\tau>0$ we have the following inequality:
\begin{equation}
	\lim_{k\rightarrow\infty}{\rm Pr}\left(\sup_{j\geq n}\sum_{i=m(j\tau)}^{m(j\tau+\tau)-1}\left|\epsilon_i\rho(\gamma_k)\right|\right)=0. 
\end{equation}

\hspace{-10pt}\textbf{(2.7)} For each $\theta_1, \theta_2$, the difference
\begin{equation}
	\left|f(\theta_1, \gamma)-f(\theta_2, \gamma)\right|=\left|\theta_1-\theta_2\right|.
\end{equation}
When $\theta_1-\theta_2\rightarrow0$, the absolute difference $|\theta_1-\theta_2|\rightarrow 0$. 

	%
	%
	
	Denote $\theta_k(\omega)$ as the time averaged MSE up to frame $k$ of sample path $\omega$. Then according to \cite[p.166, Theorem 1.1]{Kushner2003}, with probability 1, sequence $\{\theta_k(\omega)\}$ converges to some limit set of the ODE
	\begin{equation}
		\dot \theta=\overline{f}(\theta)=-\theta.\label{eq:ode-mse}
	\end{equation}
	
	Because $\overline{f}(0)=0$, the minimum error $\theta=0$ is an equilibrium point of the ODE in equation~\eqref{eq:ode-mse}. Moreover, as $\overline{f}(\cdot)$ is a monotonic decreasing function, it can be easily verified through Lyapunov stability criterion that $0$ is a unique stability point of the ODE \eqref{eq:ode-mse}. Therefore, $\theta_k$ converges to 0 with probability 1, and the time averaged MSE converges to $\overline{\mathcal{E}}_{\pi^\star}$ with probability 1. 

	\section{Proof of Claim~\ref{claim:thm2}}\label{pf:claim}
	Before we starts to prove each condition in Claim~\ref{claim:thm2}, we provide the following corollary from Theorem \ref{thm:rate-converge}:
	\begin{corollary}
		There exists a $\Gamma<\infty$ so that $\mathbb{E}[(\gamma_k-\gamma^\star)^2/\eta_k]<\Gamma, \forall k$. Recall that the step-sizes is selected to be $\eta_k=\frac{1}{4 D_{\text{lb}}k^\alpha}$, where $\alpha\in(0.5, 1]$, we then have:
		\begin{align}
			&\mathbb{E}[(\gamma_k-\gamma^\star)^2]\leq\frac{D_{\mathsf{lb}}\Gamma}{2k^\alpha}<\infty,\\ &\mathbb{E}[\gamma_k^2]\leq2(\mathbb{E}[(\gamma_k-\gamma^\star)^2]+(\gamma^\star)^2)<\infty.\label{eq:secondorderbd}
		\end{align}
		Through Cauchy-Schwarz inequality, we have:
		\begin{equation}
			\mathbb{E}\left[|\gamma_k-\gamma^\star|\right]\leq\sqrt{\mathbb{E}[(\gamma_k-\gamma^\star)^2]}\leq\sqrt{\frac{D_{\mathsf{lb}}\Gamma}{2k^\alpha}}. \label{eq:absub}
		\end{equation}
	\end{corollary}
	\hspace{-10pt}\textbf{(2.1): }According to the definition of $G_k$ from equation~\eqref{eq:G-def}, the expectation $\mathbb{E}\left[|G_k|\mathbb{I}_{|\theta_k|\leq N}\right]$ can be upper bound as follows:
	\begin{align}
		&\mathbb{E}\left[|G_k|\mathbb{I}_{(|\theta_k|\leq N)}\right]\nonumber\\
		\leq&\mathbb{E}\left[|\theta_k|\mathbb{I}_{(|\theta_k|\leq N)}\right]+\mathbb{E}\left[\int_{t=S_k}^{S_{k+1}}(X_t-\hat{X}_t)^2\text{d}t\right]+\mathbb{E}\left[\gamma_kL_k\right]. \label{eq:supg}
	\end{align}
	
	The first term on the RHS of inequality \eqref{eq:supg} satisfies
	\begin{equation}\sup_k\mathbb{E}\left[|\theta_k|\mathbb{I}_{(|\theta_k|\leq N)}\right]\leq N<\infty.\label{eq:as-1}
	\end{equation}
	
	The expectation of the second term can be computed as follows,
	\begin{align}
		&\mathbb{E}\left[\int_{t=S_k}^{S_{k+1}}(X_t-\hat{X}_t)^2\text{d}t\right]\nonumber\\
		=&\mathbb{E}\left[\int_{t=S_k}^{S_k+D_k}(X_t-X_{S_k}+X_{S_k}-X_{S_{k-1}})^2\text{d}t\right.\nonumber\\
		&\left.+\int_{S_k+D_k}^{S_k+D_k+W_k}(X_t-X_{S_k})^2\text{d}t\right]\nonumber\\
		=&\mathbb{E}\left[(X_{S_{k}}-X_{S_{k-1}})^2D_k\right]\nonumber\\
		&+2\mathbb{E}\left[(X_{S_k}-X_{S_{k-1}})\int_{t=S_k}^{S_k+D_k}(X_t-X_{S_k})\text{d}t\right]\nonumber\\
		&+\mathbb{E}\left[\int_{t=S_k}^{S_{k+1}}(X_t-X_{S_k})^2\text{d}t\right]\nonumber\\
		=&\mathbb{E}\left[\max\{3\gamma_{k-1}, Z_D^2\}\right]\overline{D}+\frac{1}{6}\mathbb{E}\left[\max\{3\gamma_k, Z_D^2\}^2\right]\nonumber\\
		\leq&\mathbb{E}[3\gamma_{k-1}]\overline{D}+\frac{1}{6}\mathbb{E}[(3\gamma_k)^2]+\overline{D}B^{1/4}+\frac{1}{2}\sqrt{B}.\label{eq:as-2}
	\end{align}
	Inequality \eqref{eq:secondorderbd} and \eqref{eq:absub} implies inequality \eqref{eq:as-2} is bounded for all $k$. Therefore, the second term on the RHS of inequality \eqref{eq:supg} can be upper bounded as follows:
	\begin{align}
		&\sup_{k}\mathbb{E}\left[\int_{t=S_k}^{S_{k+1}}(X_t-\hat{X}_t)^2\text{d}t\right]\nonumber\\
		=&\sup_k\left(\mathbb{E}\left[\max\{3\gamma_{k-1}, Z_D^2\}\right]\overline{D}+\frac{1}{6}\mathbb{E}\left[\max\{3\gamma_k, Z_D^2\}\right]\right)\nonumber\\
		<&\infty.\label{eq:as-3}
	\end{align}
	
	Similarly, since $\mathbb{E}[\gamma_k^2]<\infty$ is bounded by \eqref{eq:secondorderbd} is bounded, we can upper bound the third term on the RHS of inequality \eqref{eq:supg} as follows:
	\begin{equation}
		\sup_k\mathbb{E}\left[\gamma_kL_k\right]=\sup_k\mathbb{E}\left[\gamma_k\max\{3\gamma_{k}, Z_D^2\}\right]<\infty.
	\end{equation}
	
	Taking the supremum of inequality \eqref{eq:supg} and then plugging equality~\eqref{eq:as-1}-\eqref{eq:as-3} into the inequality verifies Claim (2.1).
	
	Notice that statement (2.3)-(2.7) has similar forms, 	\begin{equation}
		\lim_{k\rightarrow\infty}\text{Pr}\left(\sup_{j\geq k}\left|\sum_{i=k}^j\epsilon_i\psi_i\right|\geq \mu\right)=0. \label{eq:universal}
	\end{equation}where $\psi_k$ can be the bias term $\beta_{k, i}$, the martingale sequence $\delta M_k$ or the difference $f(\theta, \gamma_k)-\overline{f}(\theta)$ and $\rho(\gamma_k)$. We then provide the following lemma:
	\begin{lemma}\label{lemma:limsupsufficient}
		If one of the following condition holds, then \eqref{eq:universal} holds:
		
		\textbf{(S.1)} $\psi_k$ is a martingale sequence and $\sup_k\mathbb{E}[\psi_k^2]<\infty$. The correlation satisfies $\mathbb{E}[\psi_i\psi_j]=0, \forall i\neq j$.
		
		\textbf{(S.2)} $\mathbb{E}[|\psi_k|]=\mathcal{O}(k^{-\zeta}), \zeta>0$. 
	\end{lemma}
	\begin{proof}
		If condition (S.1) holds, since $\epsilon_k=\frac{1}{k}$ satisfies $\sum_k\epsilon_k^2<\infty$, equality \eqref{eq:universal} holds because of \cite[p. 172, example 3]{Kushner2003}. 
		
		If condition (S.2) holds, there exists a $\Psi$ so that $\mathbb{E}[\psi_k]=\Psi k^{-\zeta}$. For each $\mu>0$, we first upper bound $\text{Pr}\left(\sup_{j\geq k}\left|\sum_{i=k}^j\epsilon_i\psi_i\right|\geq \mu\right)$ for each $k$ as follows:
		\begin{align}
			&\text{Pr}\left(\sup_{j\geq k}\left|\sum_{i=k}^j\epsilon_i\psi_i\right|\geq \mu\right)
			\leq\text{Pr}\left(\sum_{i=k}^\infty\epsilon_i|\psi_i|\geq \mu\right)\nonumber\\
			\overset{(a)}{\leq}&\frac{1}{\mu}\mathbb{E}\left[\sum_{i=k}^{\infty}i^{-1} \left|\psi_i\right|\right]
			\overset{(b)}{\leq}\frac{\Psi}{\mu}\left(\sum_{i=k}^\infty i^{-1-\zeta}\right)=\frac{\Psi}{\mu\zeta}(k-1)^{-\zeta}.\label{eq:s2}
		\end{align}
		where inequality $(a)$ is from the Markov inequality; inequality $(b)$ is from statement (S.2). Finally, taking the limit of \eqref{eq:s2} yields \eqref{eq:universal}. 
	\end{proof}
	
	\hspace{-10pt}\textbf{(2.2): }Since function $f(\theta, \gamma;\delta X):=-\theta+\frac{1}{6}\max\{3\gamma, \delta X^2\}^2-\gamma\max\{3\gamma, \delta X^2\}$ is continuous for each $\delta X$, the expectation $\overline{f}(\theta)=\mathbb{E}[f(\theta, \gamma^\star;\delta X^2)]$ is continuous for $\theta$. 
	
	\hspace{-10pt}\textbf{(2.3): }Recall the definition of $f(\theta, \gamma)$ and $\overline{f}(\gamma)$ from equation~\eqref{eq:f}, \eqref{eq:fbar}. The absolute difference between $f(\theta, \gamma)$ and $\overline{f}(\theta)$ can be upper bounded by:
	\begin{align}
		&\left|f(\theta, \gamma)-\overline{f}(\theta)\right|=\left|\overline{g}_0(\gamma)-\overline{g}_0(\gamma^\star)\right|\nonumber\\
		\overset{(a)}\leq&3(\gamma-\gamma^\star)^2+3|\gamma-\gamma^\star|\overline{D},\label{eq:rhoub}
	\end{align}
	where inequality $(a)$ is because function $\overline{g}_0(\gamma)$ is concave and $|\overline{g}''_0(\gamma)|<3$ according to Lemma~\ref{lemma:g}-(i).
	Therefore $\mathbb{E}[|f(\theta, \gamma)-\overline{f}(\theta)|]=\mathcal{O}(k^{-\alpha/2})$, which satisfies statement (S.2) in Lemma~\ref{lemma:limsupsufficient}. This verifies inequality \eqref{eq:mse-claim2}.

	\hspace{-10pt}\textbf{(2.4): }The difference $\delta M_k=\delta M_{k, 1}+2\delta M_{k, 2}+\delta M_{k, 3}-\delta M_{k, 4}$ consists of four parts. Through the union bound, the probability that $\sup_{j\geq k}\left|\sum_{i=k}^j\epsilon_i\delta M_{i}\right|\geq \mu$ can be upper bounded by:
	\begin{align}
		&\lim_{k\rightarrow\infty}\text{Pr}\left(\sup_{j\geq k}\left|\sum_{i=k}^j\epsilon_i\delta M_i\right|\geq \mu\right)\nonumber\\
		\leq& \sum_{a=1}^4\lim_{k\rightarrow\infty}\text{Pr}\left(\sup_{j\geq k}\left|\sum_{i=k}^j\epsilon_i\delta M_{i, a}\right|\geq \mu/5\right).\label{eq:c4-1}
	\end{align}
	
	We will then show that each item on the RHS of inequality \eqref{eq:c4-1} has limit 0. The first term $\delta M_{k, 1}=(X_{S_k}-X_{S_{k-1}})^2\left(D_k-\overline{D}\right)$. Since $D_k-\overline{D}$ depends only on the delay in frame $k$ and has mean zero, term $\mathbb{E}[\delta M_{k, 1}\delta M_{k+i, 1}]=0, \forall i>0$. The second moment of $\delta M_{k, 1}$ can be upper bounded as follows:
	\begin{align}
		&\mathbb{E}\left[(X_{S_k}-X_{S_{k-1}})^4(D_k-\overline{D})^2\right]\nonumber\\
		=&\mathbb{E}\left[(X_{S_k}-X_{S_{k-1}})^4\right]\text{Var}[D^2]
		\leq\mathbb{E}\left[\max\{3\gamma_k, Z_D^2\}^2\right]\mathbb{E}[D^2].
	\end{align} 
	
	By Theorem \ref{thm:rate-converge}, the expectation $\mathbb{E}[(\gamma_k-\gamma^\star)^2/\eta_k]=\mathcal{O}(1)$. Since $\eta_k\rightarrow 0$, $\mathbb{E}[\gamma_k^2]$ is bounded. Therefore, $\sup_{k}\mathbb{E}\left[\delta M_{k, 1}^2\right]\leq\infty$. Then according to \cite[p.142, Eq.~(5.3.18)]{Kushner2003}
	\begin{equation}
		\lim_{k\rightarrow\infty}\text{Pr}\left(\sup_{j\geq k}\left|\sum_{i=k}^j\epsilon_i\delta M_{i, 1}\right|\geq \mu/5\right)=0. \label{eq:c4-2}
	\end{equation}
	
	Similarly, recall that $\delta M_{k, 2}=(X_{S_k}-X_{S_{k-1}})\cdot\left(\int_{S_k}^{S_k+D_k}(X_t-X_{S_k})\text{d}t\right)$. Sequence $\delta M_{k, 2}$ is a martingale sequence with mean zero. Moreover, $\mathbb{E}[M_{k, 2}M_{k+i, 2}]=0, \forall i\geq 1$. The variance $\text{Var}[\delta M_{k, 2}]$ can be bounded as follows:
	\begin{align}
		&\text{Var}[\delta M_{k, 2}]
		=\mathbb{E}[\delta M_{k, 2}^2]\nonumber\\			
		=&\mathbb{E}\left[(X_{S_k}-X_{S_{k-1}})^2\right]\cdot\mathbb{E}\left[\left(\int_{t=S_k}^{S_k+D_k}(X_t-X_{S_k})\text{d}t\right)^2\right]\nonumber\\
		=&\mathbb{E}\left[\max\{3\gamma_{k-1}, Z_D^2\}\right]\cdot\mathbb{E}[D^2].
	\end{align}
	
	Inequality \eqref{eq:absub} upper bounds $\mathbb{E}[\max\{3\gamma_{k-1}, Z_D^2\}]$ and verifies (S.1) in Lemma~\ref{lemma:limsupsufficient}. Therefore, we have:
	\begin{equation}
		\lim_{k\rightarrow\infty}\text{Pr}\left(\sup_{j\geq k}\left|\sum_{i=k}^j\epsilon_i\delta M_{i, 2}\right|\geq \mu/5\right)=0. \label{eq:c4-3}
	\end{equation}
	
	It can be verified that the sequence $\delta M_{k, 3}=\int_{S_k}^{S_{k+1}}(X_t-X_{S_k})^2\text{d}t-\mathbb{E}_k\left[\int_{S_k}^{S_{k+1}}(X_t-X_{S_k})^2\text{d}t\right]$ is a martingale sequence. It then remains to upper bound its variance, which is as follows:
	\begin{align}
		&\text{Var}[\delta M_{k, 3}]\nonumber\\
		=&\mathbb{E}\left[\left(\int_{S_k}^{S_k+D_k}(X_t-X_{S_k})^2\text{d}t-\int_{S_k+D_k}^{S_{k+1}}(X_t-X_{S_k})^2\text{d}t\right)^2\right]\nonumber\\
		\overset{(c)}{\leq}&2\mathbb{E}\left[\left(\int_{S_k}^{S_k+D_k}(X_t-X_{S_k})^2\text{d}t\right)^2\right]\nonumber\\
		&+2\mathbb{E}\left[\left(\int_{S_k+D_k}^{S_{k+1}}(X_t-X_{S_k})^2\text{d}t\right)^2\right]\nonumber\\
		\overset{(d)}{\leq}&2\mathbb{E}\left[\left(\int_{S_k}^{S_k+D_k}(X_t-X_{S_k})^2\text{dt}\right)^2\right]\nonumber\\
		&+2\mathbb{E}\left[\text{Pr}(Z_D^2\geq 3\gamma_k)(3\gamma_k)^2\mathbb{E}_k[L_k^2]\right]\nonumber\\
		\overset{(e)}{\leq}&N_1+\mathbb{E}\left[\mathbb{E}[Z_D^4](\frac{10}{3}(3\gamma_k)^2+3\sqrt{B})\right],\label{eq:mk3}
	\end{align}
	where inequality $(c)$ is because $\mathbb{E}[(a-b)^2]\leq 2\mathbb{E}[a^2+b^2]$; inequality $(d)$ is because if $L_k\geq D_k$, then $(X_t-X_{S_k})^2\leq 3\gamma_k$ for $t\in[S_k+D_k, S_{k+1}]$; inequality $(e)$ is because through Markov inequality $\text{Pr}(Z_D^2\geq 3\gamma_k)\leq \mathbb{E}[Z_D^4]/(3\gamma_k)^2$ and $\mathbb{E}_k[L_k^2]\leq (\frac{10}{3}(3\gamma_k)^2+3\sqrt{B})$ from Lemma~\ref{lemma:4order}. Since $\mathbb{E}[Z_D^4]\leq 3\mathbb{E}[D^2]<3\sqrt{B}$ and $\mathbb{E}[\gamma_k^2]$ is bounded according to inequality \eqref{eq:secondorderbd}, $
	\text{Var}[\delta M_{k, 3}]$ is bounded according to inequality \eqref{eq:mk3}. Condition (S.1) in Lemma~\ref{lemma:limsupsufficient} is satisfied and we have
	\begin{equation}
		\lim_{k\rightarrow\infty}\text{Pr}\left(\sup_{j\geq k}\left|\sum_{i=k}^j\epsilon_i\delta M_{i, 3}\right|\geq \mu/5\right)=0. \label{eq:c4-4}
	\end{equation}
	
	Following similar approaches, the second order expansion of the fourth term is bounded, i.e., 
	\begin{equation}
		\text{Var}[\delta M_{k, 4}]\leq\mathbb{E}_k[G_{k, 4}]\leq(\gamma^\star+\overline{D})^2\cdot\frac{10}{3}\left(3\gamma_k^2+3\sqrt{B}\right). 
	\end{equation}
	Again using Lemma~\ref{lemma:limsupsufficient} condition (S.1), we have:
	\begin{equation}
		\lim_{k\rightarrow\infty}\text{Pr}\left(\sup_{j\geq k}\left|\sum_{i=k}^j\epsilon_i\delta M_{i, 4}\right|\geq \mu/5\right)=0. \label{eq:c4-5}
	\end{equation}
	
	Plugging inequalities \eqref{eq:c4-2}, \eqref{eq:c4-3}, \eqref{eq:c4-4} and \eqref{eq:c4-5} into \eqref{eq:c4-1} completes the proof of \eqref{eq:martingale-claim2}. 
	
	\hspace{-10pt}\textbf{(2.5): } Through the union bound we have:
	\begin{align}
		&\lim_{k\rightarrow\infty}\text{Pr}\left(\sup_{j\geq k}\left|\sum_{i=k}^j\epsilon_i(\beta_{i, 1}+\beta_{i, 2}+\beta_{i, 3})\right|\geq \mu\right)\nonumber\\
		\leq&\sum_{a=1}^3		\lim_{k\rightarrow\infty}\text{Pr}\left(\sup_{j\geq k}\left|\sum_{i=k}^j\epsilon_i\beta_{i, a}\right|\geq \mu/3\right).\label{eq:btot}
	\end{align}
	
	For simplicity, define event 
	\[\mathcal{A}_{a, k}\triangleq\sup_{j\geq k}\left|\sum_{i=k}^j\epsilon_i\beta_{i, a}\right|\geq \mu/3.\]
	
	We then upper bound the probability $\text{Pr}(\mathcal{A}_{a,k})$ and analyzing their asymptotic performance. 
	
	To upper bound event $\mathcal{A}_{1, k}$, we need to upper bound the expectation of $\beta_{k, 1}$ defined in \eqref{eq:beta-def} as follows:
	\begin{align}&\mathbb{E}\left[|\beta_{k, 1}|\right]\nonumber\\
		\overset{(f)}{=}&\overline{D}\mathbb{E}\left[\left|\mathbb{E}_k[3\gamma_{k-1}, Z_D^2]-\mathbb{E}_k[3\gamma_{k}, Z_D^2]\right|\right]\nonumber\\
		=&\overline{D}\left(\mathbb{E}[3\gamma_{k-1}, Z_D^2]-\mathbb{E}[3\gamma^\star, Z_D^2]+\mathbb{E}[3\gamma^\star, Z_D^2]-\mathbb{E}[3\gamma_{k}, Z_D^2]\right)\nonumber\\
		\leq&3\overline{D}\mathbb{E}\left[|\gamma_{k-1}-\gamma^\star|+|\gamma_k-\gamma^\star|\right]\nonumber\\
		\overset{(g)}{=}&\mathcal{O}(k^{-\alpha}),
	\end{align}
	where equality $(f)$ is from definition \eqref{eq:beta-def}, equality $(g)$ is from inequality \eqref{eq:absub}. 
	Since $\alpha\in(0.5, 1]$, which satisfies condition (S.2) in Lemma \ref{lemma:limsupsufficient}. We have:
	
	\begin{equation}
		\lim_{k\rightarrow\infty}\text{Pr}\left(\sup_{j\geq k}\left|\sum_{i=k}^j\epsilon_i\beta_{i,1}\right|\geq \mu/3\right)=0. \label{eq:b1}
	\end{equation}
	
	Next we upper bound $\text{Pr}(\mathcal{A}_{2, k})$ and analyzing its asymptotic behavior. Variable $\beta_{k, 2}$ has mean zero because  \begin{align}\mathbb{E}[\beta_{k, 2}]=&\mathbb{E}\left[\mathbb{E}_k\left[(X_{S_{k-1}+D_{k-1}}-X_{S_{k-1}})^2\right]-l(\gamma_{k-1})\right]\nonumber\\
		=&\max\{3\gamma_{k-1}, Z_D^2\}-l(\gamma_{k-1})=0.
	\end{align} The variance of $\beta_{k, 2}$ is upper bounded by 
	\begin{align}
		&\text{Var}\left[\overline{D}\left((X_{S_k}-X_{S_{k-1}})^2-l(\gamma_{k-1})\right)\right]\nonumber\\
		=&\overline{D}^2\mathbb{E}\left[\max\{3\gamma_{k-1}, Z_D^2\}^2\right]\overset{(h)}{<}\infty,
	\end{align}
	where inequality $(h)$ is due to \eqref{eq:secondorderbd}. Using condition (S.1) Lemma~\ref{lemma:limsupsufficient}, we have:
	\begin{equation}
		\lim_{k\rightarrow\infty}\text{Pr}\left(\sup_{j\geq k}\left|\sum_{i=k}^j\epsilon_i\beta_{i, 2}\right|\geq \mu/3\right)=0. \label{eq:b2}
	\end{equation}
	
	Finally the third bias term satisfies $\mathbb{E}[|\beta_{k, 3}|]=\mathbb{E}[l(\gamma_k)\cdot|\gamma_k-\gamma^\star|]\leq\sqrt{\mathbb{E}[\max\{3\gamma_k, Z_D^2\}]^2\mathbb{E}[(\gamma_k-\gamma^\star)^2]}$. Since $\mathbb{E}[\max\{3\gamma_k, Z_D^2\}]\leq \mathbb{E}[3\gamma_k]+\overline{D}$ is bounded according to \eqref{eq:absub} and $\mathbb{E}[(\gamma_k-\gamma^\star)^2]=\mathcal{O}(k^{-\alpha})$, $\mathbb{E}[|\beta_{k, 3}|]=\mathcal{O}(k^{-\alpha})$. Condition (S.2) in Lemma~\ref{lemma:limsupsufficient} is verified and we have
	\begin{align}
		\lim_{k\rightarrow\infty}\text{Pr}\left(\sup_{j\geq k}\sum_{i=k}^j\epsilon_i\beta_{i, 3}\geq\mu/3\right)=0. \label{eq:b3}
	\end{align}
	
	Plugging \eqref{eq:b1}, \eqref{eq:b2} and \eqref{eq:b3} into \eqref{eq:btot} verifies statement (2.5). 
	
	\hspace{-10pt}\textbf{(2.6): }Function $\rho(\gamma)\leq l(\gamma^\star)|\gamma-\gamma^\star|+\frac{3}{2}(\gamma-\gamma^\star)^2$. Since $\mathbb{E}[|\gamma-\gamma^\star|^2]=\mathcal{O}(k^{-\alpha})$ and $\mathbb{E}[|\gamma-\gamma^\star|]=\mathcal{O}(k^{-\alpha/2})$ satisfies condition (S.2) in Lemma~\ref{lemma:limsupsufficient}, statement (2.6) is verified. 
\begin{IEEEbiographynophoto}{Haoyue Tang}(Student Member, IEEE) received the B.Eng and Ph.D. degrees from the Department of Electronic Engineering, Tsinghua University, Beijing, China, in 2017 and 2022, respectively. She is now a postdoctoral research associate at Yale University. She was a Visiting Student with Technische Universitat München from September 2015 to February 2016, and Télécom Paris from January 2019 to March 2019. Her research interests include age of information, stochastic network optimization, and statistical learning theory.
\end{IEEEbiographynophoto}

\begin{IEEEbiographynophoto}{Yin Sun}
	received the B.Eng. and Ph.D. degrees in electronic engineering from
	Tsinghua University in 2006 and 2011, respectively.
	From 2011 to 2017, he was a Post-Doctoral Scholar
	and a Research Associate with The Ohio State
	University. He is currently an Assistant Professor
	with the Department of Electrical and Computer
	Engineering, Auburn University. He coauthored a
	monograph Age of Information: A New Metric for
	Information Freshness (Morgan and Claypool Publishers, 2019). His research interests include age of
	information, networking, robotic control, information theory, and machine
	learning. He is a member of the ACM. His articles received the Best Student
	Paper Award from the IEEE/IFIP WiOpt 2013, the Best Paper Award from
	the IEEE/IFIP WiOpt 2019, and runner-up for the Best Paper Award of ACM
	MobiHoc 2020 and the 2021 Journal of Communications and Networks (JCN)
	Best Paper Award. He received the Auburn Author Award of 2020. He cofounded the Age of Information Workshop in 2018
\end{IEEEbiographynophoto}


\begin{IEEEbiographynophoto}{Leandros Tassiulas}
	received the
	Ph.D. degree in electrical engineering from the University of Maryland, College Park, MD, USA, in
	1991. He held Faculty positions with the Polytechnic University, New York, NY, USA, University of
	Maryland, and University of Thessaly, Greece. He
	is currently the John C. Malone Professor of electrical engineering with Yale University, New Haven,
	CT, USA. His research interests include computer
	and communication networks, with an emphasis on
	fundamental mathematical models and algorithms of
	complex networks, architectures and protocols of wireless systems, sensor
	networks, novel internet architectures, and experimental platforms for network
	research. His most notable contributions include the max-weight scheduling
	algorithm and the back-pressure network control policy, opportunistic scheduling in wireless, the maximum lifetime approach for wireless network energy
	management, and the consideration of joint access control and antenna transmission management in multiple antenna wireless systems. His research has been
	recognized by several awards, including the IEEE Koji Kobayashi Computer
	and Communications Award, the Inaugural INFOCOM 2007 Achievement
	Award for fundamental contributions to resource allocation in communication
	networks, the INFOCOM 1994 and 2017 best paper awards, the National Science
	Foundation (NSF) Research Initiation Award in 1992, the NSF CAREER Award
	in 1995, the Office of Naval Research Young Investigator Award in 1997, and
	the Bodossaki Foundation Award in 1999.
\end{IEEEbiographynophoto}
\bibliographystyle{IEEEtran}
\bibliography{bibfile}

\newpage 

\section{Proof of Lemma~\ref{coro:sig-dep-reformulate}}\label{pf:sig-dep-reformulate}
\begin{proof}
	First we will turn the time-averaged MSE computation into frame-level computation. For stationary policy $\pi$ that decides sampling time $S_{k+1}$ only on information $\mathcal{I}_k$, tuple $\{(I_k, (S_{k+1}-S_k))\}$ is a regenerative sequence. Recall that $E_k=\int_{S_k}^{S_{k+1}}(X_t-X_{S_k})^2\text{d}t$ and $L_k:=S_{k+1}-S_k$ are the cumulative estimation error and length of frame $k$, which are both generative because policy $\pi$ is stationary. Therefore, sequence $\{\frac{1}{K}\mathbb{E}\left[\sum_{k=1}^KE_k\right]\}$ and $ \{\frac{1}{K}\mathbb{E}\left[\sum_{k=1}^KL_k\right]\}$ have limits. Then according to the renewal reward theory \cite{ross2013applied}, the time averaged MSE can be computed by:
	\begin{align}
		&\limsup_{T\rightarrow\infty}\mathbb{E}\left[\int_{t=0}^T\left(X_t-\hat{X}_t\right)^2\text{d}t\right]\nonumber\\
		=&\limsup_{K\rightarrow\infty}\frac{\mathbb{E}\left[\sum_{k=1}^K\int_{S_k}^{S_{k+1}}(X_t-X_{S_{k-1}})^2\text{d}t\right]}{\mathbb{E}\left[\sum_{k=1}^K\left(S_{k+1}-S_k\right)\right]}\nonumber\\
		{=}&\limsup_{K\rightarrow\infty}\frac{\sum_{k=1}^K\mathbb{E}\left[\int_{S_k}^{S_{k+1}}(X_t-X_{S_{k-1}})^2\text{d}t\right]}{\sum_{k=1}^K\mathbb{E}\left[\left(S_{k+1}-S_k\right)\right]}.\label{eq:frame-mse}
	\end{align}

	To simplify the computation of equation~\ref{eq:frame-mse}, we will first introduce the following lemma:
	\begin{lemma}[Lemma 6, \cite{sun_wiener}Restated]
		\label{lemma:4-2}Let $Z_t$ be a Wiener process starting from time zero, let $\tau$ be a stopping time of $Z_t$, we have:
		\begin{equation}
			\frac{1}{6}\mathbb{E}\left[Z_\tau^4\right]=\mathbb{E}\left[\int_{0}^{\tau}Z_t^2\text{d}t\right]
		\end{equation}
	\end{lemma}
	
	Using Lemma~\ref{lemma:4-2}, we can then compute the expected cumulative estimation error during interval $[S_k, R_k]$. Notice that during the interval, the $k$-th sample has not been received. Therefore, the estimation error $X_t-\hat{X}_t=X_t-X_{S_{k-1}}$ can be viewed as a Wiener process starting from time $S_{k-1}$. We can decouple and compute the cumulative estimation error as follows:
	\begin{align}
		&\mathbb{E}\left[\int_{S_k}^{R_k}(X_t-\hat{X}_t)^2\text{d}t\right]\nonumber\\
		=&\mathbb{E}\left[\int_{S_{k-1}}^{R_k}(X_t-X_{S_{k-1}})^2\text{d}t\right]-\mathbb{E}\left[\int_{S_{k-1}}^{S_k}(X_t-X_{S_{k-1}})^2\text{d}t\right]\nonumber\\
		=&\frac{1}{6}\mathbb{E}\left[(X_{R_k}-X_{S_{k-1}})^4\right]-\frac{1}{6}\mathbb{E}\left[(X_{S_{k}}-X_{S_{k-1}})^4\right].\label{eq:rr-1}
	\end{align}
	
	Similarly, we can then obtain that:
	\begin{align}
		&\mathbb{E}\left[\int_{R_k}^{S_{k+1}}(X_t-\hat{X}_t)^2\text{d}t\right]\nonumber\\
		=&\frac{1}{6}\mathbb{E}\left[(X_{S_{k+1}}-X_{S_k})^4\right]-\frac{1}{6}\mathbb{E}\left[(X_{R_k}-X_{S_k})^4\right].\label{eq:rr-2}
	\end{align}
	
	Notice that the transmission delay $D_k$ is i.i.d across all slots, since we only focus on stationary policies whose waiting time $W_k$ relies only on recent information $\mathcal{I}_k:=\{Y_k, (W_{S_k+t}-W_{S_k}), \forall t\geq 0\}$, we have:
	\[\mathbb{E}[(X_{S_k}-X_{S_{k-1}})^4]=\mathbb{E}[(X_{S_{k+1}}-X_{S_k})^4]\]
	
	Therefore, by summing up \eqref{eq:rr-1} and \eqref{eq:rr-2}, for any policy $\pi$ that makes decisions only on $\mathcal{I}_k$, the expected  cumulative estimation error in frame $k$ can be computed by:
	\begin{align}
		&\mathbb{E}\left[E_k\right]=\mathbb{E}\left[\int_{S_k}^{S_{k+1}}(X_t-\hat{X}_t)^2\text{d}t\right]\nonumber\\
		=&\frac{1}{6}\mathbb{E}\left[(X_{R_k}-X_{S_{k-1}})^4\right]-\frac{1}{6}\mathbb{E}\left[(X_{R_k}-X_{S_k})^4\right]\nonumber\\
		=&\frac{1}{6}\mathbb{E}\left[\left((X_{R_k}-X_{S_k})+(X_{S_k}-X_{S_{k-1}})\right)^4\right]-\frac{1}{6}\mathbb{E}\left[(X_{R_k}-X_{S_k})^4\right]\nonumber\\
		=&\frac{1}{6}\mathbb{E}\left[(X_{R_k}-X_{S_k})^4\right]+\frac{2}{3}\mathbb{E}\left[(X_{R_k}-X_{S_k})^3(X_{S_k}-X_{S_{k-1}})\right]+\mathbb{E}\left[(X_{R_k}-X_{S_k})^2(X_{S_k}-X_{S_{k-1}})^2\right]\nonumber\\
		&+\frac{2}{3}\mathbb{E}\left[(X_{R_k}-X_{S_k})(X_{S_k}-X_{S_{k-1}})^3\right]+\frac{1}{6}\mathbb{E}\left[(X_{S_k}-X_{S_{k-1}})^4\right]-\frac{1}{6}\mathbb{E}\left[(X_{R_k}-X_{S_k})^4\right]\nonumber\\
		\overset{(a)}{=}&\mathbb{E}[(X_{R_k}-X_{S_k})^2(X_{S_k}-X_{S_{k-1}})^2]+\frac{1}{6}\mathbb{E}\left[(X_{R_k}-X_{S_k})^4\right]\nonumber\\
		\overset{(b)}{=}&\mathbb{E}\left[(X_{R_k}-X_{S_k})^2\right]\cdot\mathbb{E}[(X_{S_k}-X_{S_{k-1}})^2]+\frac{1}{6}\mathbb{E}\left[(X_{S_k}-X_{S_{k-1}})^4\right]\nonumber\\
		\overset{(c)}{=}&\mathbb{E}[R_k-S_k]\cdot\mathbb{E}[S_{k+1}-S_k]+\frac{1}{6}\mathbb{E}\left[(X_{S_k}-X_{S_{k-1}})^4\right]\nonumber\\
		=&\mathbb{E}[L_k]\cdot\overline{D}+\frac{1}{6}\mathbb{E}\left[(X_{S_k}-X_{S_{k-1}})^4\right],\label{eq:E-k-expect}
	\end{align}
	where equality $(a)$ is obtained because $(X_{R_k}-X_{S_k})$ is independent of $(X_{S_k}-X_{S_{k-1}})$, since $\mathbb{E}[X_{R_k}-X_{S_k}]=\mathbb{E}[X_{S_k}-X_{S_{k-1}}]=0$ due to Wiener process evolution, we have $\mathbb{E}[(X_{R_k}-X_{S_k})(X_{S_k}-X_{S_{k-1}})^3]=0$, $\mathbb{E}[(X_{R_k}-X_{S_k})^3(X_{S_k}-X_{S_{k-1}})]=0$; equality $(b)$ is obtained because $(X_{R_k}-X_{S_k})$ and $(X_{S_k}-X_{S_{k-1}})$ are independent; and equality $(c)$ is because the Wald's Lemma. 
	
	Therefore, for any stationary policy $\pi$ that makes sampling decision only on $\mathcal{I}_k$, with probability 1, the objective function in the MSE minimization problem~\ref{pb:mse} can be rewritten as:
	\begin{align}
		&\limsup_{T\rightarrow\infty}\mathbb{E}\left[\frac{1}{T}\int_{t=0}^T(X_t-\hat{X}_t)^2\text{d}t\right]\nonumber\\
		=&\limsup_{K\rightarrow\infty}\frac{\sum_{k=1}^K\left(\frac{1}{6}\mathbb{E}\left[(X_{S_k}-X_{S_{k-1}})^4\right]+\mathbb{E}[L_k]\overline{D}\right)}{\sum_{k=1}^K\mathbb{E}[L_k]}\nonumber\\
		=&\limsup_{K\rightarrow\infty}\frac{\sum_{k=1}^K\mathbb{E}\left[\frac{1}{6}(X_{S_k}-X_{S_{k-1}})^4\right]}{\sum_{k=1}^K\mathbb{E}[L_k]}+\overline{D}. 
	\end{align}
\end{proof}

		\section{Proof of Corollary~\ref{thm:mse-rate}}\label{pf:mse-rate}
The cumulative MSE up to the beginning of frame $k+1$ can be decomposed into:
\begin{align}
	&\mathbb{E}\left[\int_0^{S_{k+1}}(X_t-\hat{X}_t)^2\text{d}t\right]\nonumber\\
	\overset{(a)}{=}&\sum_{k'=1}^k\mathbb{E}\left[\int_{S_{k'}}^{S_{k'}+D_{k'}}(X_t-X_{S_{k'-1}})^2\text{d}t+\int_{S_{k'}+D_{k'}}^{S_{k'+1}}(X_t-X_{S_{k'}})^2\text{d}t\right]\nonumber\\
	=&\sum_{k'=1}^k\mathbb{E}\left[\int_{S_{k'}}^{S_{k'}+D_{k'}}(X_t-X_{S_{k'}}+X_{S_{k'}}-X_{S_{k'-1}})^2\text{d}t\right.\nonumber\\
	&\left.+\int_{S_{k'}+D_{k'}}^{S_{k'+1}}(X_t-X_{S_{k'}})^2\text{d}t\right]\nonumber\\
	=&\sum_{k'=1}^k\mathbb{E}\left[\int_{S_{k'}}^{S_{k'}+D_{k'}}\left((X_t-X_{S_{k'}})^2\right.\right.\nonumber\\
	&\left.\left.+2(X_t-X_{S_{k'}})(X_{S_{k'}}-X_{S_{k'-1}})+(X_{S_{k'}}-X_{S_{k'-1}})^2\right)\text{d}t\right.\nonumber\\
	&\left.+\int_{S_{k'}+D_{k'}}^{S_{k'+1}}(X_t-X_{S_{k'}})^2\text{d}t\right]\nonumber\\
	\overset{(b)}{=}&\sum_{k'=1}^k\mathbb{E}\left[(X_{S_{k'}}-X_{S_{k'-1}})^2D_{k'}+\int_{S_{k'}}^{S_{k'+1}}(X_t-X_{S_{k'}})^2\text{d}t\right]\nonumber\\
	\overset{(c)}{=}&\sum_{k'=1}^k\mathbb{E}[(X_{S_{k'}}-X_{S_{k'-1}})^2]\overline{D}+\sum_{k'=1}^k\frac{1}{6}\mathbb{E}\left[(X_{S_{k'+1}}-X_{S_{k'}})^4\right], \label{eq:mse}
\end{align}
where equality $(a)$ is because $\hat{X}_t=X_{S_{k-1}}, 
\forall t\in[S_{k}+D_k)$ and $\hat{X}_t=X_{S_k}, \forall t\in[S_k+D_k, S_{k+1})$; equality $(b)$ is because $\mathbb{E}[X_t-X_{S_{k'}}]=0$ and because $X_{S_{k'}}-X_{S_{k'-1}}$ is independent of $D_{k'}$; equality $(c)$ is because $D_{k'}$ is independent of $(X_{S_{k'}}-X_{S_{k'-1}})$. 

With equation \eqref{eq:mse}, we proceed to bound the difference $\mathbb{E}\left[\int_0^{S_{k+1}}(X_t-\hat{X}_t)^2\text{d}t\right]-(\gamma^\star+\overline{D})\mathbb{E}\left[\sum_{k'=1}^kL_{k'}\right]$ as follows:
\begin{align}
	&\mathbb{E}\left[\int_0^{S_{k+1}}(X_t-\hat{X}_t)^2\text{d}t\right]-(\gamma^\star+\overline{D})\mathbb{E}\left[\sum_{k'=1}^kL_{k'}\right]\nonumber\\
	\overset{(a)}{=}&\cancel{\sum_{k'=1}^k\mathbb{E}[(X_{S_{k'+1}}-X_{S_{k'}})^2]\overline{D}}+\sum_{k'=1}^k\frac{1}{6}\mathbb{E}\left[(X_{S_{k'+1}}-X_{S_{k'}})^4\right]\nonumber\\
	&-\gamma^\star\mathbb{E}[(X_{S_{k'}+1}-X_{S_{k'}})^2]-\cancel{\overline{D}\sum_{k'=1}^k\mathbb{E}[(X_{S_{k'}+1}-X_{S_{k'}})^2]}\nonumber\\
	=&\sum_{k'=1}^k\left(\mathbb{E}[q(\gamma_{k'})-\gamma^\star l(\gamma_{k'})]\right)\nonumber\\
	=&\sum_{k'=1}^k\mathbb{E}\left[q(\gamma_{k'})-\gamma_{k'}l(\gamma_{k'})+(\gamma_{k'}-\gamma^\star)l(\gamma_{k'})\right]\nonumber\\
	\overset{(b)}{\leq}&\sum_{k'=1}^k\mathbb{E}\left[q(\gamma^\star)-\gamma_{k'}l(\gamma^\star)+(\gamma_{k'}-\gamma^\star)l(\gamma_{k'})\right]\nonumber\\
	\overset{(c)}{=}&\sum_{k'=1}^k\mathbb{E}[(\gamma^\star-\gamma_{k'})(l(\gamma^\star)-l(\gamma_{k'}))]\nonumber\\
	\overset{(d)}{\leq}&3\sum_{k'=1}^k\mathbb{E}[(\gamma_{k'}-\gamma^\star)^2],
\end{align}
where equality $(a)$ is because of equation \eqref{eq:mse} and by martingale stopping theorem, $\mathbb{E}[(X_{S_{k'+1}}-X_{S_{k'}})^2]=\mathbb{E}[L_{k'}]$; inequality $(b)$ is because choosing threshold $\gamma_{k'}$ minimizes function $q(\gamma)-\gamma_kl(\gamma)$; inequality $(c)$ is obtained because $q(\gamma^\star)-\gamma_{k'}l(\gamma^\star)=q(\gamma^\star)-\gamma^\star l(\gamma^\star)+(\gamma^\star-\gamma_{k'})l(\gamma^\star)$ and $q(\gamma^\star)-\gamma^\star l (\gamma^\star)=0$ according to Lemma~\ref{lemma:g}-(ii); inequality $(d)$ is obtained because $|l(\gamma^\star)-l(\gamma_{k'})|=\left|\mathbb{E}[\max\{3\gamma^\star, Z_D^2\}]-\mathbb{E}[\max\{3\gamma_{k'}, Z_D^2\}]\right|\leq 3|\gamma_{k'}-\gamma^\star|$.

\section{Proof of Lemma~\ref{coro:gammadep-bound}}\label{pf:gammadep-bound}
\begin{proof}
	Consider a constant wait policy $\pi_{\mathsf{const}}$ that chooses $W_k\equiv\frac{1}{f_{\mathsf{max}}}$ regardless of the transmission delay and estimation error $\mathcal{I}_k=(D_k, (X_{S_k}+t-X_{S_k}))$. Let $Z_t$ be a Wiener process starting from time 0. Then according to equation~\eqref{eq:sig-dep-rr-goal} from Corollary~\ref{coro:sig-dep-reformulate}, the time-average MSE by using policy $\pi_{\mathsf{const}}$ can be computed by:
	\begin{align}
		\overline{\mathcal{E}}_{\pi_{\mathsf{const}}}=&\limsup_{K\rightarrow\infty}\frac{\sum_{k=1}^K\mathbb{E}\left[\frac{1}{6}(X_{S_k}-X_{S_{k-1}})^4\right]}{\sum_{k=1}^K\mathbb{E}[L_k]}+\overline{D}\nonumber\\	
		\overset{(a)}{=}&\limsup_{K\rightarrow\infty}\frac{\frac{1}{6}\sum_{k=1}^K\mathbb{E}\left[Z_{D_k+\frac{1}{f_{\mathsf{max}}}}^4\right]}{\frac{1}{K}\sum_{k=1}^K\mathbb{E}[L_k]}+\overline{D}\nonumber\\
		\overset{(b)}{=}&\limsup_{K\rightarrow\infty}\frac{\frac{1}{2}\sum_{k=1}^K\mathbb{E}\left[(D_k+\frac{1}{f_{\mathsf{max}}})^2\right]}{\frac{1}{K}\sum_{k=1}^K\mathbb{E}[L_k]}+\overline{D}\nonumber\\
		\overset{(c)}{=}&\frac{1}{2}\frac{M+2\overline{D}\frac{1}{f_{\mathsf{max}}}+\left(\frac{1}{f_{\mathsf{max}}}\right)^2}{\overline{D}+\frac{1}{f_{\mathsf{max}}}}+\overline{D},
	\end{align}
	where equality $(a)$ is because by using policy $\pi_{\mathsf{const}}$, given transmission delay $D_k$, the difference $X_{S_{k-1}+t}-X_{S_{k-1}}$ evolves like Wiener process $Z_t$ starting from $t=0$ and $S_{k+1}-S_{k}=D_k+\frac{1}{f_{\mathsf{max}}}$ due to the constant wait policy; equality $(b)$ is obtained because given delay $D_k$, $Z_{D_k+\frac{1}{f_{\mathsf{max}}}}\sim\mathcal{N}(0, D_k+\frac{1}{f_{\mathsf{max}}})$ is a zero-mean Gaussian distribution with variance $D_k+\frac{1}{f_{\mathsf{max}}}$, and therefore the fourth order moment $\mathbb{E}\left[Z_{D_k+\frac{1}{f_{\mathsf{max}}}}^4|D_k\right]=3\left(D_k+\frac{1}{f_{\mathsf{max}}}\right)^2$; equality $(c)$ is obtained because $D_k$ is i.i.d following distribution $\mathbb{P}_D$, by definition $M=\mathbb{E}_{\mathbb{P}_D}[D^2]$ and $\overline{D}=\mathbb{E}_{\mathbb{P}_D}[D]$. Since $\pi_{\mathsf{const}}$ may not be the MSE minimum sampling policy, we have $\overline{\mathcal{E}}_{\pi^\star}\leq\overline{\mathcal{E}}_{\pi_{\mathsf{const}}}$. Therefore, recall the definition $\gamma^\star=\overline{\mathcal{E}}_{\pi^\star}-\overline{D}$ from Subsection~\ref{sec:dep-off}, we have:
	\begin{equation}
		\gamma^\star\leq \overline{\mathcal{E}}_{\pi_{\mathsf{const}}}-\overline{D}=\frac{1}{2}\frac{M+2\overline{D}\frac{1}{f_{\mathsf{max}}}+\left(\frac{1}{f_{\mathsf{max}}}\right)^2}{\overline{D}+\frac{1}{f_{\mathsf{max}}}}=:\gamma_{\mathsf{ub}}. 
	\end{equation}
	
	We then derive the lower bound of $\overline{\mathcal{E}}_{\pi^\star}$. Recall that the optimum decision rule of policy $\pi^\star$ is given in equation \eqref{eq:opt-dep}, according to \cite[Theorem 1]{sun_wiener}, $\overline{\mathcal{E}}_{\pi^\star}$ can be computed by
	\begin{align}
		&\overline{\mathcal{E}}_{\pi^\star}=\frac{1}{6}\frac{\mathbb{E}\left[\max\{3\left(\gamma^\star+\nu^\star\right), Z_{D}^2\}^2\right]}{\mathbb{E}\left[\max\{3\left(\gamma^\star+\nu^\star\right), Z_{D}^2\}\right]}+\overline{D}\nonumber\\
		\overset{(d)}{\geq}&\frac{1}{6}\mathbb{E}[\max\{3(\gamma^\star+\nu^\star), Z_D^2\}]+\overline{D}\geq\frac{1}{6}\mathbb{E}[Z_D^2]+\overline{D}=\frac{7}{6}\overline{D}, 
	\end{align}
	where inequality $(d)$ is from Cauchy-Schwartz inequality. 
	
	Finally, $\gamma^\star$ can be lower bonded by:
	\begin{align}
		\gamma^\star=\overline{\mathcal{E}}_{\pi^\star}-\overline{D}\geq\frac{1}{6}\overline{D}. 
	\end{align}
\end{proof}

\section{Proof of Lemma~\ref{lemma:4order}}\label{pf:4order}
For any time $t$, the value $Z_t$ of the Wiener process $Z_t\sim\mathcal{N}(0, t)$, according to \cite[Theorem 7.5.6]{probability}, $\forall \theta\in\mathbb{R}$, sequence $M_t(\theta):=\exp\left(\theta Z_t-\frac{\theta^2}{2}t\right)$ is a martingale with initial value $M_0(\theta)=1, \forall\theta$. 

Let $T\geq 0$ fixed as a constant, then $\tau_\gamma\wedge T$ is a stopping time, where $a\wedge b=\min\{a, b\}$. Then according to the optional stopping theorem \cite[Theorem 7.5.1]{probability}, denote $\phi_T(\theta)$ to be the expected value of $M_{\tau_\gamma\wedge T}(\theta)$, we have:
\begin{equation}
	\phi_T(\theta):=\mathbb{E}\left[M_{\tau_\gamma\wedge T}(\theta)\right]=\mathbb{E}\left[M_{0}(\theta)\right]=1,\forall \theta. 
\end{equation}

Therefore, the $n$-th order derivative of function $\phi_T(\theta)$, denoted by $\phi_T^{(n)}(\theta)$ can be computed by:
\begin{equation}
	\phi^{(n)}_T(\theta)=\frac{\partial^n\mathbb{E}\left[M_{\tau_\gamma\wedge T}(\theta)\right]}{\partial \theta^n}=0.
\end{equation}

For each sample path $\omega$, the absolute value $\left|Z_{l_\gamma\wedge T}\right|\leq \sqrt{3\gamma}$. Therefore the derivative $\left|\frac{\partial^n M_{\tau_\gamma\wedge T}(\theta)}{\partial \theta^n}\right|$ is bounded and continuous. Then according to Leibniz rule we have
\begin{equation}
	\mathbb{E}\left[\frac{\partial^n M_{\tau_\gamma\wedge T}(\theta)}{\partial \theta^n}\right]=\frac{\partial ^n\mathbb{E}\left[M_{\tau_\gamma\wedge T}(\theta)\right]}{\partial\theta^n}=0.\label{eq:2-nd}
\end{equation}

First according to \cite[Theorem 7.5.1\& Theorem 7.5.5]{probability}, $\forall\gamma<\infty$, the mean of stopping time $\tau_\gamma$ is bounded, i.e., 
\begin{equation}
	\mathbb{E}\left[l_\gamma\right]=\mathbb{E}[Z_{l_\gamma}^2]=\mathbb{E}[\max\{3\gamma, Z_D^2\}]=3\gamma+\overline{D}<\infty..
\end{equation}

To obtain the second-order moment of $\tau_\gamma$, we compute the 4-th order derivative of $M_{l_\gamma\wedge T})(\theta)$, i.e., 
\begin{equation}
	\frac{\partial^4M_{l_\gamma\wedge T}(\theta)}{\partial\theta}\big|_{\theta=0}=Z_{l_\gamma\wedge T}^4-6(l_\gamma\wedge T)\cdot Z_{l_\gamma\wedge T}^2+3(l_\gamma\wedge T)^2. \label{eq:4nd}
\end{equation}

Plugging \eqref{eq:4nd} into  \eqref{eq:2-nd}, we have:
\begin{align}
	&\mathbb{E}\left[(l_\gamma\wedge T)^2\right]\nonumber\\
	=&2\mathbb{E}\left[(\frac{1}{\sqrt{2}}l_\gamma\wedge T)\cdot \sqrt{2}Z_{l_\gamma\wedge T}^2\right]-\frac{1}{3}\mathbb{E}\left[Z_{l_\gamma\wedge T}^4\right]\nonumber\\
	\leq&\frac{1}{2}\mathbb{E}\left[(l_\gamma\wedge T)^2\right]+\frac{5}{3}\mathbb{E}[Z_{l_\gamma\wedge T}^2]
\end{align}

Therefore, for each $T<\infty$, the second order moment of $(l_\gamma\wedge T)$ can be upper bounded by:
\begin{align}
	&\mathbb{E}\left[(l_\gamma\wedge T)^2\right]\leq\frac{10}{3}\mathbb{E}\left[Z_{l_\gamma\wedge T}^2\right]\nonumber\\
	=&\frac{10}{3}\mathbb{E}\left[\max\{3\gamma, Z_D^2\}\right]\leq\frac{10}{3}(3\gamma+\overline{D})<\infty. 
\end{align}

Finally, we can upper bound the fourth order moment of $\tau_{\gamma}$ by computing the 8-th order derivative of $M_{\tau_\gamma\wedge T}(\theta)$ at $\theta=0$, i.e., 
\begin{align}
	&\frac{\partial^8 M_{l_\gamma\wedge T}(\theta)}{\partial \theta^8}\nonumber\\
	=&Z_{l_\gamma\wedge T}^8-28(l_{\gamma}\wedge T)\cdot Z_{\tau_\gamma\wedge T}^6+210(l_{\gamma}\wedge T)^2\cdot Z_{l_\gamma\wedge T}^4\nonumber\\
	&-420(l_{\gamma}\wedge T)^3\cdot Z_{l_\gamma\wedge T}^2+105(l_\gamma\wedge T)^4.\label{eq:8nd}
\end{align}

Plugging equation~\eqref{eq:8nd} into \eqref{eq:2-nd}, we have:
\begin{align}
	&\mathbb{E}\left[(l_\gamma\wedge T)^4\right]\nonumber\\
	=&-\frac{1}{105}\mathbb{E}\left[Z_{l_\gamma\wedge T}^8-28(l_{\gamma}\wedge T)\cdot Z_{\tau_\gamma\wedge T}^6\right.\nonumber\\
	&\left.+210(l_{\gamma}\wedge T)^2\cdot Z_{\tau_\gamma\wedge T}^4-420(l_{\gamma}\wedge T)^3\cdot Z_{\tau_\gamma\wedge T}^2\right]\nonumber\\
	\leq&4\mathbb{E}\left[(l_\gamma\wedge T)^3\cdot Z_{l_\gamma\wedge T}^2\right]+\frac{3}{15}\mathbb{E}\left[(l_\gamma\wedge T)\cdot Z_{l_\gamma\wedge T}^6\right]\nonumber\\
	\leq&4\left(\mathbb{E}\left[(l_\gamma\wedge T)^{4}\right]\right)^{3/4}\cdot \left(\mathbb{E}\left[Z_{l_\gamma\wedge T}^8\right]\right)^{1/4}\nonumber\\
	&+\frac{3}{15}\left(\mathbb{E}\left[(l_\gamma\wedge T)^4\right]\right)^{1/4}\cdot\left(\mathbb{E}\left[Z_{l_\gamma\wedge T}^8\right]\right)^{3/4}\nonumber\\
	\leq&4\left(\mathbb{E}\left[(l_\gamma\wedge T)^{4}\right]\right)^{3/4}\cdot \left((3\gamma)^4+105B\right)^{1/4}\nonumber\\
	&+\frac{3}{15}\left(\mathbb{E}\left[(l_\gamma\wedge T)^4\right]\right)^{1/4}\cdot\left((3\gamma)^4+105B\right)^{3/4}. \label{eq:8nd-2}
\end{align}

Inequality \eqref{eq:8nd-2} implies
\begin{equation}
	\mathbb{E}\left[(l_\gamma\wedge T)^4\right]\leq 4^3\left((3\gamma)^4+105B\right). 
\end{equation}
Let $T\rightarrow\infty$ and then use the dominated convergence theorem on the LHS of inequality~\eqref{eq:8nd-2}, we conclude that $\mathbb{E}[l_\gamma^4]$ is bounded. 
%
%
\begin{equation}
	\mathbb{E}[\tau_\gamma^4]=-\frac{1}{105}\left(1-28+350-1708\right)(3\gamma)^4=\frac{277}{21}(3\gamma)^4. 
\end{equation}

		\section{Proof of Lemma~\ref{lemma:g}}\label{pf:lemma-g}
		\begin{proof}First we can compute the derivatives as follows:
			
			\begin{align}
				&\frac{\text{d}}{\text{d}\gamma}\mathbb{E}\left[\frac{1}{6}\left((3\gamma)^2\cdot\mathbb{I}_{(Z_D^2\leq 3\gamma)}+Z_D^4\cdot\mathbb{I}_{(Z_D^2>3\gamma)}\right)\right]\nonumber\\
				=&3\gamma\text{Pr}\left(Z_D^2\leq 3\gamma\right).\\
				&\frac{\text{d}}{\text{d}\gamma}\mathbb{E}\left[\left(3\gamma\cdot\mathbb{I}_{(Z_D^2\leq 3\gamma)}+Z_D^2\cdot\mathbb{I}_{(Z_D^2>3\gamma)}\right)\right]\nonumber\\
				=&3\text{Pr}\left(Z_D^2\leq 3\gamma\right)
			\end{align}
			Therefore, the monotonic decreasing characteristics can be verified through
			\begin{align}
				&\overline{g}_0'(\gamma)=-\mathbb{E}\left[\max\{3\gamma, Z_D^2\}\right]<0\\
				&\overline{g}_0''(\gamma)=-3\text{Pr}(Z_D^2\leq 3\gamma)<0.
			\end{align}
			
			Through Taylor expansion, since $\overline{g}_0(\gamma)$ is monotonically decreasing, for $\gamma<\gamma^\star$ we have:
			\begin{align}
				\overline{g}_0(\gamma)\geq&\overline{g}_0(\gamma^\star)+\overline{g}_0'(\gamma^\star)(\gamma-\gamma^\star)=-l(\gamma^\star)(\gamma-\gamma^\star)\geq 0.
				\label{eq:g-1}
			\end{align}
			
			Since $\gamma-\gamma^\star<0$, inequality \eqref{eq:g-1} implies
			\begin{equation}
				(\gamma-\gamma^\star)\overline{g}_0(\gamma)\leq-l(\gamma^\star)(\gamma-\gamma^\star)^2.\label{eq:gleq}
			\end{equation}
			
			And for $\gamma>\gamma^\star$, we have:
			\begin{align}
				&\overline{g}_0(\gamma)\nonumber\\
				\overset{(a)}{\leq}&\mathbb{E}\left[\frac{1}{6}\max\{3\gamma^\star, Z_D^2\}^2-\gamma\max\{3\gamma^\star, Z_D^2\}\right]\nonumber\\
				\overset{(b)}{=}&(\gamma^\star-\gamma)l(\gamma^\star). \label{eq:g-2}
			\end{align}
			where inequality $(a)$ is obtained because choosing the stopping time to be $\tau=\inf\{t\geq D\big||Z_t|\geq\sqrt{3\gamma}\}$ minimizes function $\mathbb{E}\left[\frac{1}{6}Z_{\tau}^4-\gamma \tau\right]$ and equality $(b)$ is because $\mathbb{E}\left[\frac{1}{6}\max\{3\gamma^\star, Z_D^2\}^2-\gamma^\star\max\{3\gamma^\star, Z_D^2\}\right]=0$. 
			
			Multiplying $(\gamma-\gamma^\star)$ on both sides of \eqref{eq:g-2} we have:
			\begin{equation}
				(\gamma-\gamma^\star)\overline{g}_0(\gamma)\leq -l(\gamma^\star)(\gamma^\star-\gamma)^2, \forall \gamma>\gamma^\star. \label{eq:ggeq}
			\end{equation}
			
			Combining \eqref{eq:gleq} and \eqref{eq:ggeq} finishes the proof of Lemma~\ref{lemma:g}-(iii). 
			
		\end{proof}

		\section{Proof of Lemma~\ref{lemma:regrettoerror}}\label{pf:regrettoerror}
\begin{proof}
	For $l\geq \overline{D}$, let $\Pi_l\triangleq\{\pi|\mathbb{E}[Z_\tau^2]=l, \forall \pi\in\Pi\}$ whose squared error at the time of sample is $l$.  Next, we establish the lower bound of $\mathbb{E}\left[\int_0^\tau Z_t^2\text{d}t\right]$ for any policy $\pi\in\Pi_l$, which can be formulated into the following optimization problem:
	\begin{equation}
		\inf_{\pi}\mathbb{E}\left[\int_0^{\tau} Z_t^2\text{d}t\right], \text{ s.t. }\mathbb{E}\left[\tau \right]=l, \tau\geq  D.\label{eq:lfixopt}  
	\end{equation}
	
	As is shown in \cite[Theorem 7]{sun_wiener}, the optimum solution to \eqref{eq:lfixopt} has a threshold structure, and the optimum sampling policy for each sample path is as follows:
	
	%
	%
	\begin{equation}
		\tau=\inf\{t\geq D||Z_t|\geq\lambda^\star\}, \label{eq:lagrange-stop}
	\end{equation}
	where the selection of $\lambda^\star$ satisfies:
	\begin{equation}
		\mathbb{E}[\tau]=\mathbb{E}[Z_{\tau}^2]=l.\label{eq:taulstar} 
	\end{equation}
	
	Through Lemma~\ref{lemma:4-2}, we can compute the optimum solution to \eqref{eq:lfixopt} as follows:
	\begin{equation}
		\mathbb{E}\left[\int_0^\tau Z_t^2\text{d}t\right]=\frac{1}{6}\mathbb{E}[Z_\tau^4]=\frac{1}{6}\mathbb{E}\left[\max\{Z_D^2, 3\gamma\}^2\right]. \label{eq:opt-lfixed}
	\end{equation}
	
	To finish the proof of inequality \eqref{eq:regrettoerror}, it then remains to lower bound \eqref{eq:opt-lfixed} as follows:
	\begin{equation}
		\frac{1}{6}\mathbb{E}\left[\max\{3\gamma, Z_D^2\}^2\right]-\gamma\mathbb{E}\left[\max\{3\gamma, Z_D^2\}\right]. 
	\end{equation}
	
	The analysis is divided into the following two cases. For simplicity, denote $\gamma_l$ to be the threshold such that $\mathbb{E}[\max\{Z_D^2, 3\gamma_l\}]=l$. 
	\begin{itemize}
		\item Case 1: $l\geq l(\gamma^\star)$, it can be easily verify that $\gamma_l \geq \gamma^\star$. Therefore, we have:
		\begin{align}
			&\frac{1}{6}\mathbb{E}\left[\max\{Z_D^2, 3\gamma_l\}^2\right]\nonumber\\
			=&\frac{1}{6}(3\gamma_l)^2\text{Pr}(Z_D^2\leq 3\gamma_l)+\frac{1}{6}\mathbb{E}\left[Z_D^4\cdot\mathbb{I}(Z_D^2\geq 3\gamma_l)\right]\nonumber\\
			=&\frac{1}{6}\left(\mathbb{E}\left[(3\gamma^\star)^2\mathbb{I}(Z_D^2\leq3\gamma^\star)\right]+\mathbb{E}\left[Z_D^4\mathbb{I}(Z_D^2>3\gamma^\star)\right]\right.\nonumber\\
			&\left.+\mathbb{E}[((3\gamma_l)^2-(3\gamma^\star)^2)\mathbb{I}(Z_D^2\leq3\gamma^\star)]\right.\nonumber\\
			&\left.+\mathbb{E}[((3\gamma_l)^2-Z_D^4)\mathbb{I}(3\gamma^\star\leq Z_D^2\leq 3\gamma_l)]\right)\nonumber\\
			\overset{(a)}{\geq}&q(\gamma^\star)+\frac{1}{6}\mathbb{E}\left[(3\gamma_l-3\gamma^\star)^2\mathbb{I}(Z_D^2\leq3\gamma^\star)\right]\nonumber\\&+\frac{1}{6}\mathbb{E}[3\gamma^\star(3\gamma_l-3\gamma^\star)\mathbb{I}(Z_D^2\leq 3\gamma^\star)]\nonumber\\&+\frac{1}{6}\mathbb{E}\left[6\gamma^\star(3\gamma_l-Z_D^2)\mathbb{I}(3\gamma^\star\leq Z_D^2\leq 3\gamma_l)\right]\nonumber\\
			\overset{(b)}{\geq}&\gamma^\star l(\gamma^\star)+\frac{1}{6}p_w(l(\gamma_l)-l(\gamma^\star))^2+\gamma^\star(l(\gamma_l)-l^\star)\nonumber\\
			=&\gamma^\star l+\frac{1}{6}p_w(l-l(\gamma^\star))^2,
		\end{align}
		where inequality $(a)$ is obtained because $(3\gamma_l)^2-(3\gamma^\star)^2=(3\gamma_l-3\gamma^\star)^2+2\times 3\gamma^\star(3\gamma_l-3\gamma^\star)$ and for  $Z_D^2=x$ that satisfies $3\gamma^\star\leq x\leq 3\gamma_l$, $(3\gamma_l)^2-x^2\geq 6\gamma^\star(3\gamma_l-x)$; inequality $(b)$ is because  $l(\gamma_l)-l(\gamma^\star)=\mathbb{E}\left[(3\gamma_l-3\gamma^\star)\mathbb{I}(Z_D^2\leq3\gamma^\star)\right]+\mathbb{E}\left[(3\gamma_l-Z_D^2)\mathbb{I}(3\gamma^\star\leq Z_D^2\leq 3\gamma_l)\right]$. 
		\item Case 2: $l\leq l(\gamma^\star)$, similarly, it can be verified that $3\gamma_l\leq3\gamma^\star$. As a result:
		\begin{align}
			&\frac{1}{6}\mathbb{E}\left[\max\{Z_D^2, 3\gamma_l\}^2\right]\nonumber\\
			=&\frac{1}{6}\mathbb{E}\left[(3\gamma_l)^2\mathbb{I}(Z_D^2\leq 3\gamma_l)\right]+\frac{1}{6}\mathbb{E}\left[Z_D^4\mathbb{I}(Z_D^2>3\gamma_l)\right]\nonumber\\
			=&\frac{1}{6}\left(\mathbb{E}\left[(3\gamma^\star)^2\mathbb{I}(Z_D^2\leq 3\gamma^\star)\right]+\mathbb{E}\left[Z_D^4\mathbb{I}(Z_D^2>3\gamma^\star)\right]\right.\nonumber\\&\left.-\mathbb{E}\left[((3\gamma^\star)^2-(3\gamma_l)^2)\mathbb{I}(Z_D^2\leq 3\gamma^\star)\right]\right.\nonumber\\
			&\left.-\mathbb{E}\left[(Z_D^2-(3\gamma_l)^2)\mathbb{I}(3\gamma_l\leq Z_D^2\leq3\gamma^\star)\right]\right)\nonumber\\
			\overset{(c)}{\geq}&q(\gamma^\star)+\frac{1}{6}\mathbb{E}\left[(3\gamma_l-3\gamma^\star)^2\mathbb{I}(Z_D^2\leq 3\gamma^\star)\right]\nonumber\\&+\frac{1}{6}\mathbb{E}\left[6\gamma^\star(3\gamma_l-3\gamma^\star)\mathbb{I}(Z_D^2\leq 3\gamma^\star)\right]\nonumber\\&-\frac{1}{6}\mathbb{E}\left[3\gamma^\star(3\gamma^\star-Z_D^2)\mathbb{I}(3\gamma_l\leq Z_D^2\leq 3\gamma^\star)\right]\nonumber\\
			=&\gamma^\star l(\gamma^\star)+\gamma^\star\mathbb{E}[(3\gamma_l-3\gamma^\star)\mathbb{I}(Z_D^2\leq 3\gamma^\star)]\nonumber\\&+\gamma^\star\mathbb{E}[(3\gamma^\star-Z_D^2)\mathbb{I}(3\gamma_l\leq Z_D^2\leq 3\gamma^\star)]\nonumber\\&+\frac{1}{6}p_w(l-l(\gamma^\star))^2\nonumber\\
			=&\gamma^\star l+\frac{1}{6}p_w(l-l(\gamma^\star))^2,
		\end{align}
		where inequality $(c)$ is obtained similarly as inequality $(a)$ and $(b)$. 
	\end{itemize}
\end{proof}		
		\section{Proof of Corollary~\ref{coro:squareub}}\label{pf:corosquareub}
		
		\begin{proof}
			The first conclusion follows directly from Lemma 1. 
			%
			
			The conditional expectation of $(X_{S_{k+1}}-X_{S_k})^4$ can be upper bounded by:
			\begin{align}
				&\mathbb{E}_k[(X_{S_{k+1}}-X_{S_k})^8]=\mathbb{E}_k\left[\max\{3\gamma_k, Z_D^2\}^4\right]\nonumber\\
				\leq&\left((3\gamma_k)^4+105\mathbb{E}[D^4]\right)=(3\gamma_k)^4+105B.
			\end{align}
			
			Finally, we can upper bounded the second order moment of $E_k$ as follows:
			\begin{align}
				&\mathbb{E}_k\left[E_k^2\right]\nonumber\\
				=&\mathbb{E}_k\left[\left(\int_{S_k}^{S_k+D_k}(X_t-X_{S_{k-1}})^2\text{d}t\right.\right.\nonumber\\
				&\left.\left.+\int_{S_k+D_k}^{S_{k+1}}(X_t-X_{S_k})^2\text{d}t\right)^2\right]\nonumber\\
				=&\mathbb{E}_k\left[\left(\int_{S_k}^{S_k+D_k}(X_t-X_{S_k}+X_{S_k}-X_{S_{k-1}})^2\text{d}t\right.\right.\nonumber\\
				&\left.\left.+\int_{S_k+D_k}^{S_{k+1}}(X_t-X_{S_k})^2\text{d}t\right)^2\right]\nonumber\\
				=&\mathbb{E}_k\Big[\Big((X_{S_k}-X_{S_{k-1}})^2D_k\nonumber\\
				&+2(X_{S_k}-X_{S_{k-1}})\cdot\int_{S_k}^{S_k+D_k}(X_t-X_{S_k})\text{d}t\nonumber\\
				&+\int_{S_k}^{S_{k+1}}(X_t-X_{S_k})^2\text{d}t\Big)^2\Big]\nonumber\\
				\overset{(c)}{\leq}& 3\mathbb{E}_k\left[(X_{S_k}-X_{S_{k-1}})^4D_k^2\right]\nonumber\\
				&+12\mathbb{E}_k\left[(X_{S_k}-X_{S_{k-1}})^2\left(\int_{S_k}^{S_k+D_k}(X_t-X_{S_k})\text{d}t\right)^2\right]\nonumber\\
				&+3\mathbb{E}_k\left[\left(\int_{S_k}^{S_{k+1}}(X_t-X_{S_k})^2\text{d}t\right)^2\right],\label{eq:ek-1}
			\end{align}
			where inequality $(c)$ is from Cauchy-Schwartz $\mathbb{E}[(a+b+c)^2]\leq3\mathbb{E}[a^2+b^2+c^2]$.
			
			Since the transmission delay $D_k$ is independent of $X_{S_k}-X_{S_{k-1}}$, the first term on the RHS of inequality~\eqref{eq:ek-1} can be upper bounded by:
			\begin{align}
				&\mathbb{E}_k\left[(X_{S_k}-X_{S_{k-1}})^4D_k^2\right]=(X_{S_k}-X_{S_{k-1}})^4\mathbb{E}[D^2]\nonumber\\
				&\leq(X_{S_k}-X_{S_{k-1}})^4\sqrt{B}. 
			\end{align}
			
			To upper bound the second and third term on the RHS of inequality~\eqref{eq:ek-1}, we introduce the following Lemma, whose proof is provided in Appendix~\ref{pf:square-ub}
			\begin{lemma}\label{lemma:square-ub}
				Recall that $Z_t$ is a wiener process staring from time 0 and let $l_\gamma:=\inf\{t\geq D||Z_t|\geq\sqrt{3\gamma}\}$ be the frame length when threshold $\gamma$ is used. When $\mathbb{E}[D^4]\leq B$, we have the following results:
				\begin{align}
					&\mathbb{E}\left[\left(\int_{t=0}^{l_\gamma}Z_t\text{d}t\right)^2\right]\leq\nonumber\\
					&\hspace{1cm}\left(\frac{277}{31}(3\gamma)^4+B\right)\cdot\left((3\gamma)^2+\left(\frac{4}{3}\right)^4\cdot 3\sqrt{B}\right)\nonumber\\
					&\hspace{1cm}=:C_1(\gamma, B),\label{eq:inte-1}\\
					&\mathbb{E}\left[\left(\int_{t=0}^{l_\gamma}Z_t^2\text{d}t\right)^2\right]\leq\nonumber\\
					&\hspace{1cm}\left(\frac{277}{31}(3\gamma)^4+B\right)\cdot\left((3\gamma)^4+\left(\frac{8}{7}\right)^8\cdot 105B\right)\nonumber\\
					&\hspace{1cm}=:C_2(\gamma, B).\label{eq:inte-2}
				\end{align}
			\end{lemma}
			
			Plugging inequality \eqref{eq:inte-1} and \eqref{eq:inte-2} into \eqref{eq:ek-1}, we can upper bound $\mathbb{E}[E_k^2|\gamma_k, X_{S_k}-X_{S_{k-1}}]$ by:
			\begin{align}
				\mathbb{E}_k[E_k^2]=&3(X_{S_k}-X_{S_{k-1}})^4\sqrt{B}+12C_1(\gamma_k, B)(X_{S_k}-X_{S_{k-1}})^2\nonumber\\
				&+3C_2(\gamma, B). 
			\end{align}

		\end{proof}
		
		\section{Proof of Lemma~\ref{lemma:square-ub}}\label{pf:square-ub}
		Recall that $l_\gamma=\inf\{t\geq D||\sqrt{Z}_t|\geq \sqrt{3\gamma}\}$ is a stopping time of the Wiener process $Z_t$ starting from time $0$. Then for any time $\tau\geq 0$ we have:
		\begin{align}
			&\mathbb{E}\left[\left(\int_{t=0}^{l_\gamma} Z_t^p\text{d}t\right)^2\right]\nonumber\\
			\leq&\mathbb{E}\left[\left(\int_{0}^{l_\gamma}\left(\sup_{0\leq t'\leq l_{\gamma}}|Z_{t'}|\right)^p\text{d}t\right)^2\right]\nonumber\\
			=&\mathbb{E}\left[l_\gamma^2\cdot\left(\sup_{0\leq t'\leq l_\gamma}|Z_{t'}|^{2p}\right)\right]\nonumber\\
			\overset{(a)}{\leq}&\sqrt{\mathbb{E}\left[l_\gamma^4\right]\cdot\mathbb{E}\left[\sup_{0\leq t'\leq l_\gamma}|Z_{t'}|^{4p}\right]},\label{eq:int-ub-1}
		\end{align}
		where inequality $(a)$ from the Cauchy-Schwartz inequality.
		
		From Lemma \ref{lemma:4order}, $\mathbb{E}[l_\gamma^4]$ can be bounded as follows:
		\begin{align}
			\mathbb{E}[l_\gamma^4]\leq 256\left((3\gamma)^4+105B\right).\label{eq:l4}
		\end{align}
		
		Next, we prove $\mathbb{E}[\sup_{0\leq t'\leq \tau}|Z_{t'}|^{2p}]$ is bounded. Recall that the stopping rule is obtained by:
		\begin{equation}
			l_\gamma=\inf\{t\geq D||Z_t|\geq \sqrt{3\gamma}\}.
		\end{equation}
		
		Then $\mathbb{E}\left[\sup_{0\leq t'\leq l_\gamma}|Z_{t'}|^{4p}\right]$ can be upper bounded by:
		\begin{align}
			&\mathbb{E}\left[\sup_{0\leq t'\leq l_\gamma}|Z_{t'}|^{4p}\right]\nonumber\\
			=&\mathbb{E}\left[\left(\sup_{0\leq t'\leq l_\gamma}|Z_{t'}|^{4p}\right)\cdot\mathbb{I}(l_\gamma>D)\right]\nonumber\\
			&+\mathbb{E}\left[\left(\sup_{0\leq t'\leq l_\gamma}|Z_{t'}|^{4p}\right)\cdot\mathbb{I}(l_\gamma\leq D)\right]\nonumber\\
			\overset{(b)}{\leq}&(3\gamma)^{2p}+\mathbb{E}\left[\sup_{0\leq t'\leq D}|Z_{t'}|^{4p}\right],\label{eq:lemma-last-1}
		\end{align}
		where inequality $(b)$ is because if $l_\gamma>D$, then we have $|Z_t|\leq\sqrt{3\gamma}, \forall t\in[D, l_\gamma)$ and therefore $\sup_{0\leq t'\leq l_\gamma}|Z_{t'}|^{4p}\leq(3\gamma)^{2p}+\sup_{0\leq t'\leq D}|Z_{t'}|^{4p}$. For each $D<\infty$, since the Wiener process $Z_t$ is a martingale and $D$ is a stopping time, for each $d<\infty$, we can upper bound $\mathbb{E}\left[\sup_{0\leq t'\leq D}|Z_{t'}|^{4p}\right]$ as follows:
		\begin{align}
			&\mathbb{E}\left[\sup_{0\leq t'\leq d}|Z_{t'}|^{4p}\right]\overset{(c)}{\leq} \left(\frac{4p}{4p-1}\right)^{4p}\mathbb{E}\left[Z_{d}^{4p}\right]\nonumber\\
			\overset{(d)}{=}&\left(\frac{4p}{4p-1}\right)^{4p}\frac{(4p)!}{2^{2p}(2p)!}d^{2p}. \label{eq:lemma-last}
		\end{align}
		where inequality $(c)$ is because of the Doob's maximal inequality~\cite[p.54, Theorem 1.7]{doobmaximal} and equality $(d)$ is because $Z_d$ follows a Guassian distribution. 
		
		When the transmission delay $D$ is fourth order bounded, for $p=1$ and $2$, plugging $\mathbb{E}[D^4]\leq B$ and $\mathbb{E}[D^2]\leq\sqrt{B}$ into inequality \eqref{eq:lemma-last} and \eqref{eq:lemma-last-1}, we have:
		\begin{align}
			\mathbb{E}\left[\sup_{0\leq t'\leq l_\gamma}|Z_{t'}|^{4}\right]\leq& (3\gamma)^{2}+\left(\frac{4}{3}\right)^4\cdot 3\mathbb{E}[D^2]\nonumber\\
			\leq&(3\gamma)^{2}+\left(\frac{4}{3}\right)^4\cdot3\sqrt{B},\label{eq:p2}\\
			\mathbb{E}\left[\sup_{0\leq t'\leq l_\gamma}|Z_{t'}|^{8}\right]\leq& (3\gamma)^{4}+\left(\frac{8}{7}\right)^8\cdot105\mathbb{E}[D^4]\nonumber\\
			\leq&(3\gamma)^{4}+\left(\frac{8}{7}\right)^8\cdot105B. \label{eq:p1}
		\end{align}
		
		Plugging inequality \eqref{eq:p1}, \eqref{eq:p2} and \eqref{eq:l4} into inequality \eqref{eq:int-ub-1}, we have:
		\begin{align}
			\mathbb{E}\left[\left(\int_{t=0}^{l_\gamma}Z_t\text{d}t\right)^2\right]\leq& 256\left((3\gamma)^4+105B\right)\nonumber\\
			&\times\left((3\gamma)^2+\left(\frac{4}{3}\right)^4\cdot 3\sqrt{B}\right),\\
			\mathbb{E}\left[\left(\int_{t=0}^{l_\gamma}Z_t^2\text{d}t\right)^2\right]\leq& 256\left((3\gamma)^4+105B\right)\nonumber\\
			&\times\left((3\gamma)^4+\left(\frac{8}{7}\right)^8\cdot 105B\right).
		\end{align}

%




\end{document}